\documentclass[journal]{IEEEtran}

\usepackage{setspace, dsfont}
\usepackage{amssymb, amsmath, mathrsfs, amsthm, stmaryrd, upgreek, mathtools}
\usepackage{pgfplots, caption, tikz, subfigure, graphicx}
\usepackage{url, hyperref}
\usepackage[compress]{cite}
\usepackage{bbm}
\usetikzlibrary{spy}
\usepgflibrary{arrows} 
\usepackage{varwidth}
\def\tablename{Table}

\usepackage{lipsum}

\usepackage{graphicx,wrapfig,tikz}
\usetikzlibrary{calc,hobby} 

\usepackage{scalerel,stackengine}
\stackMath
\newcommand\reallywidehat[1]{%
\savestack{\tmpbox}{\stretchto{%
  \scaleto{%
    \scalerel*[\widthof{\ensuremath{#1}}]{\kern-.6pt\bigwedge\kern-.6pt}%
    {\rule[-\textheight/2]{1ex}{\textheight}}
  }{\textheight}%
}{0.5ex}}%
\stackon[1pt]{#1}{\tmpbox}%
}

\makeatletter
\def\@IEEEsectpunct{.\ \,}
\makeatother

\usepackage[font=footnotesize]{caption}

\newcommand{\Rd}{\mathbb{R}^d}

\newcommand{\RR}{\mathbb{R}}

\newcommand{\sig}{\text{sig}}

\newcommand{\QQ}{\mathcal{Q}}

\newcommand{\la}{\lambda}

\newcommand\mydots{\ifmmode\ldots\else\makebox[1em][c]{.\hfil.\hfil.}\fi}
\newcommand{\VV}{\Vert}

\DeclareMathOperator{\supp}{supp}
\DeclareMathOperator{\Real}{Re}
\DeclareMathOperator{\Imag}{Im}

\def\ba#1\ea{\begin{align*}#1\end{align*}}	
\def\ban#1\ean{\begin{align}#1\end{align}}	
\def\bac#1\eac{\vspace{\abovedisplayskip}{\par\centering$\begin{aligned}#1\end{aligned}$\par}\addvspace{\belowdisplayskip}}	

\newtheorem{theorem}{Theorem}

\newtheorem{lemma}{Lemma}

\newtheorem{definition}{Definition}
\newtheorem{proposition}{Proposition}
\newtheorem{corollary}{Corollary}

\newtheorem{remark}{Remark}
\newtheorem*{SL}{Schur's Lemma}

\usepackage{endnotes}

\begin{document}
\title{A Mathematical Theory of Deep Convolutional\\ Neural Networks for Feature Extraction}

\author{Thomas~Wiatowski~and~Helmut~B\"olcskei,~\IEEEmembership{Fellow,~IEEE}
\thanks{The authors are with the Department of Information Technology and Electrical Engineering, ETH Zurich,  8092 Zurich, Switzerland. Email:~\{withomas,~boelcskei\}@nari.ee.ethz.ch}
\thanks{The material in this paper was presented in part at the 2015 IEEE International Symposium on Information Theory (ISIT), Hong Kong, China. }
\thanks{Copyright (c) 2017 IEEE. Personal use of this material is permitted.  However, permission to use this material for any other purposes must be obtained from the IEEE by sending a request to pubs-permissions@ieee.org.}
}

\maketitle

\begin{abstract}
Deep convolutional neural networks have led to breakthrough results in numerous practical  machine learning tasks such as classification of images in the ImageNet data set, control-policy-learning to play Atari games or the board game Go, and image captioning. Many of these applications first perform feature extraction and then feed the results thereof into a trainable classifier. The mathematical analysis of deep convolutional neural networks for feature extraction was initiated by Mallat, 2012. Specifically, Mallat considered so-called scattering networks based on a wavelet transform followed by the modulus non-linearity in each network layer, and proved translation invariance (asymptotically in the wavelet scale parameter) and deformation stability of the corresponding feature extractor. This paper complements Mallat's results by developing  a theory  that encompasses  general convolutional transforms, or in more technical parlance, general semi-discrete frames (including Weyl-Heisenberg filters, curvelets, shearlets, ridgelets, wavelets, and learned filters), ge\-neral Lipschitz-continuous non-linearities (e.g., rectified linear units, shifted logistic sigmoids, hyperbolic tangents, and modulus functions), and general Lipschitz-continuous pooling operators emulating, e.g., sub-sampling and averaging. In addition, all of these elements can be different in different network layers. For the resulting feature extractor we prove a translation invariance result  of vertical nature in the sense of the features becoming progressively more translation-invariant with increasing network depth, and we establish deformation sensitivity bounds that apply to signal classes such as, e.g., band-limited functions, cartoon functions, and Lipschitz functions. 
\end{abstract}


%
\IEEEpeerreviewmaketitle

\section{Introduction}
%
%
%
%
\IEEEPARstart{A}{}central task in machine learning  is feature extraction \cite{Bengio,Bishop,Duda} as, e.g., in the context of  handwritten digit classification \cite{MNIST}. The features to be extracted in this case correspond, for example, to the edges of the digits. The idea behind feature extraction is that feeding characteristic features of the signals---rather than the signals themselves---to a trainable  classifier (such as, e.g., a  support vector machine (SVM) \cite{SVM}) improves classification performance. Specifically, non-linear feature extractors (obtained, e.g., through the use of a so-called kernel in the context of SVMs) can map input signal space dichotomies that are not linearly separable into linearly separable feature space dichotomies \cite{Bishop}. Sticking to the example of handwritten digit classification, we would, moreover, want the feature extractor to be invariant to the digits' spatial location within the image, which leads to the requirement of translation invariance. In addition,  it is desirable that the  feature extractor  be robust with respect to (w.r.t.) handwriting styles. This can be accomplished by demanding limited sensitivity of the features to certain non-linear deformations of the signals to be classified. 

Spectacular success in  practical machine learning tasks has been reported for  feature extractors generated by so-called deep convolutional neural networks (DCNNs)\cite{Bengio,LeCunNIPS89,Rumelhart,LeCunProc,LeCun,Nature,he2015deep,Goodfellow-et-al-2016}. These networks are composed of multiple layers, each of which computes convolutional transforms, followed by non-linearities and pooling\footnote{In the literature ``pooling'' broadly refers to some form of combining ``nearby'' values of a signal (e.g., through averaging) or picking one representative value (e.g, through maximization or sub-sampling).} operators. While DCNNs can be used to perform classification (or other machine learning tasks such as regression) directly \cite{Bengio,LeCunNIPS89,LeCunProc,LeCun,Nature}, typically based on the output of the last network layer, they can also act as stand-alone feature extractors \cite{Huang,Jarrett,hierachies,Poultney,Pinto,Serre,GaborLowe} with the resulting features fed into a classifier  such as a SVM. The present paper pertains to the latter philosophy.

The mathematical ana\-lysis of feature extractors generated by DCNNs was pioneered by Mallat in \cite{MallatS}. Mallat's theory applies to so-called scattering networks, where signals are propagated through layers that compute a semi-discrete wavelet transform (i.e., convolutions with filters that are obtained from a mother wavelet through scaling and rotation operations), followed by the modulus non-linearity,  without subsequent pooling.  The resulting feature extractor is shown to be translation-invariant (asymptotically in the scale parameter of the underlying wavelet transform) and stable w.r.t. certain non-linear deformations. Moreover, Mallat's scattering networks lead to state-of-the-art results in various classification tasks \cite{Bruna,Sifre,Anden}. 

\paragraph*{Contributions} 
DCNN-based feature extractors that were found to work well in practice employ a wide range of i) filters, namely pre-specified structured filters  such as wavelets \cite{Jarrett,Pinto,Serre,GaborLowe}, pre-specified unstructured filters such as random filters \cite{Jarrett,hierachies}, and filters that are learned in a supervised  \cite{Huang,Jarrett} or an unsupervised \cite{Jarrett,hierachies,Poultney} fashion, ii) non-linearities beyond the modulus function \cite{Jarrett,MallatS,GaborLowe}, namely hyperbolic tangents \cite{Huang,Jarrett,hierachies}, rectified linear units \cite{Glorot,Nair}, and logistic sigmoids \cite{Mohamed,Glorot2}, and iii) pooling operators, namely sub-sampling \cite{Pinto}, average pooling \cite{Huang,Jarrett}, and max-pooling  \cite{Jarrett,hierachies,Serre,GaborLowe}. In addition, the filters, non-linearities, and pooling operators  can be different in different network layers \cite{Goodfellow-et-al-2016}. The goal of this paper is to develop a mathematical theory that encompasses all these elements (apart from max-pooling) in full generality. 

Convolutional transforms as employed in DCNNs can be interpreted as semi-discrete signal transforms    \cite{MallatW,MallatEdges,UnserWavelets,Vandergheynst,CandesDonoho2,Grohs_alpha,Shearlets,Grohs_transport} (i.e., convolutional transforms with filters that are countably parametrized). Corresponding prominent representatives are curvelet \cite{CandesNewTight,CandesDonoho2,Grohs_alpha} and shearlet \cite{OriginShearlets,Shearlets} transforms, both of which are known to be highly effective in extracting features characterized by curved edges in images. Our theory allows for general semi-discrete signal transforms, general Lipschitz-continuous non-linearities (e.g., rectified linear units, shifted logistic sigmoids, hyperbolic tangents, and modulus functions), and incorporates continuous-time Lipschitz pooling operators that emulate discrete-time sub-sampling and averaging. Finally, different network layers may be equipped with different convolutional transforms, different (Lipschitz-continuous) non-linearities, and different (Lipschitz-continuous) pooling operators. 

Regarding translation invariance, it was argued, e.g., in \cite{Huang,Jarrett,hierachies,Serre,GaborLowe},  that in practice invariance of the  features is crucially governed by   network depth and by the presence of pooling operators (such as, e.g., sub-sampling \cite{Pinto}, average-pooling \cite{Huang,Jarrett}, or max-pooling  \cite{Jarrett,hierachies,Serre,GaborLowe}).  We show that the general feature extractor considered in this paper, indeed, exhibits such a \textit{vertical} translation invariance and that pooling plays a crucial role in achieving it. Specifically, we prove that the depth of the network determines the extent to which the extracted features are translation-invariant. We also show that pooling is necessary to obtain vertical translation invariance as otherwise the features remain fully translation-covariant irrespective of network depth. We furthermore establish a deformation sensitivity bound valid for signal classes  such as, e.g., band-limited functions, cartoon functions \cite{grohs_wiatowski}, and Lipschitz functions \cite{grohs_wiatowski}. This bound shows that small non-linear deformations of the input signal lead to small changes in the corresponding feature vector. 

In terms of mathematical techniques, we draw heavily from continuous frame theory \cite{Antoine,Kaiser}. We develop a proof machinery that is completely detached from the structures\footnote{Structure here refers to the structural relationship between the convolution kernels in a given layer, e.g., scaling and rotation operations in the case of the wavelet transform.} of the semi-discrete transforms and the specific form of the Lipschitz non-linearities and Lipschitz pooling operators. The proof of our deformation sensitivity bound  is based on two key elements, namely  Lipschitz continuity of  the feature extractor and a deformation sensitivity bound for the signal class under consideration, namely band-limited functions (as established in the present paper) or  cartoon functions and Lipschitz functions as shown in \cite{grohs_wiatowski}. This ``decoupling'' approach has important practical ramifications as it shows that whenever we have  deformation sensitivity bounds for a signal class, we automatically get deformation sensitivity bounds for the DCNN feature extractor operating on that signal class. Our results hence establish that vertical translation invariance and limited sensitivity to deformations---for signal classes with inherent deformation insensitivity---are guaran\-teed by the network structure per se rather than the specific convolution kernels, non-linearities, and pooling operators.


\paragraph*{Notation}
The complex conjugate of $z \in \mathbb{C}$ is denoted by $\overline{z}$. We write $\Real(z)$ for the real, and $\Imag(z)$ for the imaginary part of $z \in \mathbb{C}$. The Euclidean inner product of $x,y \in \mathbb{C}^d$ is $\langle x, y \rangle:=\sum_{i=1}^{d}x_i \overline{y_i}$, with associated norm $|x|:=\sqrt{\langle x, x \rangle}$. We denote the identity matrix by  $E\in\mathbb{R}^{d\times d}$. For the matrix $M\in \mathbb{R}^{d\times d}$, $M_{i,j}$ designates the entry in its $i$-th row and $j$-th column, and for a tensor $T\in \mathbb{R}^{d\times d \times d}$, $T_{i,j,k}$ refers to its $(i,j,k)$-th component. The supremum norm of the matrix $M\in \mathbb{R}^{d\times d}$ is defined as $|M|_{\infty}:=\sup_{i,j}|M_{i,j}|$, and the supremum norm of the tensor $T\in \mathbb{R}^{d\times d \times d}$ is $|T|_{\infty}:=\sup_{i,j,k}|T_{i,j,k}|$. We write $B_r(x)\subseteq \Rd$ for the open ball of radius $r>0$ centered at $x\in \Rd$. $O(d)$ stands for the orthogonal group of dimension $d \in \mathbb{N}$, and $SO(d)$ for the special orthogonal group.

 For a Lebesgue-measurable function $f:\Rd \to \mathbb{C}$, we write $\int_{\Rd} f(x) \mathrm dx$ for the integral of $f$ w.r.t. Lebesgue measure $\mu_L$. For $p \in [1,\infty)$,  $L^p(\Rd)$ stands for the space of Lebesgue-measurable functions $f:\Rd \to \mathbb{C}$ satisfying $\VV f \VV_p:= (\int_{\Rd}|f(x)|^p \mathrm dx)^{1/p}<\infty.$ $L^\infty(\Rd)$ denotes the space of Lebesgue-measurable functions $f:\Rd \to \mathbb{C}$ such that  $\VV f \VV_\infty:=\inf \{\alpha>0 \ | \ |f(x)| \leq \alpha \text{ for a.e.}\footnote{Throughout ``a.e.''  is w.r.t. Lebesgue measure.}\  x \in \Rd \}<\infty$.  For $f,g \in L^2(\Rd)$ we set $\langle f,g \rangle :=\int_{\Rd}f(x)\overline{g(x)}\mathrm dx$. For $R>0$, the space of $R$-band-limited functions is denoted as $
L^2_R(\Rd):=\{ f \in L^2(\Rd) \ | \ \supp(\widehat{f}) \subseteq B_{R}(0)\}.$ For a countable set $\QQ$, $(L^2(\Rd))^\QQ$ stands for the space of sets $s:=\{ s_q \}_{q\in \QQ}$, $s_q \in L^2(\Rd)$, for all $q \in \QQ$, satisfying $||| s |||:=(\sum_{q \in \QQ} \VV s_q\VV_2^2)^{1/2}<\infty$. 

 $\text{Id}:L^p(\Rd) \to L^p(\Rd)$ denotes  the identity operator on $L^p(\Rd)$. The tensor product of functions $f,g:\Rd\to \mathbb{C}$ is $(f\otimes g)(x,y):=f(x)g(y)$, $ (x,y)\in \RR^{d}\times \RR^{d}$. 
 The operator norm of the bounded linear operator $A:L^p(\Rd) \to L^q(\Rd)$ is defined as  $\VV A\VV_{p,q}:=\sup_{\VV f \VV_p=1}\VV Af \VV_q$. We denote the Fourier transform of $f \in L^1(\Rd)$ by $\widehat{f}(\omega):=\int_{\Rd}f(x)e^{-2\pi i \langle x,  \omega \rangle }\mathrm dx$ and extend it in the usual way to $L^2(\Rd)$ \cite[Theorem 7.9]{Rudin}. The convolution of $f\in L^2(\Rd)$ and $g\in L^1(\Rd)$ is $(f\ast g)(y):=\int_{\Rd}f(x)g(y-x)\mathrm dx$. We write $(T_tf)(x):=f(x-t)$, $t \in \Rd$, for the translation operator, and $(M_\omega f)(x):=e^{2\pi i \langle x , \omega\rangle }f(x)$, $\omega \in \Rd$, for the modulation operator. Involution is defined by $(If)(x):=\overline{f(-x)}$. 
 
 A multi-index $\alpha=(\alpha_1,\dots,\alpha_d) \in \mathbb{N}_0^d$ is an ordered $d$-tuple of non-negative integers $\alpha_i \in \mathbb{N}_0$. For a multi-index $\alpha \in \mathbb{N}_0^d$, $D^\alpha$ denotes the differential operator $D^\alpha:=(\partial/\partial x_1)^{\alpha_{1}}\dots (\partial/\partial x_d)^{\alpha_{d}}$, with order $|\alpha|:=\sum_{i=1}^d \alpha_i$. If $|\alpha|=0$, $D^\alpha f:=f$, for $f:\Rd \to \mathbb{C}$. The space of functions $f:\Rd \to \mathbb{C}$ whose derivatives $D^\alpha f$ of order at most $N\in \mathbb{N}_0$ are continuous is designated by $C^N(\Rd,\mathbb{C})$, and the space of  infinitely differentiable functions is   $C^\infty(\Rd,\mathbb{C})$.  $S(\Rd,\mathbb{C})$ stands for the Schwartz space, i.e., the space of  functions $f \in C^\infty(\Rd,\mathbb{C})$ whose derivatives $D^\alpha f$ along with the function itself are rapidly decaying \cite[Section 7.3]{Rudin} in the sense of $\sup_{|\alpha|\leq N} \sup_{x\in \Rd}(1+|x|^2)^N|(D^\alpha f)(x)|<\infty$, for all $N\in \mathbb{N}_0$. We denote the gradient of a function $f:\Rd \to \mathbb{C}$ as $\nabla f$. The space of continuous mappings $v:\RR^p \to \RR^q$ is $C(\RR^p,\RR^q)$, and for $k,p,q \in \mathbb{N}$, the space of $k$-times continuously differentiable mappings $v:\RR^p \to \RR^q$ is written as $C^k(\RR^p,\RR^q)$. For a mapping $v:\Rd \to \Rd$,  we let $Dv$ be its Jacobian matrix, and $D^2v$  its Jacobian tensor, with associated norms $\VV v \VV_\infty:=\sup_{x\in\Rd} |v(x)|$, $\VV D v \VV_\infty:=\sup_{x\in \Rd}|(Dv)(x)|_{\infty}$, and $\VV D^2 v \VV_\infty:=\sup_{x\in \Rd}|(D^2v)(x)|_{\infty}$.



\section{Scattering networks}\label{architecture}
We set the stage by  reviewing scattering networks as introduced in \cite{MallatS}, the basis of which is a multi-layer architecture that involves a wavelet transform followed by the modulus non-linearity, without subsequent pooling. Specifically, \cite[Definition 2.4]{MallatS} defines the feature vector $\Phi_W(f)$ of the signal $f\in L^2(\Rd)$ as the set\footnote{We emphasize that the  feature vector $\Phi_W(f)$ is a union of the sets of feature vectors $\Phi_W^n(f)$.}
\begin{equation}\label{help11}
\Phi_W(f):=\bigcup_{n=0}^\infty \Phi_W^n(f),
\end{equation}
where $\Phi_W^0(f):=\{f\ast \psi_{(-J,0)} \}$, and
\begin{align*} 
&\Phi_W^n(f)\\
&:=\bigg\{ \big(U\big[\underbrace{\lambda^{^{(j)}}\hspace{-0.1cm},\dots,\lambda^{^{(p)}}}_{n \ \text{indices}}\big]f\big) \ast \psi_{(-J,0)}\bigg\}_{\lambda^{^{(j)}}\hspace{-0.1cm},\dots,\lambda^{^{(p)}}\in\Lambda_{\text{W}}\backslash \{(-J,0)\}},
\end{align*}
for all $n\in \mathbb{N}$,
with
$$U\big[\lambda^{^{(j)}}\hspace{-0.1cm},\dots,\lambda^{^{(p)}}\big]f:=\underbrace{\big| \cdots \big| \ |f \ast \psi_{\lambda^{^{(j)}}}| \ast  \psi_{\lambda^{^{(k)}}}\big|\cdots \   \ast \psi_{\lambda^{^{(p)}}}\big|}_{n-\text{fold convolution followed by modulus}}.$$ Here, the index set $\Lambda_{\text{W}}:=\big\{ (-J,0)\big\}\cup\big\{ (j,k) \ | \ j \in \mathbb{Z} \text{ with }j>-J, \ k \in \{ 0,\dots,K-1\} \big\}$ contains pairs of scales $j$ and directions $k$ (in fact, $k$ is the index of the direction described by the rotation matrix $r_k$), and 
\begin{equation}\label{dir_wav}
\psi_{\lambda}(x):=2^{dj}\psi(2^jr_k^{-1}x),
\end{equation}
where $ \lambda=(j,k) \in\Lambda_{\text{W}}\backslash \{(-J,0)\}$ are directional wavelets \cite{Lee,AntoineVander,MallatW} with (complex-valued) mother wavelet $\psi\in L^1(\Rd)\cap L^2(\Rd)$. The $r_k$, $k\in \{ 0,\dots,K-1\}$, are elements of a finite rotation group $G$ (if $d$ is even, $G$ is a subgroup of $SO(d)$; if $d$ is odd, $G$ is a subgroup of $O(d)$). The index $(-J,0) \in \Lambda_{\text{W}}$ is associated with the low-pass filter $\psi_{(-J,0)}\in L^1(\Rd)\cap L^2(\Rd) $, and $J\in \mathbb{Z}$ corresponds to the coarsest scale resolved by the directional wavelets \eqref{dir_wav}.

The family of functions $\{ \psi_{\lambda} \}_{\la \in \Lambda_{\text{W}}}$ is taken to form a semi-discrete Parseval frame $$\Psi_{\Lambda_{\text{W}}}:=\{T_bI\psi_\la \}_{b\in \Rd, \la \in \Lambda_{\text{W}}},$$ for $L^2(\Rd)$  \cite{Antoine,Kaiser,MallatW} and hence satisfies
\begin{equation*}\label{LPW}
\sum_{\la \in \Lambda_{\text{W}}}\int_{\Rd} |\langle f, T_b I \psi_\lambda \rangle |^2 \mathrm db = \sum_{\la \in \Lambda_{\text{W}}} \VV f\ast\psi_{\lambda} \VV_2^2=\VV f \VV_2^2,
\end{equation*}
for all $f \in L^2(\Rd)$, where $\langle f, T_bI\psi_\lambda \rangle=(f\ast \psi_\lambda)(b),$  $(\lambda,b)\in \Lambda_{\text{W}}\times\Rd,$ are the underlying frame coefficients. Note that for given $\lambda \in \Lambda_{\text{W}}$, we actually have a continuum of frame coefficients as the translation parameter $b\in \Rd$ is left unsampled. We refer to Figure \ref{wavfig} for a frequency-domain illustration of a semi-discrete directional wavelet frame. In Appendix \ref{sec:sdf}, we give a brief review of the general theory of semi-discrete frames, and in Appendices \ref{sec:sdf0} and \ref{sec:sdf2} we collect structured example frames in $1$-D and $2$-D, respectively.

  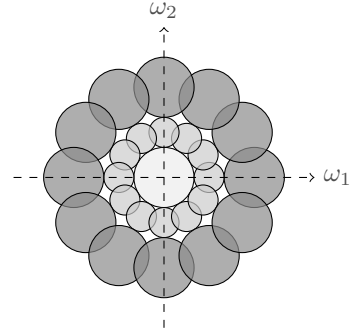
\begin{figure}[t]
\begin{center}
\begin{tikzpicture}


	\begin{scope}[scale=.8]


	\filldraw[fill=black!5!white, draw=black]
   		  (0,0) circle (.5cm);
	\pgfsetfillopacity{.7}
	\filldraw[fill=black!20!white, draw=black]
   		(.75,0) circle (.25cm);
	\filldraw[fill=black!20!white, draw=black]
   		(0.6495,.375) circle (.25cm);
	\filldraw[fill=black!20!white, draw=black]
   		(0.375,0.6495) circle (.25cm);
	\filldraw[fill=black!20!white, draw=black]
   		(0,.75) circle (.25cm);	
	\filldraw[fill=black!20!white, draw=black]
   		(-0.375,0.6495) circle (.25cm);
	\filldraw[fill=black!20!white, draw=black]
   		(-0.6495,0.375) circle (.25cm);
	\filldraw[fill=black!20!white, draw=black]
   		(-.75,0) circle (.25cm);	
	\filldraw[fill=black!20!white, draw=black]
   		(-0.6495,-.375) circle (.25cm);
	\filldraw[fill=black!20!white, draw=black]
   		(-0.375,-0.6495) circle (.25cm);
	\filldraw[fill=black!20!white, draw=black]
  		(0,-.75) circle (.25cm);	
	\filldraw[fill=black!20!white, draw=black]
   		(0.375,-0.6495) circle (.25cm);
	\filldraw[fill=black!20!white, draw=black]
   		(0.6495,-.375) circle (.25cm);
	\filldraw[fill=black!45!white, draw=black]
   		(1.5,0) circle (.5cm);
	\filldraw[fill=black!45!white, draw=black]
   		(1.299,0.75) circle (.5cm);
	\filldraw[fill=black!45!white, draw=black]
   		(.75,1.299) circle (.5cm);
	\filldraw[fill=black!45!white, draw=black]
   		(0,1.5) circle (.5cm);
	\filldraw[fill=black!45!white, draw=black]
   		(-.75,1.2990) circle (.5cm);
	\filldraw[fill=black!45!white, draw=black]
   		(-1.299,.75) circle (.5cm);
	\filldraw[fill=black!45!white, draw=black]
   		(-1.5,0) circle (.5cm);
	\filldraw[fill=black!45!white, draw=black]
   		(-1.299,-.75) circle (.5cm);	
	\filldraw[fill=black!45!white, draw=black]
   		(-.75,-1.299) circle (.5cm);	
	\filldraw[fill=black!45!white, draw=black]
   		(0,-1.5) circle (.5cm);	
	\filldraw[fill=black!45!white, draw=black]
   		(.75,-1.299) circle (.5cm);
	\filldraw[fill=black!45!white, draw=black]
   		(1.299,-.75) circle (.5cm);
	\draw[->,dashed] (-2.5,0) -- (2.5,0) node[right] {$\omega_1$} ;
	\draw[->,dashed] (0,-2.5) -- (0,2.5)node[above] {$\omega_2$};

		\end{scope}
\end{tikzpicture}
\end{center}
\caption{Partitioning of the frequency plane $\RR^2$ induced by a semi-discrete directional wavelet frame with $K=12$ directions.}
\label{wavfig}
\end{figure}

The architecture corresponding to the feature extractor $\Phi_W$ in \eqref{help11}, illustrated in Fig. \ref{fig2}, is known as \textit{scattering network} \cite{MallatS}, and employs the frame $\Psi_{\Lambda_{\text{W}}}$ and the modulus non-linearity $|\cdot|$ in every network layer, but does not include pooling. For given $n \in \mathbb{N}$, the set $\Phi_W^n(f)$ in \eqref{help11} corresponds to the features of the function $f$ generated in the $n$-th network layer,  see Fig. \ref{fig2}. 
\begin{figure*}[t!]
\centering
\begin{tikzpicture}[scale=2.9,level distance=10mm,>=angle 60]


  \tikzstyle{every node}=[rectangle, inner sep=1pt]
  \tikzstyle{level 1}=[sibling distance=30mm]
  \tikzstyle{level 2}=[sibling distance=10mm]
  \tikzstyle{level 3}=[sibling distance=4mm]
  \node {$f$}
	child[grow=90, level distance=.45cm] {[fill=gray!50!black] circle (0.5pt)
		child[grow=130,level distance=0.5cm] 
        		child[grow=90,level distance=0.5cm] 
        		child[grow=50,level distance=0.5cm]
		child[level distance=.32cm,grow=215, densely dashed, ->] {}  
	}
        child[grow=150] {node {$|f\ast\psi_{\la^{^{(j)}} }|$}
	child[level distance=.75cm,grow=215, densely dashed, ->] {node {$|f\ast\psi_{\la^{^{(j)}}}|\ast\psi_{(-J,0)}$}
	}
	child[grow=83, level distance=0.5cm] 
	child[grow=97, level distance=0.5cm] 
        child[grow=110] {node {$||f\ast \psi_{\la^{^{(j)}}}|\ast\psi_{\la^{^{(l)}}}|$}
	child[level distance=.9cm,grow=215, densely dashed, ->] {node {$||f\ast \psi_{\la^{^{(j)}}}|\ast\psi_{\la^{^{(l)}}}|\ast\psi_{(-J,0)}$}
	}
        child[grow=130] {node {$|||f\ast \psi_{\la^{^{(j)}}}|\ast\psi_{\la^{^{(l)}}}|\ast\psi_{\lambda^{^{(m)}}}|$}
	child[level distance=0.75cm,grow=215, densely dashed, ->] {node[align=left]{\\ $\cdots$}}}
        child[grow=90,level distance=0.5cm] 
 	child[grow=50,level distance=0.5cm]
       }
       child[grow=63, level distance=1.05cm] {[fill=gray!50!black] circle (0.5pt)
	child[grow=130,level distance=0.5cm] 
       child[grow=90,level distance=0.5cm] 
       child[grow=50,level distance=0.5cm] 
       child[level distance=.375cm,grow=325, densely dashed, ->] {}     
       }
       }
       child[grow=30] {node {$|f\ast\psi_{\la^{^{^{(p)}}}}|$}
       child[level distance=0.75cm, grow=325, densely dashed, ->] {node {$|f\ast\psi_{\la^{^{(p)}}}|\ast\psi_{(-J,0)}$}
	}
	child[grow=83, level distance=0.5cm] 
	child[grow=97, level distance=0.5cm] 
        child[grow=117, level distance=1.05cm] {[fill=gray!50!black] circle (0.5pt)
        child[grow=130,level distance=0.5cm] 
        child[grow=90,level distance=0.5cm] 
        child[grow=50,level distance=0.5cm] 
        child[level distance=.375cm,grow=215, densely dashed, ->] {}  
	 }
        child[grow=70] {node {$||f\ast \psi_{\la^{^{(p)}}}|\ast\psi_{\la^{^{(r)}}}|$}
	 child[level distance=0.9cm,grow=325, densely dashed, ->] {node {$||f\ast \psi_{\la^{^{(p)}}}|\ast\psi_{\la^{^{(r)}}}|\ast\psi_{(-J,0)}$}}
	child[grow=130,level distance=0.5cm] 
         child[grow=90,level distance=0.5cm] 
             child[grow=50] {node {$|||f\ast \psi_{\la^{^{(p)}}}|\ast\psi_{\la^{^{(r)}}}|\ast\psi_{\lambda^{^{(s)}}}|$}
             child[level distance=0.75cm,grow=325, densely dashed, ->] {node[align=left]{\\ $\cdots$}}}
	}
     }
	child[level distance=0.75cm, grow=215, densely dashed, ->] {node {$f\ast \psi_{(-J,0)}$}};
\end{tikzpicture}
\caption{Scattering network architecture based on wavelet filters and the modulus non-linearity. The elements of the feature vector $\Phi_W(f)$ in \eqref{help11} are indicated at the tips of the arrows. } 
\label{fig2}
\end{figure*}
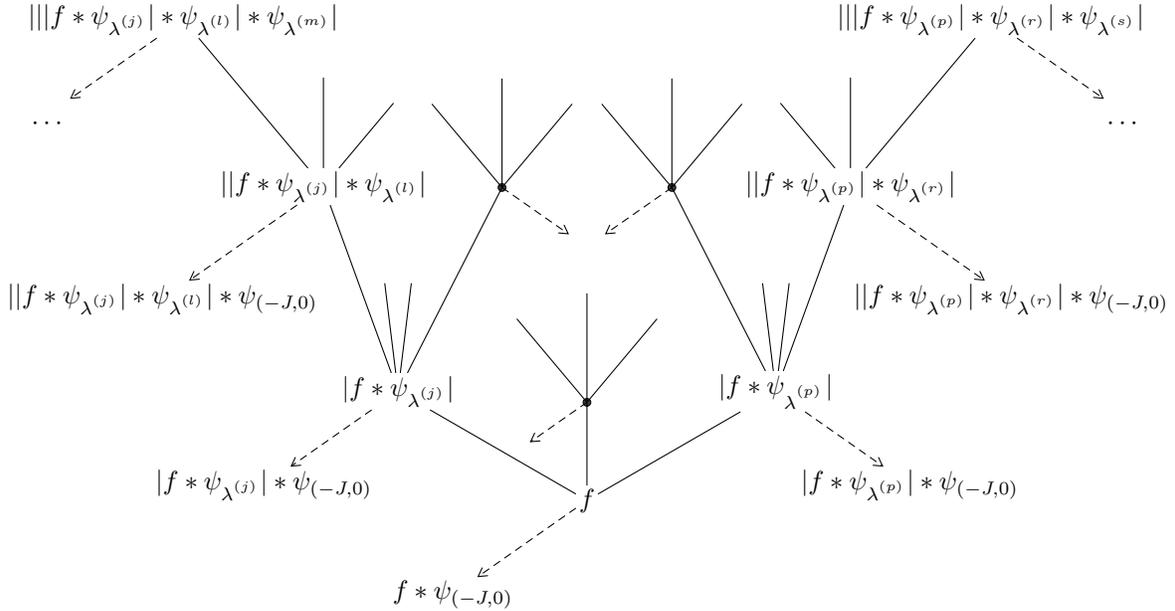
\begin{remark}\label{remark:coeff}
The function $|f\ast \psi_{\lambda}|$, $\la\in \Lambda_{\text{W}}\backslash\{(-J,0)\}$, can be thought of as indicating the locations of singularities of $f \in L^2(\Rd)$. Specifically, with the relation of $|f\ast \psi_{\lambda}|$ to the Canny edge detector \cite{Canny} as described in  \cite{MallatEdges}, in dimension $d=2$, we can think of $|f\ast\psi_\la|=|f\ast\psi_{(j,k)}|$, $\la=(j,k)\in\Lambda_{\text{W}}\backslash\{(-J,0)\}$, as an image at scale $j$ specifying the locations of edges of the image $f$ that are oriented in direction $k$. Furthermore, it was argued in \cite{Bruna,Anden,Oyallon} that the feature vector $\Phi^1_W(f)$ generated in the first layer of the scattering network is very similar, in dimension $d=1$, to mel frequency cepstral coefficients \cite{Davis}, and in dimension $d=2$ to SIFT-descriptors \cite{Lowe,Tola}. 
\end{remark}

It is shown in \cite[Theorem 2.10]{MallatS} that the feature extractor  $\Phi_W$ is translation-invariant in the sense of
\begin{equation}\label{approxti}
\lim_{J\to\infty}||| \Phi_W(T_t f) -\Phi_W(f)|||=0,
\end{equation}
for all $f\in L^2(\Rd)$ and  $t \in \Rd$.
This invariance result is asymptotic in the scale parameter $J\in \mathbb{Z}$ and does not depend on the network depth, i.e., it guarantees full translation invariance in every network layer. Furthermore, \cite[Theorem 2.12]{MallatS} establishes that $\Phi_W$ is stable w.r.t. deformations of the form $(F_{\tau}f)(x):=f(x-\tau(x)).$ More formally, for the function space $(H_{W},\VV \cdot \VV_{H_{W}}\hspace{-0.05cm})$ defined in \cite[Eq. 2.46]{MallatS}, it is shown in \cite[Theorem 2.12]{MallatS} that there exists a constant $C>0$ such that for all $f \in H_{W}$, and  $\tau \in C^1(\Rd,\Rd)$ with\footnote{It is actually the assumption $\VV D \tau \VV_\infty\leq\frac{1}{2d}$, rather than $\VV D \tau \VV_\infty\leq\frac{1}{2}$ as stated in \cite[Theorem 2.12]{MallatS}, that is needed in \cite[p. 1390]{MallatS} to establish that $|\det(E- (D \tau)(x))| \geq 1- d\VV D \tau \VV_\infty\geq 1/2$.}
 $\VV D \tau \VV_\infty\leq\frac{1}{2d}$, the deformation error satisfies the following deformation stability bound 
\begin{align}\label{mallatbound}
&||| \Phi_W(F_{\tau}f) -\Phi_W(f)|||\nonumber\\
&\leq C\big(2^{-J}\VV \tau \VV_\infty + J\VV D \tau \VV_\infty + \VV D^2 \tau \VV_\infty\big)\VV f \VV_{H_{W}}.
\end{align}
Note that this upper bound goes to infinity as translation invariance through  $J\to \infty$ is induced. In practice signal classification based on scattering networks is performed as follows. First, the function $f$ and the wavelet frame atoms $\{\psi_\lambda \}_{\lambda \in \Lambda_{\text{W}}}$ are discretized to finite-dimensional vectors. The resulting scattering network then computes the finite-dimensional feature vector $\Phi_W(f)$, whose dimension is typically reduced through an orthogonal least squares step \cite{SChen}, and then feeds the result into a trainable classifier such as, e.g., a SVM. State-of-the-art results for scattering networks were reported for  various classification tasks such as handwritten digit recognition \cite{Bruna}, texture discrimination \cite{Bruna,Sifre}, and musical genre classification \cite{Anden}.

\section{General deep convolutional \\feature extractors}\label{properties}
As already mentioned, scattering networks follow the architecture of DCNNs \cite{Bengio,LeCunNIPS89,Rumelhart,LeCunProc,LeCun,Nature,Huang,Jarrett,hierachies,Poultney,Pinto,Serre,GaborLowe} in the sense of cascading convolutions (with atoms $\{ \psi_\la\}_{\la\in\Lambda_{\text{W}}}$ of the wavelet frame $\Psi_{\Lambda_{\text{W}}}$) and non-linearities, namely the modulus function, but without pooling.  General DCNNs as studied in the literature exhibit a number of additional features: 
\begin{itemize}
\item[--]{a wide variety of filters are employed, namely pre-specified unstructured filters such as random filters \cite{Jarrett,hierachies}, and filters that are learned in a supervised \cite{Huang,Jarrett} or an unsupervised \cite{Jarrett,hierachies,Poultney} fashion.} 
\item[--]{a wide variety of  non-linearities are used such as, e.g., hyperbolic tangents \cite{Huang,Jarrett,hierachies}, rectified linear units \cite{Glorot,Nair}, and logistic sigmoids \cite{Mohamed,Glorot2}.}
\item[--]{convolution and the application of a non-linearity is typically followed by a pooling operator such as, e.g., sub-sampling \cite{Pinto}, average-pooling \cite{Huang,Jarrett}, or max-pooling \cite{GaborLowe,Serre,Jarrett,hierachies}.}
\item[--]{the filters, non-linearities, and pooling operators are allowed to be different in different network layers \cite{Nature,Goodfellow-et-al-2016}.}
\end{itemize}
As already mentioned, the purpose of this paper is to develop a mathematical theory of DCNNs for feature extraction that encompasses  all of the aspects above (apart from  max-pooling) with the proviso that the pooling operators we analyze are continuous-time emulations of discrete-time pooling operators. Formally, compared to scattering networks, in the $n$-th network layer, we replace the wavelet-modulus   operation $|f\ast \psi_\la|$ by a convolution with the atoms  $g_{\la_n}\in L^1(\Rd) \cap L^2(\Rd)$ of a general semi-discrete frame $\Psi_{n}:=\{T_b I g_{\lambda_n} \}_{b\in \Rd,\la_n \in \Lambda_n}$ for $L^2(\Rd)$ with countable index set $\Lambda_n$ (see Appendix \ref{sec:sdf} for a brief review of the theory of semi-discrete frames), followed by a non-linearity $M_n:L^2(\Rd)\to L^2(\Rd)$ that satisfies the Lipschitz property $\VV M_nf-M_nh\VV_2 \leq L_n \VV f-h \VV_2$, for all $f,h\in L^2(\Rd)$, and $M_nf=0$ for $f=0$. The output of this non-linearity, $M_n(f\ast g_{\la_{n}})$, is then pooled according to
\begin{equation}\label{eq:u1}
f\mapsto S_n^{d/2} P_n(f) (S_n\cdot),
\end{equation}   
where $S_n\geq 1$ is the pooling factor and $P_n:L^2(\Rd)\to L^2(\Rd)$ satisfies the Lipschitz property $\VV P_nf-P_nh\VV_2 \leq R_n \VV f-h \VV_2$, for all $f,h\in L^2(\Rd)$, and $P_nf=0$ for $f=0$.
We next comment on the individual elements in our network architecture in more detail. The frame atoms $g_{\la_n}$ are arbitrary and can, therefore, also be taken to be structured, e.g., Weyl-Heisenberg functions, curvelets, shearlets, ridgelets, or wavelets as considered in \cite{MallatS} (where the atoms $g_{\lambda_n}$ are obtained from a mother wavelet through scaling and rotation operations, see Section \ref{architecture}). The corresponding semi-discrete signal transforms\footnote{Let  $\{ g_\lambda\}_{\lambda \in \Lambda}\subseteq L^1(\Rd)\cap L^2(\Rd)$  be a set of functions indexed by a countable set $\Lambda$. Then, the mapping 
$
f\mapsto \{ f\ast g_\lambda(b)\}_{b\in \Rd, \lambda\in \Lambda}= \{ \langle f,T_{b}Ig_\lambda\rangle\}_{\lambda\in \Lambda}$, $f\in L^2(\Rd),$ is called a semi-discrete signal transform, as it depends on discrete indices $\lambda \in \Lambda$ and continuous variables $b\in \Rd$. We can think of this mapping as the analysis operator in frame theory \cite{Daubechies}, with the proviso that  for given $\lambda \in \Lambda$, we actually have a continuum of frame coefficients as the translation parameter $b\in \Rd$ is left unsampled.}, briefly reviewed in Appendices \ref{sec:sdf0} and \ref{sec:sdf2}, have been employed successfully in the literature in various feature extraction tasks  \cite{GaborFeature1,Wavelet1FeaturesD,Wavelet1Features2,UnserWavelets,GChen,Qiao,Ganesan,Plonka,Dettori}, but  their use---apart from wavelets---in DCNNs appears to be new. We refer the reader to Appendix \ref{app_nonlinear} for a detailed discussion of several relevant example non-linearities (e.g., rectified linear units, shifted logistic sigmoids, hyperbolic tangents, and, of course, the modulus function) that fit into our framework. We next explain how the continuous-time pooling operator \eqref{eq:u1} emulates discrete-time pooling by sub-sampling \cite{Pinto} or by averaging  \cite{Huang,Jarrett}. Consider a one-dimensional discrete-time signal $f_{\text{d}}\in \ell^2(\mathbb{Z}):=\{f_{\text{d}}:\mathbb{Z}\to\mathbb{C} \ | \ \sum_{k\in \mathbb{Z}} |f_{\text{d}}[k]|^2<\infty \}$. Sub-sampling by a factor of $S\in \mathbb{N}$ in discrete time is defined by \cite[Sec. 4]{Vaidyanathan}
\begin{equation*}\label{eq:ssdsgps}
f_{\text{d}} \mapsto h_{\text{d}}:=f_{\text{d}}[S\cdot]
\end{equation*}
and amounts to simply retaining every $S$-th sample of $f_{\text{d}}$. The discrete-time Fourier transform of $h_{\text{d}}$ is given by a summation over translated and dilated copies of $\widehat{f_{\text{d}}}$ according to \cite[Sec. 4]{Vaidyanathan}
\begin{equation}\label{eq:dtft}
\widehat{h_{\text{d}}}(\theta):=\sum_{k\in \mathbb{Z}}h_{\text{d}}[k]e^{-2\pi i k\theta}=\frac{1}{S}\sum_{k=0}^{S-1}\widehat{f_{\text{d}}}\Big(\frac{\theta-k}{S}\Big).
\end{equation}
The translated copies of $\widehat{f_{\text{d}}}$ in \eqref{eq:dtft} are a consequence of the $1$-periodicity of the discrete-time Fourier transform. We therefore emulate the discrete-time sub-sampling operation in continuous time through the dilation operation 
\begin{equation}\label{eq:css}
f\mapsto h:=S^{d/2}f(S\cdot),\quad f\in L^2(\Rd),
\end{equation}
which in the frequency domain amounts to dilation according to $\widehat{h}=S^{-d/2}\widehat{f}(S^{-1}\cdot)$. The scaling by $S^{d/2}$ in \eqref{eq:css} ensures unitarity of the continuous-time sub-sampling operation. The overall operation in \eqref{eq:css} fits into our general definition of pooling  as it can be recovered from \eqref{eq:u1} simply by  taking $P$ to equal the identity mapping  (which is, of course, Lipschitz-continuous with Lipschitz constant $R=1$ and satisfies $\text{Id}f=0$ for $f=0$). Next, we consider average pooling. In discrete time average pooling is defined by
\begin{equation}\label{eq:avg1}
f_{\text{d}} \mapsto h_{\text{d}}:=(f_{\text{d}}\ast \phi_{\text{d}})[S\cdot]
\end{equation}
for the (typically compactly supported) ``averaging kernel'' $\phi_{\text{d}}\in \ell^2(\mathbb{Z})$ and the averaging factor $S\in \mathbb{N}$. Taking $\phi_d$ to be a box function of length $S$ amounts to computing local averages of $S$ consecutive samples. Weighted averages are obtained by identifying the desired weights with the averaging kernel $\phi_d$. The operation \eqref{eq:avg1} can be emulated in continuous time according to
\begin{equation}\label{new_pool}
f\mapsto S^{d/2}\big( f\ast \phi\big)(S\cdot),\quad f\in L^2(\Rd),
\end{equation}
with the averaging window $\phi\in L^1(\Rd)\cap L^2(\Rd)$. We note that \eqref{new_pool} can be recovered from \eqref{eq:u1} by  taking $P(f)=f\ast \phi$, $f\in L^2(\Rd)$, and noting that convolution with $\phi$ is Lipschitz-continuous with Lipschitz constant $R=\| \phi \|_1$ (thanks to Young's inequality \cite[Theorem 1.2.12]{Grafakos}) and trivially satisfies $Pf=0$ for $f=0$.
In the remainder of the paper, we refer to the operation in \eqref{eq:u1} as \textit{Lipschitz pooling through dilation} to indicate that \eqref{eq:u1} essentially amounts to the application of a Lipschitz-continuous mapping followed by a continuous-time dilation. Note, however, that the operation in \eqref{eq:u1} will not be unitary in general. 

\begin{figure*}[t!]
\centering
\begin{tikzpicture}[scale=2.9,level distance=10mm,>=angle 60]

  \tikzstyle{every node}=[rectangle, inner sep=1pt]
  \tikzstyle{level 1}=[sibling distance=30mm]
  \tikzstyle{level 2}=[sibling distance=10mm]
  \tikzstyle{level 3}=[sibling distance=4mm]
  \node {$U[e]f=f$}
	child[grow=90, level distance=.45cm] {[fill=gray!50!black] circle (0.5pt)
		child[grow=130,level distance=0.5cm] 
        		child[grow=90,level distance=0.5cm] 
        		child[grow=50,level distance=0.5cm]
		child[level distance=.25cm,grow=215, densely dashed, ->] {}  
	}
        child[grow=150] {node {$U\big[\lambda_1^{(j)}\big]f$}
	child[level distance=.75cm,grow=215, densely dashed, ->] {node {$\big(U\big[\lambda_1^{(j)}\big]f\big)\ast\chi_{1}$}
	}
	child[grow=83, level distance=0.5cm] 
	child[grow=97, level distance=0.5cm] 
        child[grow=110] {node {$U\big[\big(\lambda_1^{(j)},\lambda_2^{(l)}\big)\big]f$}
	child[level distance=.9cm,grow=215, densely dashed, ->] {node {$\big(U\big[\big(\lambda_1^{(j)},\lambda_2^{(l)}\big)\big]f\big)\ast\chi_{2}$}
	}
        child[grow=130] {node {$U\big[\big(\lambda_1^{(j)},\lambda_2^{(l)},\lambda_3^{(m)}\big)\big]f$}
	child[level distance=0.75cm,grow=215, densely dashed, ->] 
	{node[align=left]{\\ $\cdots$}}
	}
        child[grow=90,level distance=0.5cm]
 	child[grow=50,level distance=0.5cm]
       }
       child[grow=63, level distance=1.05cm] {[fill=gray!50!black] circle (0.5pt)
	child[grow=130,level distance=0.5cm] 
       child[grow=90,level distance=0.5cm] 
       child[grow=50,level distance=0.5cm]
       child[level distance=.25cm,grow=325, densely dashed, ->] {}    
       }
       }
       child[grow=30] {node {$U\big[\lambda_1^{(p)}\big]f$}
       child[level distance=0.75cm, grow=325, densely dashed, ->] {node {$\big(U\big[\lambda_1^{(p)}\big]f\big)\ast\chi_{1}$}
	}
	child[grow=83, level distance=0.5cm] 
	child[grow=97, level distance=0.5cm]
        child[grow=117, level distance=1.05cm] {[fill=gray!50!black] circle (0.5pt)
        child[grow=130,level distance=0.5cm] 
        child[grow=90,level distance=0.5cm] 
        child[grow=50,level distance=0.5cm] 
        child[level distance=.25cm,grow=215, densely dashed, ->] {}  
	 }
        child[grow=70] {node {$U\big[\big(\lambda_1^{(p)},\lambda_2^{(r)}\big)\big]f$}
	 child[level distance=0.9cm,grow=325, densely dashed, ->] {node {$\big(U\big[\big(\lambda_1^{(p)},\lambda_2^{(r)}\big)\big]f\big)\ast\chi_{2}$}}
	child[grow=130,level distance=0.5cm] 
         child[grow=90,level distance=0.5cm] 
             child[grow=50] {node {$U\big[\big(\lambda_1^{(p)},\lambda_2^{(r)},\lambda_3^{(s)}\big)\big]f$}
             child[level distance=0.75cm,grow=325, densely dashed, ->] {node[align=left]{\\ $\cdots$}}}
	}
     }
	child[level distance=0.75cm, grow=215, densely dashed, ->] {node {$f\ast \chi_0$}};
\end{tikzpicture}
\caption{Network architecture underlying the general DCNN feature extractor. The index $\lambda_{n}^{(k)}$ corresponds to the $k$-th atom $g_{\lambda_{n}^{(k)}}$ of the frame $\Psi_n$ associated with the $n$-th network layer. The function $\chi_{n}$ is the output-generating atom of the $n$-th layer.} 
\label{fig_gsn}
\end{figure*}
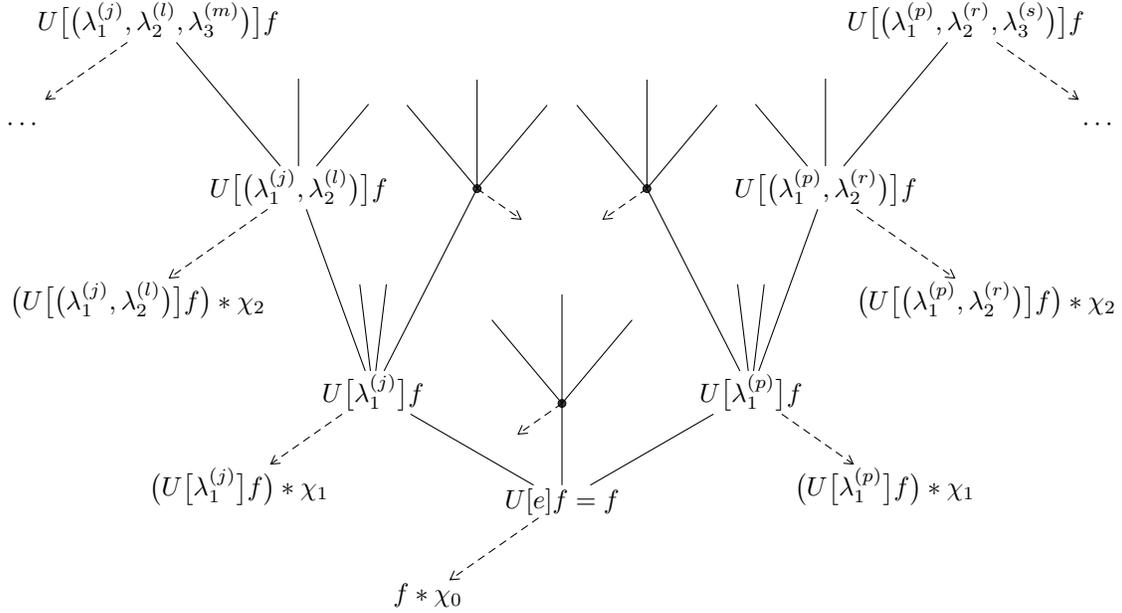

We next state definitions and collect preliminary results needed for the  analysis of the general DCNN feature extractor  considered. The basic building blocks of this network are the triplets $(\Psi_n,M_n,P_n)$ associated with individual network layers $n$ and referred to as \textit{modules}.
\begin{definition} For $n \in \mathbb{N}$, let $\Psi_{n}=\{T_bIg_{\la_n} \}_{b \in \Rd, \la_n \in \Lambda_n}$  be a semi-discrete frame for $L^2(\Rd)$ and let $M_n:L^2(\Rd) \to L^2(\Rd)$ and $P_n:L^2(\Rd) \to L^2(\Rd)$ be  Lipschitz-continuous operators with  $M_nf=0$ and $P_nf=0$ for $f=0$, respectively. Then, the sequence of triplets
$$
\Omega:=\big( (\Psi_n,M_n,P_n)\big)_{n \in \mathbb{N}}
$$
is referred to as a module-sequence. 
\end{definition}
The following definition introduces the concept of paths on index sets, which will prove useful in formalizing the feature extraction network. The idea for this formalism is due to  \cite{MallatS}.
\begin{definition} \label{changes} Let $\Omega=\big( (\Psi_n,M_n,P_n)\big)_{n \in \mathbb{N}}$ be a module-sequence, let $\{ g_{\lambda_n}\}_{\lambda_n \in \Lambda_n}$ be the atoms of the frame $\Psi_n$, and let $S_n\geq 1$ be the pooling factor (according to \eqref{eq:u1}) associated with the $n$-th network layer. Define the operator $U_n$ associated with the $n$-th layer of the network  as $U_n: \Lambda_n \times L^2(\Rd) \to L^2(\Rd)$, 
\begin{equation}\label{eq:e1} 
U_n(\lambda_n,f):=U_n[\lambda_n]f:=S_n^{d/2}P_n\big(M_n(f \ast g_{\lambda_n})\big)(S_n\cdot).
\end{equation}
For $n\in \mathbb{N}$, define the set $\Lambda_1^n:=\Lambda_1\times \Lambda_{2}\times \dots \times \Lambda_n$.  An ordered sequence $q=(\lambda_1,\lambda_{2},\dots, \lambda_n) \in \Lambda_1^n$ is called a path. For the empty path $e:=\emptyset$ we set $\Lambda_1^0:=\{ e \}$ and $U_0[e]f:=f$, for all $f\in L^2(\Rd)$. \end{definition}
The operator $U_n$ is well-defined, i.e., $U_n[\lambda_n]f\in L^2(\Rd)$, for all $(\lambda_n,f)\in \Lambda_n\times L^2(\Rd)$, thanks to
\begin{align}
&\VV U_n[\lambda_n] f \VV_2^2 = S_n^d\int_{\RR^d}\Big|P_n\big(M_n(f \ast g_{\lambda_n})\big)(S_nx)\Big|^2\mathrm dx\nonumber\\&=\int_{\RR^d}\Big|P_n\big(M_n(f \ast g_{\lambda_n})\big)(y)\Big|^2\mathrm dy\nonumber\\
&=\| P_n\big(M_n(f\ast g_{\lambda_n})\big)\|_2^2\leq R_n^2 \VV M_n(f\ast g_{\lambda_n})\VV_2^2 \label{eq:abba}\\
&\leq L_n^2R_n^2\VV f\ast g_{\lambda_n}\VV_2^2\leq B_nL_n^2R_n^2 \VV f\VV_2^2\label{wellU}.
\end{align}
For the inequality in \eqref{eq:abba} we used the Lipschitz continuity of $P_n$ according to $\VV P_nf-P_nh\VV^2_2 \leq R_n^2 \VV f-h \VV^2_2$, together with $P_nh=0$ for $h=0$ to get $\VV P_nf \VV_2^2\leq R_n^2\VV f \VV^2_2$. Similar arguments lead to the first inequality in \eqref{wellU}. The last step in \eqref{wellU} is thanks to
$$\VV f\ast g_{\lambda_n}\VV_2^2\leq \sum_{\lambda'_{n} \in \Lambda_n}\VV f\ast g_{\lambda'_{n}}\VV_2^2 \leq B_n\VV f \VV_2^2,$$
which follows from the frame condition \eqref{condii} on $\Psi_n$. We will also need the extension of the operator $U_n$ to paths $q \in \Lambda_1^n$ according to
\begin{align}
U[q]f=&\,U[(\lambda_1,\lambda_{2},\dots,\lambda_n)]f\nonumber\\
:=&\,U_n[\lambda_n] \cdots U_{2}[\lambda_{2}]U_{1}[\lambda_{1}]f,\label{aaaa}
\end{align}
with $U[e]f:=f$.
Note that the multi-stage opera\-tion \eqref{aaaa} is again well-defined thanks to
\begin{equation}\label{eq:p13}
\VV U[q] f \VV_2^2 \leq \Bigg(\prod_{k=1}^{n}B_kL_k^2R_k^2 \Bigg) \VV f \VV^2_2,
\end{equation}
for  $q\in \Lambda_1^n$ and $f\in L^2(\Rd)$,  which follows by repeated application of \eqref{wellU}. 

In scattering networks one atom $ \psi_\lambda$, $\la \in  \Lambda_{\text{W}}$, in the wavelet frame $\Psi_{\Lambda_{\text{W}}}$, namely the low-pass filter $\psi_{(-J,0)}$, is singled out to generate the extracted features according to \eqref{help11}, see also Fig. \ref{fig2}. We follow this construction and designate one of the atoms in each frame  in the module-sequence $\Omega=\big( (\Psi_n,M_n,P_n)\big)_{n \in \mathbb{N}}$ as the output-generating atom $\chi_{n-1}:=g_{\lambda^\ast_n}$, $\lambda^\ast_n \in \Lambda_n$, of the $(n-1)$-th layer. The atoms $\{ g_{\lambda_n}\}_{\lambda_n \in \Lambda_n\backslash\{ \lambda^\ast_n\}}\cup \{ \chi_{n-1}\}$ in $\Psi_n$ are thus used across two consecutive layers in the sense of $\chi_{n-1}=g_{\lambda^\ast_n}$ generating the output in the $(n-1)$-th layer, and the $\{ g_{\lambda_n}\}_{\lambda_n \in \Lambda_n\backslash\{ \lambda^\ast_n\}}$  propagating signals from the $(n-1)$-th layer to the $n$-th layer according to \eqref{eq:e1}, see Fig. \ref{fig_gsn}. Note, however, that our theory does not require the output-generating atoms to be low-pass filters\footnote{It is evident, though,  that the actual choices of the output-generating atoms will have an impact on practical performance.}. From now on, with slight abuse of notation, we shall write $\Lambda_n$ for $\Lambda_n\backslash\{ \lambda^\ast_n\}$ as well. Finally, we note that extracting features in every network layer via an output-generating atom can be regarded as employing skip-layer connections \cite{he2015deep}, which skip  network layers further down and feed the propagated signals into the feature vector.

We are now ready to define the feature extractor $\Phi_\Omega$ based on the module-sequence $\Omega$. 
\begin{definition}\label{defn2}
 Let $\Omega=\big( (\Psi_n,M_n,P_n)\big)_{n \in \mathbb{N}}$ be a module-sequence. The feature extractor $\Phi_\Omega$ based on $\Omega$ maps $L^2(\Rd)$ to its feature vector
 \begin{equation}\label{ST}
\Phi_\Omega (f):=\bigcup_{n=0}^\infty \Phi_\Omega^n(f),
\end{equation}
where $\Phi_\Omega^n (f):= \{ (U[q]f) \ast \chi_{n} \}_{q \in \Lambda_1^n}$, for all $n\in \mathbb{N}$.
\end{definition}
The set $\Phi_\Omega^n (f)$ in \eqref{ST} corresponds to the features of the function $f$ generated in the $n$-th network layer, see Fig. \ref{fig_gsn}, where $n=0$ corresponds to the root of the network. The feature extractor $\Phi_\Omega:L^2(\Rd)\to(L^2(\Rd))^\QQ$, with $\QQ:=\bigcup_{n=0}^\infty\Lambda_1^n$, is well-defined, i.e., $\Phi_\Omega(f) \in (L^2(\Rd))^\QQ$,  for all $ f\in L^2(\Rd)$, under a technical condition on the module-sequence $\Omega$ formalized as follows.

\begin{figure*}[t]
\hspace{3.5cm}
        \begin{subfigure}[]{}
                        \hspace{-1.6cm}            
                            \includegraphics[width = .175\textwidth]{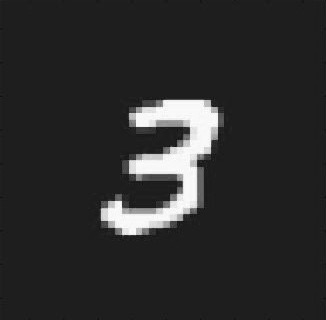}
        \end{subfigure}
        \hspace{3cm}
                \begin{subfigure}[]{}
                      \hspace{-1.6cm}    
                \includegraphics[width = .175\textwidth]{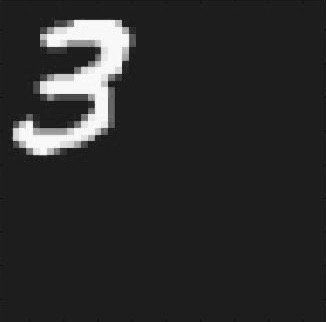}
        \end{subfigure}
                \hspace{3cm}
                        \begin{subfigure}[]{}
                      \hspace{-1.6cm}    
                \includegraphics[width = .175\textwidth]{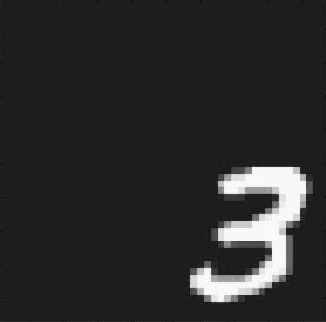}
        \end{subfigure}
             \caption{Handwritten digits from the MNIST data set \cite{MNIST}. For practical machine learning tasks (e.g., signal classification), we often want the feature vector $\Phi_\Omega(f)$ to be invariant to the digits' spatial location within the image $f$. Theorem \ref{main_inv2} establishes that the features  $\Phi^n_\Omega(f)$ become  more translation-invariant with increasing layer index $n$.}\label{fig:transaaaaa}
\end{figure*}

\begin{proposition}\label{main_well}
 Let $\Omega=\big( (\Psi_n,M_n,P_n)\big)_{n \in \mathbb{N}}$ be a module-sequence. Denote the frame upper bounds of $\Psi_n$ by $B_n>0$ and the Lipschitz constants of the operators $M_n$ and $P_n$ by $L_n>0$ and $R_n>0$, respectively. If
 \begin{equation}\label{weak_admiss1}
 \max\{B_n,B_nL_n^2R_n^2 \}\leq 1,\hspace{0.5cm} \forall n \in \mathbb{N}, 
 \end{equation}
then the feature extractor $\Phi_\Omega:L^2(\Rd)\to(L^2(\Rd))^\QQ$ is well-defined, i.e., $\Phi_\Omega(f) \in (L^2(\Rd))^\QQ$, for all $f\in L^2(\Rd)$. 
\end{proposition}
\begin{proof}
The proof is given in Appendix \ref{app_well}.
\end{proof}
As condition \eqref{weak_admiss1} is of central importance, we formalize it as follows.
\begin{definition} \label{def:weakq}
Let $\Omega=\big( (\Psi_n,M_n,P_n)\big)_{n \in \mathbb{N}}$ be a module-sequence with frame upper bounds $B_n>0$ and Lipschitz constants $L_n,  R_n>0$ of the operators $M_n$ and $P_n$, respectively. The condition
\begin{equation}\label{weak_admiss2}
 \max\{B_n,B_nL_n^2R_n^2 \}\leq 1,\hspace{0.5cm} \forall n \in \mathbb{N}, 
 \end{equation}
is referred to as admissibility condition. Module-sequences that satisfy \eqref{weak_admiss2} are called  admissible.
\end{definition}
We emphasize that condition \eqref{weak_admiss2} is easily met in practice. To see this, first note that $B_n$ is determined through the frame $\Psi_n$ (e.g., the directional wavelet frame introduced in Section \ref{architecture} has $B=1$), $L_n$ is set through the non-linearity $M_n$ (e.g., the modulus function $M=|\cdot|$ has $L=1$, see Appendix \ref{app_nonlinear}), and $R_n$ depends on the operator $P_n$ in \eqref{eq:u1} (e.g., pooling by sub-sampling amounts to $P=\text{Id}$ and has $R=1$). Obviously, condition \eqref{weak_admiss2} is met if
\begin{equation*}\label{eq:d2hut}
B_n\leq \min\{1,L_n^{-2}R_n^{-2}\}, \hspace{0.5cm} \forall n \in \mathbb{N},
\end{equation*} 
which can be satisfied by simply normalizing the frame elements of $\Psi_n$ accordingly. We refer to Proposition \ref{prop:norm} in Appendix \ref{sec:sdf} for corresponding normalization techniques, which, as explained in  Section \ref{main}, affect neither our translation invariance result nor our deformation sensitivity bounds. 
\section{Properties of the feature extractor $\Phi_\Omega$}\label{main}

\subsection{Vertical translation invariance}\label{sec:deep}
The following theorem states that under very mild decay conditions on the Fourier transforms $\widehat{\chi_n}$ of the output-generating atoms $\chi_n$, the feature extractor $\Phi_\Omega$ exhibits vertical translation invariance in the sense of the features becoming more translation-invariant with increasing network depth. This result is in line with observations made in the deep learning literature, e.g., in \cite{Huang,Jarrett,hierachies,Serre,GaborLowe}, where it is informally argued that the network outputs generated at deeper layers tend to be more translation-invariant. 
\begin{theorem}\label{main_inv2}
Let $\Omega=\big( (\Psi_n,M_n,P_n)\big)_{n \in \mathbb{N}}$ be an admissible module-sequence, let $S_n\geq 1$, $n\in \mathbb{N}$, be the pooling factors in \eqref{eq:e1}, and assume that the operators $M_n:L^2(\Rd) \to L^2(\Rd)$ and $P_n:L^2(\Rd) \to L^2(\Rd)$ commute with the translation operator $T_t$, i.e., 
\begin{equation}\label{eq:i1}
M_nT_tf=T_tM_nf, \hspace{0.5cm} P_nT_tf=T_tP_nf, \end{equation}
for all $f\in L^2(\Rd)$, $t \in \Rd$, and $n\in\mathbb{N}$. 

 \begin{itemize}
 \item[i)]{
The features  $\Phi_\Omega^n(f)$ generated in the $n$-th network layer satisfy \begin{equation}\label{eq:i2}
\Phi_{\Omega}^n(T_tf)=T_{t/(S_{1}\cdots \,S_{n})}\Phi_\Omega^n(f), \end{equation}
for all $f\in L^2(\Rd)$, $t \in \Rd$, and $n\in\mathbb{N}$, where $T_t \Phi_\Omega^n(f)$ refers to element-wise application of $T_t$, i.e., $T_t\Phi_\Omega^n(f):=\{ T_th\, |\, \forall h \in \Phi_\Omega^n(f) \}$.
}
\item[ii)]{If, in addition, there exists a constant $K >0$ (that does not depend on $n$) such that the Fourier transforms $\widehat{\chi_n}$ of the output-generating atoms $\chi_n$ satisfy the decay condition 
\begin{equation}\label{decaycondition}
|\widehat{\chi_n}(\omega)||\omega| \leq K, \hspace{0.5cm} \text{ a.e. } \omega \in \Rd,\ \forall n \in \mathbb{N}_0,
\end{equation}
then  
\begin{equation}\label{eq:vert3}
||| \Phi^n_\Omega(T_tf) - \Phi^n_\Omega(f)|||\leq \frac{2\pi| t|K}{S_1\cdots S_n}\VV f \VV_2, 
\end{equation}
for all  $f \in L^2(\Rd)$ and $t\in \Rd.$}
\end{itemize}
\end{theorem}
\begin{proof}
The proof is given in Appendix \ref{app_inv}.
\end{proof}

\begin{figure*}[t]
\hspace{3.5cm}
        \begin{subfigure}[]{}
                        \hspace{-1.6cm}            
                           \includegraphics[width = .175\textwidth]{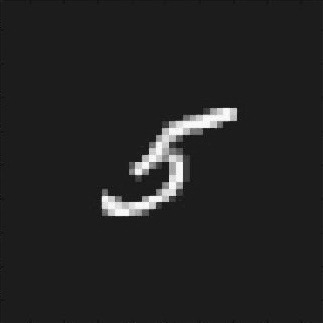}
        \end{subfigure}
        \hspace{3cm}
                \begin{subfigure}[]{}
                      \hspace{-1.6cm}    
               \includegraphics[width = .175\textwidth]{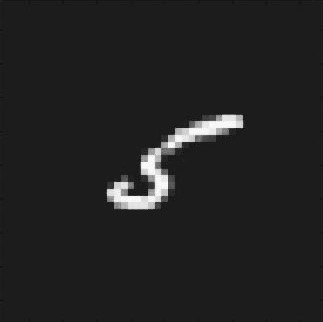}
        \end{subfigure}
                \hspace{3cm}
                        \begin{subfigure}[]{}
                      \hspace{-1.6cm}    
             \includegraphics[width = .175\textwidth]{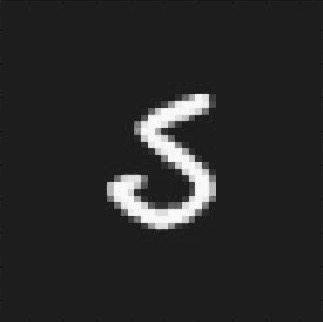}

        \end{subfigure}
             \caption{Handwritten digits from the MNIST data set \cite{MNIST}. If $f$ denotes the image of the handwritten digit ``$5$'' in (a), then---for appropriately chosen $\tau$---the function $F_\tau f=f(\cdot-\tau(\cdot))$ models images of ``$5$''  based on different handwriting styles as in (b) and (c). }\label{fig:deform}
\end{figure*}

We start by noting that all pointwise (also referred to as  memoryless in the signal processing literature) non-linearities $M_n:L^2(\Rd)\to L^2(\Rd)$ satisfy the  commutation relation in \eqref{eq:i1}. A large class of non-linearities widely used in the deep learning literature, such as rectified linear units, hyperbolic tangents, shifted logistic sigmoids, and the modulus function as employed in \cite{MallatS}, are, indeed, pointwise and hence covered by Theorem \ref{main_inv2}. 
Moreover, $P=\text{Id}$ as in pooling by sub-sampling trivially satisfies \eqref{eq:i1}. Pooling by averaging $Pf=f \ast \phi$, with $\phi \in L^1(\Rd)\cap L^2(\Rd)$, satisfies \eqref{eq:i1} as a consequence of the convolution operator commuting with the translation operator $T_t$. 

Note that \eqref{decaycondition} can easily be met by taking the output-generating atoms $\{\chi_n\}_{n\in \mathbb{N}_0}$ either to satisfy 
\begin{equation*}\label{eq:philipp}
\sup_{n\in \mathbb{N}_0}\{ \VV\chi_n \VV_1 + \VV \nabla \chi_n \VV_1\} <\infty,
\end{equation*} see, e.g., \cite[Ch. 7]{Rudin}, or to be uniformly band-limited in the sense of $\supp (\widehat{\chi_n})\subseteq B_r(0)$, for all $n \in \mathbb{N}_0$, with an $r$ that is independent of $n$ (see, e.g., \cite[Ch. 2.3]{MallatW}). The bound in \eqref{eq:vert3} shows that we can explicitly control the amount of translation invariance via the pooling factors $S_n$. This result is in line with observations made in the deep learning literature, e.g., in \cite{Huang,Jarrett,hierachies,Serre,GaborLowe}, where it is informally argued  that pooling is crucial to get translation invariance of the extracted features.  Furthermore, the condition $\lim\limits_{n\to \infty}S_1\cdot S_2\cdot \ldots \cdot S_n = \infty$ (easily met by taking  $S_n>1$, for all $n \in \mathbb{N}$) guarantees, thanks to \eqref{eq:vert3}, asymptotically full translation invariance according to 
\begin{equation}\label{awesome}
\lim\limits_{n\to \infty}||| \Phi_\Omega^n(T_tf) - \Phi_\Omega^n(f)|||=0, 
\end{equation}
for all $f \in L^2(\Rd)$ and $t\in \Rd$. This means that the features $\Phi_\Omega^n(T_tf)$ corresponding to the shifted versions $T_tf$ of the handwritten digit ``$3$'' in Figs.  \ref{fig:transaaaaa} (b) and (c) with increasing network depth increasingly ``look like'' the features $\Phi_\Omega^n(f)$ corresponding to the unshifted handwritten digit in Fig. \ref{fig:transaaaaa} (a). Casually speaking, the shift operator $T_t$ is increasingly absorbed by $\Phi_\Omega^n$ as $n\to \infty$, with the upper bound \eqref{eq:vert3} quantifying this absorption.

In contrast, the translation invariance result \eqref{approxti} in \cite{MallatS} is asymptotic in the wavelet scale parameter $J$, and does not depend on the network depth, i.e., it guarantees full translation invariance in every network layer. We honor this difference by referring to \eqref{approxti} as \textit{horizontal} translation invariance and to \eqref{awesome} as \textit{vertical} translation invariance. 

We emphasize that vertical translation invariance is a structural property. Specifically, if $P_n$ is unitary (such as, e.g.,  in the case of pooling by sub-sampling where $P_n$ simply equals the identity mapping), then so is the pooling operation in \eqref{eq:u1} owing to
\begin{align*}
&\| S_n^{d/2}P_n(f)(S_n\cdot)\|_2^2=S^{d}_n\int_{\Rd}|P_n(f)(S_nx)|^2\mathrm dx\\
&=\int_{\Rd}|P_n(f)(x)|^2\mathrm dx=\|P_n(f)\|_2^2=\|f\|_2^2,
\end{align*}
where we employed the change of variables $y=S_nx$, $\frac{\mathrm dy}{\mathrm dx}=S_n^d$. Regarding average pooling, as already mentioned, the operators $P_n(f)=f\ast \phi_n$, $f\in L^2(\Rd)$, $n \in \mathbb{N}$, are, in general, not unitary,  but we still get translation invariance as a consequence of  structural properties, namely translation covariance of the convolution operator combined with unitary dilation according to \eqref{eq:css}.

Finally, we note that in practice in certain applications it is actually translation \textit{covariance} in the sense of $\Phi^n_\Omega(T_tf)=T_t\Phi^n_\Omega(f)$, for all $f\in L^2(\Rd)$ and  $t\in \Rd$, that is desirable, for example, in facial landmark detection where the goal is to estimate the absolute position of facial landmarks in images. In such applications features in the layers closer to the root of the network are more relevant as they are less translation-invariant and more translation-covariant. The reader is referred to \cite{wiatowski2016discrete} where corresponding numerical evidence is provided. We proceed to the formal statement of our translation covariance result.
\begin{corollary}\label{main_inv55}
Let $\Omega=\big( (\Psi_n,M_n,P_n)\big)_{n \in \mathbb{N}}$ be an admissible module-sequence, let $S_n\geq 1$, $n\in \mathbb{N}$, be the pooling factors in \eqref{eq:e1}, and assume that the operators $M_n:L^2(\Rd) \to L^2(\Rd)$ and $P_n:L^2(\Rd) \to L^2(\Rd)$ commute with the translation operator $T_t$ in the sense of \eqref{eq:i1}. If, in addition, there exists a constant $K >0$ (that does not depend on $n$) such that the Fourier transforms $\widehat{\chi_n}$ of the output-generating atoms $\chi_n$ satisfy the decay condition \eqref{decaycondition}, then 
\begin{equation*}\label{eq:vert6}
||| \Phi^n_\Omega(T_tf) - T_t\Phi^n_\Omega(f)|||\leq 2\pi|t|K\big|1/(S_1\dots S_n)-1\big|\VV f \VV_2,
\end{equation*}
for all $f \in L^2(\Rd)$ and $t\in \Rd$.
\end{corollary}
\begin{proof}
The proof is given in Appendix \ref{app:covcov}.
\end{proof}  
Corollary \ref{main_inv55} shows that in the absence of pooling, i.e., taking $S_n=1$, for all $n \in \mathbb{N}$, leads to full translation covariance in every network layer. This proves that pooling is necessary to get vertical translation invariance as otherwise the features remain fully translation-covariant irrespective of the network  depth. Finally, we note that scattering networks \cite{MallatS} (which do not employ pooling operators, see Section \ref{architecture}) are rendered horizontally translation-invariant by letting the wavelet scale parameter $J\to\infty$. 

\subsection{Deformation sensitivity bound}\label{sec:defmstab}
The next result provides a bound---for band-limited signals $f\in L^2_R(\Rd)$---on the sensitivity of the feature extractor $\Phi_\Omega$  w.r.t. time-frequency deformations of the form
\begin{equation*}\label{def:deform}
(F_{\tau,\omega} f)(x):=e^{2\pi i \omega(x)}f(x-\tau(x)).
\end{equation*}
This class of deformations encompasses non-linear distortions $f(x-\tau(x))$ as illustrated in Fig. \ref{fig:deform}, and modulation-like deformations $e^{2\pi i \omega(x)}f(x)$ which occur, e.g., if the  signal $f$ is subject to an undesired modulation and we therefore have access to a bandpass version of $f$ only. 
 
The deformation sensitivity bound we derive is signal-class specific in the sense of applying to input signals belonging to a particular class, here band-limited functions. The proof technique we develop applies, however, to all signal classes that exhibit ``inherent'' deformation insensitivity in the following sense. 

\begin{definition}
A signal class $\mathcal{C}\subseteq L^2(\Rd)$ is called deformation-insensitive if there exist $\alpha,\beta, C>0$ such that for all $f\in \mathcal{C}$,  $\omega \in C(\Rd,\RR)$, and  (possibly non-linear) $\tau \in C^1(\Rd,\Rd)$ with $\VV D \tau \VV_\infty\leq\frac{1}{2d}$, it holds that
\begin{equation}\label{particulaaa}
\VV f-F_{\tau,\omega} f\VV_2 \leq C \big(\VV \tau \VV_{\infty}^{\alpha} + \VV \omega \VV_{\infty}^{\beta} \big).
\end{equation}
\end{definition}
The constant $C>0$ and the exponents $\alpha,\beta>0$ in \eqref{particulaaa}  depend on the particular signal class $\mathcal{C}$. Examples of deformation-insensitive signal classes are the class of $R$-band-limited functions (see Proposition \ref{uppersuper} in Appendix \ref{proof:uppersuper}), the class of cartoon functions \cite[Proposition 1]{grohs_wiatowski},  and the class of Lipschitz functions \cite[Lemma 1]{grohs_wiatowski}. While a deformation sensitivity bound that applies to all $f\in L^2(\Rd)$ would  be desirable, the example in Fig. \ref{fig:sigclass} illustrates the difficulty underlying   this desideratum. Specifically, we can see in Fig. \ref{fig:sigclass} that for given $\tau(x)$ and $\omega(x)$ the impact of the deformation induced by $e^{2\pi i \omega(x)}f(x-\tau(x))$ can depend drastically on the function $f\in L^2(\Rd)$ itself. The deformation stability bound  \eqref{mallatbound} for scattering networks reported in \cite[Theorem 2.12]{MallatS} applies to a signal class as well, characterized, albeit implicitly, through \cite[Eq. 2.46]{MallatS} and depending on the mother wavelet and the (modulus) non-linearity.

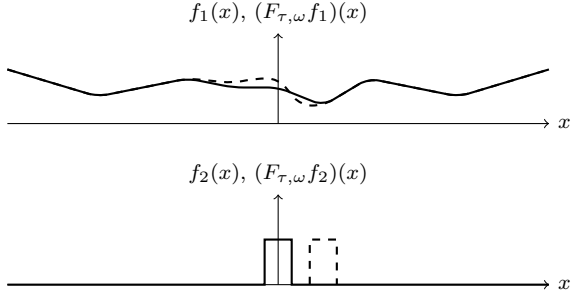
\begin{figure}[t]
\begin{center}
    \begin{tikzpicture}[scale=1.2] 
	\tikzstyle{inti}=[draw=none,fill=none]

	\draw[->] (-3,0) -- (3,0) node[right] {\footnotesize $x$};
	\draw[->] (0,0) -- (0,1) node[above] {\footnotesize $f_1(x), \,(F_{\tau,\omega} f_1)(x)$};
	
		 \draw[thick,rounded corners=4pt]	
	(-3,.6)--(-2,0.3)--(-1.,0.5)--(-.5,0.4)--(0,0.4)--(.5,0.2)--(1.,0.5)--(2,0.3)--(3,.6);
			 \draw[thick,dashed, rounded corners=4pt]	
		(-3,.6)--(-2,0.3)--(-1.,0.5)--(-.5,0.45)--(-.25,0.5)--(0,0.5)--(.25,0.2)--(.5,.2)--(1.,0.5)--(2,0.3)--(3,.6);		 
	\end{tikzpicture}\\[1.5ex]
    \begin{tikzpicture}[scale=1.2] 
    
	\tikzstyle{inti}=[draw=none,fill=none]

	\draw[->] (-3,0) -- (3,0) node[right] {\footnotesize $x$};
	\draw[->] (0,0) -- (0,1) node[above] {\footnotesize $f_2(x), \,(F_{\tau,\omega} f_2)(x)$};
	
		 \draw[thick]	
	(-3,0)--(-.15,0)--(-.15,.5)--(.15,.5)--(.15,0)--(3,0);
		 \draw[thick,dashed]	
	(.35,0)--(.35,.5)--(.65,.5)--(.65,0);
	\end{tikzpicture}
\end{center}
\caption{Impact of the deformation $F_{\tau, \omega}$, with $\tau(x)=\frac{1}{2}\,e^{-x^2}$ and $\omega=0$, on the functions $f_1\in\mathcal{C}_1\subseteq L^2(\mathbb{R})$ and $f_2\in\mathcal{C}_2\subseteq L^2(\mathbb{R})$. The signal class $\mathcal{C}_1$ consists of smooth,  slowly varying functions (e.g., band-limited functions), and $\mathcal{C}_2$ consists of compactly supported functions that exhibit discontinuities (e.g., cartoon functions \cite{Cartoon}). We observe that $f_1$, unlike $f_2$, is affected only mildly by $F_{\tau,\omega}$. The amount of deformation induced  therefore depends drastically on the specific $f\in L^2(\mathbb{R})$.}
		\label{fig:sigclass}
\end{figure}
Our signal-class specific deformation sensitivity bound is based on the following two ingredients. First, we establish---in Proposition \ref{summary} in Appendix \ref{proof:nonexpan}---that the feature extractor $\Phi_\Omega$ is Lipschitz-continuous with Lipschitz constant $L_\Omega=1$, i.e.,
\begin{equation}\label{eq:p22ab}
||| \Phi_\Omega(f) -\Phi_\Omega(h)||| \leq \VV f-h \VV_2, \hspace{0.5cm} \forall f,h \in L^2(\Rd),
\end{equation}
where, thanks to the admissibility condition \eqref{weak_admiss2}, the Lipschitz constant $L_\Omega=1$ in \eqref{eq:p22ab} is completely independent of the frame upper bounds $B_n$ and the Lipschitz-constants $L_n$ and $R_n$ of $M_n$ and $P_n$, respectively. Second, we derive---in Proposition \ref{uppersuper} in Appendix \ref{proof:uppersuper}---an upper bound on the deformation error $\VV f-F_{\tau,\omega} f\VV_2$ for $R$-band-limited functions, i.e., $f \in L^2_R(\Rd)$, according to  
\begin{equation}\label{eq:p23ab}
\VV f-F_{\tau,\omega} f\VV_2 \leq C\big(R\VV \tau \VV_{\infty} + \VV \omega \VV_\infty \big)\VV f\VV_2.
\end{equation}
The  deformation sensitivity bound for the feature extractor is then obtained by setting  $h=F_{\tau,\omega}f$ in \eqref{eq:p22ab} and using  \eqref{eq:p23ab} (see Appendix \ref{Appendix} for the corresponding technical details). This ``decoupling''  into Lipschitz continuity of $\Phi_\Omega$ and a deformation sensitivity bound for the signal class  under consideration (here, band-limited functions) has important practical ramifications as it shows that whenever we have  a deformation sensitivity bound for the signal class, we automatically get a deformation sensitivity bound for the  feature extractor thanks to its Lipschitz continuity. The same approach was used in \cite{grohs_wiatowski} to derive deformation sensitivity bounds for cartoon functions and for Lipschitz functions. 

Lipschitz continuity of $\Phi_\Omega$  according to \eqref{eq:p22ab} also guarantees that pairwise distances in the input signal space do not increase through feature extraction. An immediate consequence  is  robustness of the feature extractor w.r.t. additive noise $\eta \in L^2(\Rd)$ in the sense of
$$||| \Phi_\Omega(f+\eta) -\Phi_\Omega(f)||| \leq  \| \eta\|_2, \quad \forall f \in L^2(\Rd).$$

We proceed to the formal statement of our deformation sensitivity result.

\begin{theorem}\label{mainmain}
 Let $\Omega=\big( (\Psi_n,M_n,P_n)\big)_{n \in \mathbb{N}}$ be an admissible module-sequence. There exists a constant $C>0$ (that does not depend on $\Omega$) such that for all $f \in L^2_R(\Rd)$,  $\omega \in C(\Rd,\RR)$, and  $\tau \in C^1(\Rd,\Rd)$ with $\VV D \tau \VV_\infty\leq\frac{1}{2d}$, the feature extractor $\Phi_\Omega$ satisfies
\begin{equation}\label{mainmainmain}
||| \Phi_\Omega(F_{\tau,\omega} f)-\Phi_\Omega(f) |||\leq C \big( R\VV \tau \VV_\infty + \VV \omega\VV_\infty\big) \VV f \VV_2.
\end{equation}
\end{theorem}
\begin{proof}
The proof is given in Appendix \ref{Appendix}. 
\end{proof}
First, we note that the bound in \eqref{mainmainmain} holds for $\tau$ with sufficiently ``small''  Jacobian matrix, i.e., as long as $\VV D \tau \VV_\infty\leq\frac{1}{2d}$. We can think of this condition on the Jacobian matrix as follows\footnote{The ensuing argument is taken from \cite{grohs_wiatowski}.}: Let $f$ be an image of the handwritten digit ``$5$'' (see Fig. \ref{fig:deform} (a)). Then, $\{ F_{\tau,\omega} f \ | \ \| D\tau \|_\infty<\frac{1}{2d}\}$ is a collection of images of the handwritten digit ``$5$'', where each $F_{\tau,\omega} f$ models an image that may be generated, e.g., based on a different handwriting style (see Figs. \ref{fig:deform} (b) and (c)). The condition $\| D\tau \|_\infty<\frac{1}{2d}$ now imposes a quantitative limit on the amount of deformation tolerated. The deformation sensitivity bound \eqref{mainmainmain} provides a limit on how much the features corresponding to the images in the set $\{ F_{\tau,\omega} f \ | \ \| D\tau \|_\infty<\frac{1}{2d}\}$ can differ. The strength of  Theorem \ref{mainmain} derives itself from the fact that the only condition on the underlying module-sequence $\Omega$ needed  is admissibility according to \eqref{weak_admiss2}, which as outlined in Section \ref{properties}, can easily be obtained by normalizing the frame elements of $\Psi_n$, for all $n\in \mathbb{N}$, appropriately. This normalization does not have an impact on the constant $C$ in \eqref{mainmainmain}. More specifically, $C$ is shown in \eqref{stab_const} to be completely independent of $\Omega$. All this is thanks to the decoupling technique used to prove Theorem \ref{mainmain} being completely independent of the structures of the frames $\Psi_n$ and of the specific forms of the Lipschitz-continuous operators $M_n$ and $P_n$. The deformation sensitivity bound \eqref{mainmainmain} is very general in the sense of applying to all Lipschitz-continuous (linear or non-linear) mappings $\Phi$, not only those generated by  DCNNs. 

The bound \eqref{mallatbound} for scattering networks reported in \cite[Theorem 2.12]{MallatS} depends upon first-order $ (D \tau) $ and second-order $ (D^2 \tau) $ derivatives of $\tau$. In contrast, our bound \eqref{mainmainmain} depends on $(D\tau)$ implicitly only as we need to impose the condition $\| D \tau \|_\infty\leq\frac{1}{2d}$ for the bound to hold\footnote{We note that the condition $\| D \tau \|_\infty\leq\frac{1}{2d}$ is needed for the bound \eqref{mallatbound} to hold as well.}. We honor this difference by referring to \eqref{mallatbound} as deformation \textit{stability} bound and to our bound \eqref{mainmainmain}  as deformation \textit{sensitivity} bound.

The dependence of the upper bound in \eqref{mainmainmain} on the bandwidth $R$ reflects the intuition that the deformation sensitivity bound should depend on the input signal class ``description complexity''. Many signals of practical significance (e.g., natural images) are, however, either not band-limited due to the presence of sharp (and possibly curved) edges or exhibit large bandwidths. In the latter case, the  bound \eqref{mainmainmain} is effectively rendered void owing to its linear dependence on $R$. We refer the reader to \cite{grohs_wiatowski} where deformation sensitivity bounds for non-smooth  signals were established. Specifically, the main  contributions in \cite{grohs_wiatowski} are deformation sensitivity bounds---again obtained through decoupling---for non-linear deformations $(F_\tau f)(x)=f(x-\tau(x))$ according to
\begin{equation}\label{particul}
\VV f-F_{\tau} f\VV_2 \leq C \VV \tau \VV_{\infty}^{\alpha}, \hspace{0.5cm} \forall f \in \mathcal{C}\subseteq L^2(\Rd),
\end{equation}
for the signal classes $\mathcal{C}\subseteq L^2(\Rd)$ of cartoon functions \cite{Cartoon} and for Lipschitz-continuous functions. The constant $C>0$ and the exponent $\alpha>0$ in \eqref{particul}  depend on the particular signal class $\mathcal{C}$ and are specified in \cite{grohs_wiatowski}. As the vertical translation invariance result in Theorem \ref{main_inv2} applies to all $f\in L^2(\Rd)$, the results established in the present paper and in \cite{grohs_wiatowski} taken together  show that vertical translation invariance and limited sensitivity to  deformations---for signal classes with inherent deformation insensitivity---are guaran\-teed by the feature extraction network structure per se rather than the specific convolution kernels, non-linearities, and pooling operators.

Finally, the deformation stability bound \eqref{mallatbound} for scattering networks reported in \cite[Theorem 2.12]{MallatS} applies to the space 
\begin{equation}\label{HMN}
H_{W}:=\Big\{ f \in L^2(\Rd) \ \Big| \ \VV f \VV_{H_{W}}<\infty \Big\},
\end{equation}
where 
$$
\ \VV f \VV_{H_{W}}:=\sum_{n=0}^\infty\Big(\sum_{q \in (\Lambda_{W})_1^n} \VV U[q]f\VV_2^2\Big)^{1/2}
$$
and $(\Lambda_{W})_1^n$ denotes the set of paths $q=\big(\lambda^{^{(j)}},\dots,\lambda^{^{(p)}}\big)$ of length $n$ with $\lambda^{^{(j)}},\dots,\lambda^{^{(p)}}\in \Lambda_{W}$. While \cite[p. 1350]{MallatS} cites numerical evidence on the series $\sum_{q \in (\Lambda_{W})_1^n} \VV U[q]f\VV_2^2$ being finite for a large class of signals $f \in L^2(\Rd)$, it seems difficult to establish this analytically, let alone to show that $$\sum_{n=0}^\infty\Big(\sum_{q \in (\Lambda_{W})_1^n} \VV U[q]f\VV_2^2\Big)^{1/2}<\infty.$$In contrast, the deformation sensitivity bound \eqref{mainmainmain} applies \textit{provably} to the space of $R$-band-limited functions $L^2_R(\Rd)$. Finally, the space $H_W$ in \eqref{HMN} depends on the wavelet frame atoms $\{ \psi_\lambda\}_{\lambda \in \Lambda_{W} }$ and the (modulus) non-linearity, and thereby on the underlying signal transform, whereas $L^2_R(\Rd)$ is, trivially, independent of the module-sequence $\Omega$.
\section{Final remarks and outlook}
 It is interesting to note that the frame lower bounds $A_n>0$ of the semi-discrete frames $\Psi_n$ affect neither the vertical translation invariance result in Theorem \ref{main_inv2} nor the deformation sensitivity bound in Theorem \ref{mainmain}. In fact, the entire theory in this paper carries through as long as the collections $\Psi_{n}=\{T_b I g_{\lambda_n} \}_{b\in \Rd,\la_n \in \Lambda_n}$, for all $n \in \mathbb{N}$, satisfy the Bessel property 
 \begin{equation*}\label{bessel}
\sum_{\lambda_n\in \Lambda_n}\int_{\Rd}| \langle f,T_bIg_{\lambda_n}\rangle|^2\mathrm db  = \sum_{\lambda_n\in \Lambda_n} \VV f \ast g_{\lambda_n} \VV_2^2 \leq B_n \VV f \VV_2^2, 
 \end{equation*}
 for all $f \in L^2(\Rd)$ for some $B_n>0$, which, by Proposition \ref{freqcoverthm2}, is equivalent to \begin{equation}\label{realman}\sum_{\lambda_n \in \Lambda_n} |\widehat{g_{\lambda_n}}(\omega)|^2\leq B_n,\hspace{0.5cm} a.e. \ \omega \in \Rd.\end{equation}Pre-specified unstructured filters  \cite{Jarrett,hierachies} and learned filters \cite{Huang,Jarrett,hierachies,Poultney} are therefore covered by our theory as long as \eqref{realman} is satisfied. In classical frame theory $A_n>0$ guarantees completeness of the set $\Psi_{n}=\{T_b I g_{\lambda_n} \}_{b\in \Rd,\la_n \in \Lambda_n}$ for the signal space under consideration, here $L^2(\Rd)$. The absence of a frame lower bound $A_n>0$ therefore translates into a lack of completeness of $\Psi_n$, which may result in the frame coefficients $\langle f,T_bIg_{\lambda_n}\rangle=(f\ast g_{\lambda_n})(b)$, $(\lambda_n,b) \in \Lambda_n\times \Rd$, not containing all essential features of the signal $f$. This will, in general, have a (possibly significant) impact on practical feature extraction performance which is why ensuring the entire frame property \eqref{condii} is prudent. Interestingly, satisfying the frame property \eqref{condii} for all $\Psi_n$, $n\in \mathbb{Z}$, does, however, not guarantee that the feature extractor $\Phi_\Omega$ has a trivial null-space, i.e., $\Phi_\Omega(f)= 0$ if and only if $f=0$. We refer the reader to \cite[Appendix A]{WiatowskiEnergy} for an  example of a feature extractor with non-trivial null-space.

\appendices
\section{Semi-discrete frames}\label{sec:sdf}
This appendix gives a brief review of the theory of semi-discrete frames. A list of structured example frames of interest in the context of this paper is provided in Appendix \ref{sec:sdf0} for the $1$-D case, and in Appendix \ref{sec:sdf2} for the $2$-D case. Semi-discrete frames are instances of \textit{conti\-nuous} frames \cite{Antoine,Kaiser}, and appear in the literature, e.g., in the context of translation-covariant signal decompositions \cite{MallatEdges,UnserWavelets,Vandergheynst}, and as an intermediate step in the construction of various \textit{fully-discrete} frames  \cite{CandesDonoho2,ShearletsIntro,Grohs_transport,Grohs_alpha}. We first collect some basic results on semi-discrete frames. 
\begin{definition}\label{defn:contframe}
Let  $\{ g_\la \}_{\la \in \Lambda}\subseteq L^1(\Rd)\cap L^2(\Rd)$  be a set of functions indexed by a countable set $\Lambda$. The collection $$\Psi_\Lambda:=\{ T_b Ig_\la\}_{(\la,b) \in \Lambda \times \Rd}$$ is a semi-discrete frame for $L^2(\Rd)$ if there exist constants $A,B >0$ such that 
\begin{align}
&A \VV f \VV_2^2 \leq \sum_{\la \in \Lambda}\int_{\Rd} |\langle f, T_bIg_\lambda \rangle|^2  \mathrm db \nonumber\\
&= \sum_{\lambda\in \Lambda} \VV f \ast g_{\lambda} \VV_2^2 \leq B \VV f \VV_2^2, \hspace{0.5cm} \forall f \in L^2(\Rd).\label{condii}
\end{align}
The functions $\{ g_\la \}_{\la \in \Lambda}$ are called the atoms of the frame $\Psi_\Lambda$. When $A=B$ the frame is said to be tight. A tight frame with frame bound $A=1$ is called a Parseval frame.
\end{definition}
The frame operator associated with the semi-discrete frame $\Psi_\Lambda$ is defined in the weak sense as $S_{\Lambda}:L^2(\Rd) \to L^2(\Rd)$,
\begin{align}
S_{\Lambda}f:=&\sum_{\la \in \Lambda}\int_{\Rd} \langle f, T_bIg_\lambda \rangle (T_bIg_{\lambda}) \, \mathrm db\nonumber\\
=& \Big( \sum_{\la \in \Lambda} g_\la \ast I g_\la\Big)\ast f,\label{eq:sdfo}
\end{align}
where $\langle f, T_bIg_\lambda \rangle=(f\ast g_\lambda)(b),$ $ (\lambda,b)\in \Lambda\times\Rd,$ are called the frame coefficients.
$S_\Lambda$ is a bounded, positive, and boundedly invertible operator \cite{Antoine}.

The reader might want to think of semi-discrete frames as shift-invariant frames \cite{Jansen,RonShen} with a continuous translation parameter, and of the countable index set $\Lambda$ as labeling a collection of scales, directions, or frequency-shifts, hence the terminology \textit{semi-discrete}. 
For instance, scattering networks are based on a (single) semi-discrete wavelet frame, where the atoms $\{ g_\la\}_{\la \in \Lambda_{\text{W}}}$ are indexed by the set $\Lambda_{\text{W}}:=\big\{ (-J,0)\big\}\cup\big\{ (j,k) \ | \ j \in \mathbb{Z} \text{ with }j>-J, \ k \in \{ 0,\dots,K-1\} \big\}$ labeling a collection of scales $j$ and directions $k$.

The following result gives a so-called Littlewood-Paley condition \cite{Paley,Daubechies} for the collection $\Psi_\Lambda= \{ T_b Ig_\la\}_{(\la,b) \in \Lambda\times \Rd}$ to form a semi-discrete frame.
\begin{proposition}\label{freqcoverthm2}
Let $\Lambda$ be a countable set. The collection $\Psi_\Lambda= \{ T_b Ig_\la\}_{(\la,b) \in \Lambda\times \Rd}$ with atoms  $\{ g_\la \}_{\la \in \Lambda}\subseteq L^1(\Rd) \cap L^2(\Rd)$ is a semi-discrete frame for $L^2(\Rd)$ with frame bounds $A,B>0$ if and only if
\begin{equation}\label{freqcover}
A\leq \sum_{\lambda \in \Lambda} |\widehat{g_{\lambda}}(\omega)|^2\leq B, \hspace{0.5cm} a.e. \ \omega \in \Rd.
\end{equation}
\end{proposition}
\begin{proof}
The proof is standard and can be found, e.g., in \cite[Theorem 5.11]{MallatW}. 
\end{proof}
\begin{remark}
What is behind Proposition \ref{freqcoverthm2} is a result on the unitary equivalence between operators \cite[Definition 5.19.3]{Sell}. Specifically, Proposition \ref{freqcoverthm2} follows from the fact  that the multiplier $\sum_{\lambda \in \Lambda} |\widehat{g_{\lambda}}|^2$ is unitarily equivalent to the frame operator $S_\Lambda$ in \eqref{eq:sdfo} according to  
$$
\mathcal{F} S_\Lambda \mathcal{F}^{\,-1}=\sum_{\lambda \in \Lambda} |\widehat{g_{\lambda}}|^2,
$$
where $\mathcal{F}:L^2(\Rd) \to L^2(\Rd)$ denotes the Fourier transform. 
We refer the interested reader to \cite{HelmutTSP} where the framework of unitary equivalence was formalized in the context of shift-invariant frames for $\ell^2(\mathbb{Z})$.
\end{remark}
  \begin{figure*}
\begin{center}
\begin{tikzpicture}


	\begin{scope}[scale=.8]
	
			\pgfsetfillopacity{.7}
	\filldraw[fill=black!5!white, draw=black]
   		(-10.5,-0.5)--(-10.5,0.5)--(-9.5,0.5)--(-9.5,-0.5)--(-10.5,-0.5);
	
	\filldraw[fill=black!30!white, draw=black,xshift=-10cm]
   		(0.5,0.5)--(0.5,1)--(1,1)--(1,0.5)--(0.5,0.5);
	\filldraw[fill=black!30!white, draw=black,xshift=-10cm]
   		(-0.5,-0.5)--(-0.5,-1)--(-1,-1)--(-1,-0.5)--(-0.5,-0.5);
	\filldraw[fill=black!30!white, draw=black,xshift=-10cm]
   		(0.5,-0.5)--(1,-0.5)--(1,-1)--(.5,-1)--(0.5,-0.5);
	\filldraw[fill=black!30!white, draw=black,xshift=-10cm]
   		(-0.5,0.5)--(-0.5,1)--(-1,1)--(-1,0.5)--(-0.5,0.5);
	
	\filldraw[fill=black!40!white, draw=black,xshift=-10cm]
   		(1,1)--(1,2)--(2,2)--(2,1)--(1,1);
	\filldraw[fill=black!40!white, draw=black,xshift=-10cm]
   		(-1,-1)--(-1,-2)--(-2,-2)--(-2,-1)--(-1,-1);
	\filldraw[fill=black!40!white, draw=black,xshift=-10cm]
   		(1,-1)--(2,-1)--(2,-2)--(1,-2)--(1,-2);
	\filldraw[fill=black!40!white, draw=black,xshift=-10cm]
   		(-1,1)--(-1,2)--(-2,2)--(-2,1)--(-1,1);
	
	\filldraw[fill=black!15!white, draw=black,xshift=-10cm]
   		(.5,.5)--(1,.5)--(1,-.5)--(.5,-.5)--(.5,.5);
	\filldraw[fill=black!25!white, draw=black,xshift=-10cm]
   		(.5,-.5)--(.5,-1)--(-.5,-1)--(-.5,-.5)--(.5,-.5);
	\filldraw[fill=black!15!white, draw=black,xshift=-10cm]
   		(-.5,-.5)--(-.5,.5)--(-1,.5)--(-1,-.5)--(-.5,-.5);
	\filldraw[fill=black!25!white, draw=black,xshift=-10cm]
   		(.5,.5)--(.5,1)--(-.5,1)--(-.5,0.5)--(.5,.5);
	
	\filldraw[fill=black!45!white, draw=black,xshift=-10cm]
   		(1,1)--(2,1)--(2,-1)--(1,-1)--(1,1);
	\filldraw[fill=black!55!white, draw=black,xshift=-10cm]
   		(1,-1)--(1,-2)--(-1,-2)--(-1,-1)--(1,-1);
	\filldraw[fill=black!45!white, draw=black,xshift=-10cm]
   		(-1,-1)--(-1,1)--(-2,1)--(-2,-1)--(-1,-1);
	\filldraw[fill=black!55!white, draw=black,xshift=-10cm]
   		(1,1)--(1,2)--(-1,2)--(-1,1)--(1,1);
		
	\draw[->,dashed] (-12.5,0) -- (-7.5,0) node[right] {$\omega_1$};
	\draw[->,dashed] (-10,-2.5) -- (-10,2.5) node[above] {$\omega_2$};


	\filldraw[fill=black!5!white, draw=black]
   		  (0,0) circle (.5cm);
	\pgfsetfillopacity{.7}
	\filldraw[fill=black!20!white, draw=black]
   		(.75,0) circle (.25cm);
	\filldraw[fill=black!20!white, draw=black]
   		(0.6495,.375) circle (.25cm);
	\filldraw[fill=black!20!white, draw=black]
   		(0.375,0.6495) circle (.25cm);
	\filldraw[fill=black!20!white, draw=black]
   		(0,.75) circle (.25cm);	
	\filldraw[fill=black!20!white, draw=black]
   		(-0.375,0.6495) circle (.25cm);
	\filldraw[fill=black!20!white, draw=black]
   		(-0.6495,0.375) circle (.25cm);
	\filldraw[fill=black!20!white, draw=black]
   		(-.75,0) circle (.25cm);	
	\filldraw[fill=black!20!white, draw=black]
   		(-0.6495,-.375) circle (.25cm);
	\filldraw[fill=black!20!white, draw=black]
   		(-0.375,-0.6495) circle (.25cm);
	\filldraw[fill=black!20!white, draw=black]
  		(0,-.75) circle (.25cm);	
	\filldraw[fill=black!20!white, draw=black]
   		(0.375,-0.6495) circle (.25cm);
	\filldraw[fill=black!20!white, draw=black]
   		(0.6495,-.375) circle (.25cm);
	\filldraw[fill=black!45!white, draw=black]
   		(1.5,0) circle (.5cm);
	\filldraw[fill=black!45!white, draw=black]
   		(1.299,0.75) circle (.5cm);
	\filldraw[fill=black!45!white, draw=black]
   		(.75,1.299) circle (.5cm);
	\filldraw[fill=black!45!white, draw=black]
   		(0,1.5) circle (.5cm);
	\filldraw[fill=black!45!white, draw=black]
   		(-.75,1.2990) circle (.5cm);
	\filldraw[fill=black!45!white, draw=black]
   		(-1.299,.75) circle (.5cm);
	\filldraw[fill=black!45!white, draw=black]
   		(-1.5,0) circle (.5cm);
	\filldraw[fill=black!45!white, draw=black]
   		(-1.299,-.75) circle (.5cm);	
	\filldraw[fill=black!45!white, draw=black]
   		(-.75,-1.299) circle (.5cm);	
	\filldraw[fill=black!45!white, draw=black]
   		(0,-1.5) circle (.5cm);	
	\filldraw[fill=black!45!white, draw=black]
   		(.75,-1.299) circle (.5cm);
	\filldraw[fill=black!45!white, draw=black]
   		(1.299,-.75) circle (.5cm);
	\draw[->,dashed] (-2.5,0) -- (2.5,0) node[right] {$\omega_1$} ;
	\draw[->,dashed] (0,-2.5) -- (0,2.5)node[above] {$\omega_2$};

		\end{scope}
\end{tikzpicture}
\end{center}
\caption{Partitioning of the frequency plane $\RR^2$ induced by (left) a semi-discrete tensor wavelet frame, and (right) a semi-discrete directional wavelet frame.}
\end{figure*}
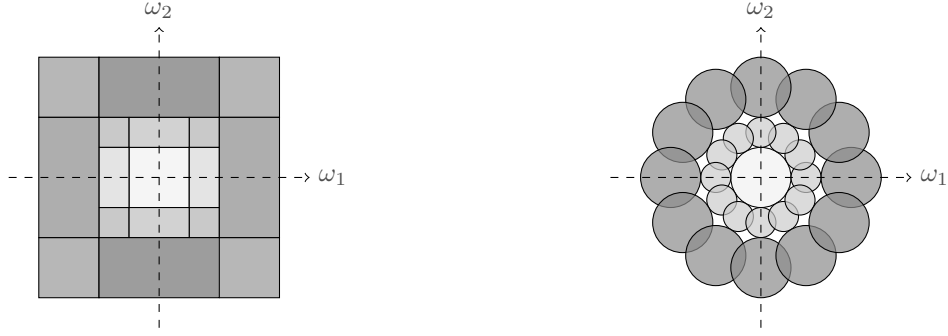
The following proposition states normalization results for semi-discrete frames that come in handy in satisfying the admissibility condition \eqref{weak_admiss2} as discussed in Section \ref{properties}. 
\begin{proposition}\label{prop:norm}
Let $\Psi_\Lambda=\{ T_b Ig_\la\}_{(\la,b) \in \Lambda \times \Rd}$ be a semi-discrete frame for $L^2(\Rd)$ with frame bounds $A,B$. 
\begin{itemize}
\item[i)]{For $C>0$, the family of functions $\widetilde{\Psi}_\Lambda:=\big\{ T_b I\widetilde{g_\la}\big\}_{(\la,b) \in \Lambda \times \Rd}$, $\widetilde{g_\la}:=C^{-1/2}g_\la, \  \forall \la \in \Lambda,$  is a semi-discrete frame for $L^2(\Rd)$ with frame bounds $\widetilde{A}:=\frac{A}{C}$ and $\widetilde{B}:=\frac{B}{C}$.}
\item[ii)]{The family of functions $\Psi^\natural_\Lambda:=\big\{ T_b I g^\natural_\la\big\}_{(\la,b) \in \Lambda \times \Rd}$, $$
g_\la^\natural:=\mathcal{F}^{-1}\Big(\widehat{g_\la}\Big(\sum_{\lambda' \in \Lambda} |\widehat{g_{\lambda'}}|^2\Big)^{-1/2}\Big), \  \forall \la \in \Lambda,$$ is a semi-discrete Parseval frame for $L^2(\Rd)$, i.e., the frame bounds satisfy $A^\natural=B^\natural=1$.}
\end{itemize}
\end{proposition}
\begin{proof}
We start by proving statement i). As $\Psi_\Lambda$ is a frame for $L^2(\Rd)$, we have
\begin{equation}\label{anakonda}
A \VV f \VV_2^2 \leq \sum_{\lambda\in \Lambda} \VV f \ast g_{\lambda} \VV_2^2 \leq B \VV f \VV_2^2, \hspace{0.5cm} \forall f \in L^2(\Rd).
\end{equation} With $g_\la=\sqrt{C}\,\widetilde{g_\la}$, for all $\lambda \in \Lambda$, in \eqref{anakonda} we get  
$A \VV f \VV_2^2 \leq \sum_{\lambda\in \Lambda} \VV f \ast \sqrt{C}\, \widetilde{g_{\lambda}} \VV_2^2 \leq B \VV f \VV_2^2$, for all $ f \in L^2(\Rd)$,
which is equivalent to $\frac{A}{C} \VV f \VV_2^2 \leq \sum_{\lambda\in \Lambda} \VV f \ast \widetilde{g_{\lambda}} \VV_2^2 \leq \frac{B}{C} \VV f \VV_2^2$, for all $f \in L^2(\Rd)$, and hence establishes i). To prove statement ii), we first note that $\mathcal{F}g_\la^\natural=\widehat{g_\la}\big(\sum_{\lambda' \in \Lambda} |\widehat{g_{\lambda'}}|^2\big)^{-1/2}$, for all $ \lambda \in \Lambda$, and thus $\sum_{\la \in \Lambda} |(\mathcal{F}g_\la^\natural)(\omega)|^2= \sum_{\la \in \Lambda} |\widehat{g_\la}(\omega)|^2\Big(\sum_{\lambda' \in \Lambda} |\widehat{g_{\lambda'}}(\omega)|^2\Big)^{-1}=1$, a.e. $\omega \in \Rd. $
Application of Proposition \ref{freqcoverthm2} then establishes that $\Psi^\natural_\Lambda$ is a semi-discrete Parseval frame for $L^2(\Rd)$, i.e., the frame bounds satisfy $A^\natural=B^\natural=1$. 
\end{proof}

\section{Examples of semi-discrete frames in \normalfont $1$-D}\label{sec:sdf0}

General $1$-D semi-discrete frames are given by collections
\begin{equation}\label{SDF_FILTER1}
\Psi=\{ T_b Ig_k\}_{(k,b) \in \mathbb{Z} \times \RR}
\end{equation}
with atoms $g_k\in  L^1(\RR)\cap L^2(\RR)$, indexed by the integers $\Lambda=\mathbb{Z}$, and satisfying the Littlewood-Paley condition 
\begin{equation}\label{eq:filterbankcondi}
 A\leq \sum_{k \in \mathbb{Z}} |\widehat{g_k}(\omega)|^2\leq B,\hspace{0.5cm} a.e. \ \omega \in \RR.
 \end{equation}
The structural example frames we consider are Weyl-Heisenberg (Gabor) frames where the $g_k$ are obtained through modulation from a prototype function, and wavelet frames where the $g_k$ are obtained through scaling from a mother wavelet.\\[1ex]
\textit{Semi-discrete Weyl-Heisenberg (Gabor) frames:} Weyl-Heisenberg fra\-mes \cite{JanssenDual,DaubechiesLandau,Groechenig,painless} are well-suited to the extraction of sinusoidal features \cite{localFourier}, and have been applied successfully in various practical feature extraction tasks\cite{GaborFeature1,GaborFeatures2}. A semi-discrete Weyl-Heisenberg frame for $L^2(\RR)$ is a collection of functions accor\-ding to \eqref{SDF_FILTER1}, where $g_{k}(x):=e^{2\pi i k x}g(x)$, $k \in \mathbb{Z}$, with  the prototype function $g \in L^1(\RR)\cap L^2(\RR)$. The atoms $\{ g_{k}\}_{k \in \mathbb{Z}}$ satisfy the Littlewood-Paley condition \eqref{eq:filterbankcondi} according to 
\begin{equation}\label{eq:Weyl-Heisenberg (Gabor)condi}
 A\leq \sum_{k \in \mathbb{Z}} |\widehat{g}(\omega-k)|^2\leq B,\hspace{0.5cm} a.e. \ \omega \in \RR.
 \end{equation}
A popular function $g \in L^1(\RR)\cap L^2(\RR)$ satisfying \eqref{eq:Weyl-Heisenberg (Gabor)condi} is the Gaussian function \cite{Groechenig}.\\[1ex]
\textit{Semi-discrete wavelet frames:} Wavelets are well-suited to the  extraction of signal features chara\-cte\-rized by singularities \cite{Daubechies,MallatEdges}, and have been applied successfully in various practical feature extraction tasks \cite{Wavelet1FeaturesD,Wavelet1Features2}.  A semi-discrete wavelet frame for $L^2(\RR)$ is a collection of functions accor\-ding to \eqref{SDF_FILTER1}, where $g_{k}(x):=2^k \psi(2^kx)$, $ k \in \mathbb{Z}$, with the mother wavelet $\psi \in L^1(\RR)\cap L^2(\RR)$. The atoms $\{ g_{k}\}_{k \in \mathbb{Z}}$ satisfy the Littlewood-Paley condition \eqref{eq:filterbankcondi} according to  
\begin{equation}\label{eq:Waveletcondi}
 A\leq \sum_{k \in \mathbb{Z}} |\widehat{\psi}(2^{-k}\omega)|^2\leq B,\hspace{0.5cm} a.e. \ \omega \in \RR.
 \end{equation}
A large class of functions $\psi$ satisfying \eqref{eq:Waveletcondi} can be obtained through a multi-resolution analysis in $L^2(\RR)$ \cite[Definition 7.1]{MallatW}. 

\begin{figure*}[t]
\begin{center}
\begin{tikzpicture}
	\begin{scope}[scale=.215]

	\filldraw[fill=black!27!white, draw=black]
   		  (21,0) circle (8cm);
	\filldraw[fill=black!20!white, draw=black]
   		  (21,0) circle (4cm);
	\filldraw[fill=black!10!white, draw=black]
   		  (21,0) circle (2cm);
	\filldraw[fill=black!10!white, draw=black]
   		  (21,0) circle (1cm);
	\filldraw[fill=black!1!white, draw=black]
   		  (21,0) circle (.5cm);
		\fill[fill=black!50!white,xshift=21cm]
 (2.5376  ,  3.0920)--(5.0751    ,6.1841)  arc (50.6250:39.3750:8.0cm) --  (3.0920 ,   2.5376) arc(39.3750:50.625:4cm); 
\fill[fill=black!50!white,xshift=21cm]
 (-2.5376  ,  -3.0920)--(-5.0751    ,-6.1841)  arc (230.6250:219.3750:8.0cm) --  (-3.0920 ,  - 2.5376) arc(219.3750:230.625:4cm); 
		
	\draw (21.3536, 0.3536) -- (21.7071,0.7071);	
	\draw (21.3536, -0.3536) -- (21.7071,-0.7071);	
	\draw (20.6464,0.3536) -- 	(20.2929,0.7071);
	\draw (20.6464,-0.3536) -- (20.2929,-0.7071);	
	\draw (21.9239,0.3827) -- (22.8478,0.7654);
	\draw (21.9239,-0.3827) -- (22.8478,-0.7654);		
	\draw (21.3827,0.9239) -- (21.7654,1.8478);	
	\draw (21.3827,-0.9239) -- (21.7654,-1.8478);	
	\draw (20.6173,0.9239) -- (20.2346,1.8478);
	\draw (20.6173,-0.9239) -- (20.2346,-1.8478);	
	\draw (19.1522,0.7654) -- (20.0761,0.3827);
	\draw (19.1522,-0.7654) -- (20.0761,-0.3827);		
	\draw(22.9616,0.3902) -- (24.9231,0.7804);
	\draw(22.6629,1.111) -- (24.3259,2.2223);	
	\draw(23.2223,3.3259) -- (22.1111,1.6629);
	\draw(21.39002,1.9616) -- (21.7804,3.9231);
	\draw(20.6098,1.9616) -- (20.2196,3.9231);
	\draw(19.8889,1.6629) -- (18.7777,3.325);	
	\draw(17.6741,2.2233) -- (19.3371,1.111);	
	\draw(19.0384,0.3902) -- (17.0769,0.7804);
	\draw(22.9616,-0.3902) -- (24.9231,-0.7804);
	\draw(22.6629,-1.111) -- (24.3259,-2.2223);	
	\draw(23.2223,-3.3259) -- (22.1111,-1.6629);
	\draw(21.39002,-1.9616) -- (21.7804,-3.9231);
	\draw(20.6098,-1.9616) -- (20.2196,-3.9231);
	\draw(19.8889,-1.6629) -- (18.7777,-3.325);	
	\draw(17.6741,-2.2233) -- (19.3371,-1.111);	
	\draw(19.0384,-0.3902) -- (17.0769,-0.7804);									
	\draw(24.9807,0.3921) -- (28.9615,0.7841);
	\draw(24.8278,1.1611) -- (28.6555,2.3223);		
	\draw(24.5277,1.8856) -- (28.0554,3.7712);
	\draw(27.1841,5.0751) -- (24.0920,2.5376);			
	\draw(26.0751,6.1841) -- (23.5376,3.0920);
	\draw(22.8856,3.5277) -- (24.7712,7.0554);
	\draw(22.1611,3.8278) -- (23.3223,7.6555);
	\draw(21.7841,7.9615) -- (21.3921,3.9807);
	\draw(20.6079,3.9807) -- (20.2159,7.9615);
	\draw(18.6777,7.6555) -- (19.8389,3.8278);						
	\draw(19.1144,3.5277) -- (17.2288,7.0554);
	\draw(15.9249,6.1841) -- (18.4624,3.0920);
	\draw(17.9080,2.5376) -- (14.8159,5.0751);
	\draw(13.9446,3.7712) -- (17.4723,1.8856);
	\draw(17.1722,1.16111) -- (13.3445,2.3223);
	\draw(13.0385,0.7841) -- (17.0193,0.3921);						
	\draw(24.9807,-0.3921) -- (28.9615,-0.7841);
	\draw(24.8278,-1.1611) -- (28.6555,-2.3223);		
	\draw(24.5277,-1.8856) -- (28.0554,-3.7712);
	\draw(27.1841,-5.0751) -- (24.0920,-2.5376);			
	\draw(26.0751,-6.1841) -- (23.5376,-3.0920);
	\draw(22.8856,-3.5277) -- (24.7712,-7.0554);
	\draw(22.1611,-3.8278) -- (23.3223,-7.6555);
	\draw(21.7841,-7.9615) -- (21.3921,-3.9807);
	\draw(20.6079,-3.9807) -- (20.2159,-7.9615);
	\draw(18.6777,-7.6555) -- (19.8389,-3.8278);						
	\draw(19.1144,-3.5277) -- (17.2288,-7.0554);
	\draw(15.9249,-6.1841) -- (18.4624,-3.0920);
	\draw(17.9080,-2.5376) -- (14.8159,-5.0751);
	\draw(13.9446,-3.7712) -- (17.4723,-1.8856);
	\draw(17.1722,-1.16111) -- (13.3445,-2.3223);
	\draw(13.0385,-0.7841) -- (17.0193,-0.3921);
	\draw[->,dashed] (12.2,0) -- (30.2,0) node[right] {$\omega_1$};
	\draw[->,dashed] (21,-9.2) -- (21,9.2)node[above] {$\omega_2$};

		\draw[->,dashed] (-26.2,0) -- (-8.2,0) node[right] {$\omega_1$};
	\draw[->,dashed] (-17,-9.2) -- (-17,9.2)node[above] {$\omega_2$};

	\filldraw[fill=black!27!white, draw=black]
   		  (-17,0) circle (8cm);
	\filldraw[fill=black!20!white, draw=black]
   		  (-17,0) circle (4cm);
	\filldraw[fill=black!15!white, draw=black]
   		  (-17,0) circle (2cm);
	\filldraw[fill=black!10!white, draw=black]
   		  (-17,0) circle (1cm);
	\filldraw[fill=black!1!white, draw=black]
   		  (-17,0) circle (.5cm);

	\draw (-16.5,0) -- (-9,0);
	\draw (-17,0.5) -- (-17,8);
	\draw (-17.5,0) -- (-25,0);
	\draw (-17,-.5) -- (-17, -8);

\filldraw[fill=black!50!white,draw=black!80!white]
(-13.3045,1.5307) -- (-9.609,3.0615) arc (22.5:45:8cm) -- (-14.1716,  2.8284) arc (45:22.5:4cm);

\filldraw[fill=black!50!white,draw=black!80!white]
(-20.6955,-1.5307) -- (-24.391,-3.0615) arc (202.5:225:8cm)--(-19.8284, - 2.8284) arc(225:202.5:4cm);

	\draw (-16.2929,0.7071) -- (-11.3431,5.6569);
	\draw (-17.7071,0.7071) -- (-22.6569,5.6569);
	\draw (-17.7071,-0.7071) -- (-22.6569,-5.6569);
	\draw (-16.2929,-0.7071) -- (-11.3431,-5.6569);
	
	\draw (-13.3045,1.5307) -- (-9.609,3.0615);
	\draw (-20.6955,-1.5307) -- (-24.3910,-3.0615);
	\draw (-13.3045,-1.5307) -- (-9.609,-3.0615);
	\draw (-20.6955,1.5307) -- (-24.3910,3.0615);
	\draw (-15.4693,3.6955) -- (-13.9385,7.3910);
	\draw (-15.4693,-3.6955) -- (-13.9385,-7.3910);
	\draw (-18.5307,3.6955) -- (-20.0615,7.3910);
	\draw (-18.5307,-3.6955) -- (-20.0615,-7.3910);

		\end{scope}
\end{tikzpicture}
\end{center}
\caption{Partitioning of the frequency plane $\RR^2$ induced by (left) a semi-discrete curvelet frame, and (right) a semi-discrete ridgelet frame.}
		\label{fig:acurvea}
\end{figure*}
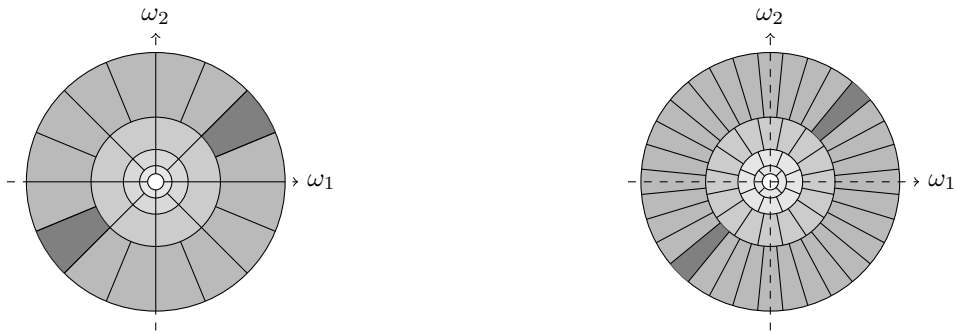
\textit{Semi-discrete curvelet frames:} 
Curvelets, introduced in \cite{CandesNewTight,CandesDonoho2}, are well-suited to the extraction of signal features characte\-rized by curve-like singularities (such as, e.g., curved edges in images), and have been applied successfully in various practical feature extraction tasks \cite{Plonka,Dettori}. 
\section{Examples of semi-discrete frames in \normalfont $2$-D}\label{sec:sdf2}
\textit{Semi-discrete wavelet frames:} Two-dimensional wavelets are well-suited to the extraction of signal features characte\-rized by point singularities (such as, e.g., stars in astronomical images \cite{GittaDonoho}), and have been applied successfully in various practical feature extraction tasks, e.g., in \cite{UnserWavelets,Serre,GaborLowe,Pinto}. Prominent families of two-dimensional wavelet frames are tensor wavelet frames and directional wavelet frames: 
 \begin{enumerate}
\item[i)]{\textit{Semi-discrete tensor wavelet frames:} A semi-discrete tensor wavelet frame for $L^2(\RR^2)$ is a collection of functions according to $\Psi_{\Lambda_{\text{TW}}}:=\{ T_bIg_{(e,j)}\}_{(e,j) \in \Lambda_{\text{TW}}, b \in \RR^2},$ $g_{(e,j)}(x):=2^{2j}\psi^e(2^jx),$
 where $\Lambda_{\text{TW}}:= \big\{ ((0,0),0)\big\}\cup\big\{ (e,j) \ | \ e \in E\backslash \{ (0,0)\},\ j\geq0 \big\}$, and $E:=\{ 0,1\}^2$. Here, the functions $\psi^e \in L^1(\RR^2)\cap L^2(\RR^2)$ are tensor products of a coarse-scale function $\phi \in L^1(\RR)\cap L^2(\RR)$ and a fine-scale function $\psi \in L^1(\RR)\cap L^2(\RR)$ according to $\psi^{(0,0)}:=\phi \otimes \phi$, $\psi^{(1,0)}:=\psi \otimes \phi,$ $\psi^{(0,1)}:=\phi \otimes \psi,$ and $\psi^{(1,1)}:=\psi \otimes \psi.$ The corresponding Littlewood-Paley condition \eqref{freqcover} reads
\begin{align}
A&\leq \big|\widehat{\psi^{(0,0)}}(\omega)\big|^2\nonumber\\
&+\sum_{j\geq 0}\sum_{e\in E\backslash \{ (0,0)\} } |\widehat{\psi^e}(2^{-j}\omega)|^2\leq B,\label{eq:tensor1} 
\end{align}
a.e. $\omega \in \RR^2.$ A large class of functions $\phi,\psi$ satisfying \eqref{eq:tensor1} can be obtained through a multi-resolution analysis in $L^2(\RR)$ \cite[Definition 7.1]{MallatW}. }
\item[ii)]{
\textit{Semi-discrete directional wavelet frames:} A semi-discrete directional wavelet frame for $L^2(\RR^2)$ is a collection of functions according to 
\begin{equation*}\label{directionali}
\Psi_{\Lambda_{\text{DW}}}:=\{ T_bIg_{(j,k)}\}_{(j,k) \in \Lambda_{\text{DW}}, b \in \RR^2},\end{equation*}
with $g_{(-J,0)}(x):=2^{-2J}\phi(2^{-J}x),$ $g_{(j,k)}(x):=2^{2j}\psi(2^jR_{\theta_k}x),$ 
where $\Lambda_{\text{DW}}:=\big\{ (-J,0)\big\}\cup \big\{ (j,k) \ | \ j\in \mathbb{Z} \text{ with } j> -J, \ k\in \{0,\dots,K-1\} \big\}$, $R_{\theta}$ is a $2\times 2$  rotation matrix defined as 
\begin{equation}\label{rotationmatrix}
 R_\theta:=\begin{pmatrix}
 \cos(\theta) & -\sin(\theta)\\
 \sin(\theta) & \cos(\theta)
 \end{pmatrix},\hspace{0.5cm} \theta \in [0,2\pi),
\end{equation}
 and $\theta_k:=\frac{2\pi k}{K}$, with $k=0,\dots,K-1$, for a fixed $K\in \mathbb{N}$, are rotation angles. The functions $\phi \in L^1(\RR^2)\cap L^2(\RR^2)$ and $\psi \in L^1(\RR^2)\cap L^2(\RR^2)$ are referred to in the literature as coarse-scale wavelet and fine-scale wavelet, respectively. The integer $J \in \mathbb{Z}$ corresponds to the coarsest scale resolved and the atoms $\{g_{(j,k)}\}_{(j,k)\in \Lambda_{\text{DW}}}$ satisfy the Littlewood-Paley condition \eqref{freqcover} according to
\begin{equation}\label{eq:direc1}
A\leq |\widehat{\phi}(2^J\omega)|^2 + \sum_{j>-J}\sum_{k=0}^{K-1} |\widehat{\psi}(2^{-j}R_{\theta_k} \omega )|^2\leq B, \end{equation}
a.e. $\omega \in \RR^2.$ Prominent examples of functions $\phi,\psi$ satisfying \eqref{eq:direc1} are the Gaussian function for $\phi$ and a modulated Gaussian function for $\psi$ \cite{MallatW}.}
\end{enumerate}

A semi-discrete curvelet frame for $L^2(\RR^2)$ is a collection of functions according to $\Psi_{\Lambda_{\text{C}}}:=\{ T_bIg_{(j,l)}\}_{(j,l) \in \Lambda_{\text{C}}, b \in \RR^2},$
with $g_{(-1,0)}(x):=\phi(x),$ $g_{(j,l)}(x):=\psi_{j}(R_{\theta_{j,l}}x),$ where $\Lambda_{\text{C}}:=\big\{ (-1,0)\big\}\cup \big\{ (j,l) \ | \ j\geq 0, \ l=0,\dots,L_j-1 \big\}$, $R_{\theta}\in\mathbb{R}^{2\times2}$ is the rotation matrix defined in \eqref{rotationmatrix}, and $\theta_{j,l}:=\pi l 2^{-\lceil j/2\rceil-1},$  for $j\geq0$, and $0\leq l<L_j:=2^{\lceil j/2 \rceil +2}$, are scale-dependent rotation angles. The functions $\phi \in L^1(\RR^2)\cap L^2(\RR^2)$ and $\psi_j \in L^1(\RR^2)\cap L^2(\RR^2)$ satisfy the Littlewood-Paley condition \eqref{freqcover} according to
\begin{equation}\label{eq:curve1}
\begin{split}
A\leq |\widehat{\phi}(\omega)|^2 +\sum_{j=0}^\infty \sum_{l=0}^{L_j-1} |\widehat{\psi_{j}}(R_{\theta_{j,l}}\omega)|^2 \leq B,
\end{split}
\end{equation}
a.e. $\omega \in \RR^2.$ The  $\psi_{j}$, $j\geq0$, are designed to have their Fourier transforms $\widehat{\psi}_{j}$ supported on a pair of opposite wedges of size $2^{-j/2}\times2^{j}$ in the dyadic corona $\{ \omega \in \RR^2 \ | \ 2^j \leq |\omega| \leq 2^{j+1}\}$, see Fig. \ref{fig:acurvea} (left). We refer the reader to \cite[Theorem 4.1]{CandesDonoho2} for constructions of functions $\phi,\psi_{j}$ satisfying \eqref{eq:curve1} with $A=B=1$.

\textit{Semi-discrete ridgelet frames:} 
Ridgelets, introduced in \cite{Ridgelet,DonohoCandesRidgelet}, are well-suited to the extraction of signal features characte\-rized by straight-line singu\-larities (such as, e.g., straight edges in images), and have been applied successfully in various practical feature extraction tasks \cite{Dettori,GChen,Qiao,Ganesan}.

A semi-discrete ridgelet frame for $L^2(\RR^2)$ is a collection of functions according to 
$
\Psi_{\Lambda_{\text{R}}}:=\{ T_bIg_{(j,l)}\}_{(j,l) \in \Lambda_{\text{R}}, b \in \RR^2},
$ 
with $g_{(0,0)}(x):=\phi(x),$ $g_{(j,l)}(x):=\psi_{(j,l)}(x),$ where $\Lambda_{\text{R}}:=\big\{ (0,0)\big\}\cup \big\{ (j,l) \ | \ j\geq 1, \ l=1,\dots,2^{j}-1 \big\}$, and the atoms $\{ g_{(j,l)}\}_{(j,l) \in \Lambda_{\text{R}}}$ satisfy the Littlewood-Paley condition \eqref{freqcover} according to 
\begin{equation}\label{eq:ridgelet}
A\leq |\widehat{\phi}(\omega)|^2 + \sum_{j=1}^\infty \sum_{l=1}^{2^{j}-1} |\widehat{\psi_{(j,l)}}(\omega)|^2\leq B,
\end{equation}
a.e. $\omega \in \RR^2.$ The  $\psi_{(j,l)} \in L^1(\RR^2)\cap L^2(\RR^2)$, $(j,l)\in\Lambda_{\text{R}}\backslash\{(0,0)\}$, are designed to be constant in the direction specified by the parameter $l$, and to have  Fourier transforms $\widehat{\psi}_{(j,l)}$ supported on a pair of opposite wedges of size $2^{-j}\times2^{j}$ in the dyadic corona $\{ \omega \in \RR^2 \ | \ 2^j \leq |\omega| \leq 2^{j+1}\}$, see Fig. \ref{fig:acurvea} (right). We refer the reader to  \cite[Proposition 6]{Grohs_transport} for constructions of functions $\phi,\psi_{(j,l)}$ satisfying \eqref{eq:ridgelet} with $A=B=1$.

\begin{remark}
For further examples of interesting structured semi-discrete frames, we refer to \cite{Shearlets}, which discusses semi-discrete shearlet frames, and \cite{Grohs_alpha}, which deals with semi-discrete $\alpha$-curvelet frames. 
\end{remark}


\section{Non-linearities}\label{app_nonlinear} 
This appendix gives a brief overview of non-linearities $M:L^2(\Rd) \to L^2(\Rd)$ that are widely used in the deep learning literature and that fit into our theory.  For each example, we establish how it satisfies the conditions on $M:L^2(\Rd) \to L^2(\Rd)$ in Theorems \ref{main_inv2} and \ref{mainmain} and in Corollary \ref{main_inv55}. Specifically, we need to verify the following:
\begin{itemize}
\item[(i)]{Lipschitz continuity: There exists a constant $L\geq0$ such that 
$\VV Mf-Mh \VV_2\leq L \VV f-h\VV_2,$ for all $f,h \in L^2(\Rd).$}
\item[(ii)]{$Mf=0$ for $f=0$.}
\end{itemize}
All non-linearities considered here are pointwise (memoryless) operators in the sense of 
\begin{equation}\label{eq:n1}
M:L^2(\Rd) \to L^2(\Rd), \hspace{0.5cm}(Mf)(x)=\rho(f(x)),
\end{equation}
where $\rho:\mathbb{C}\to \mathbb{C}$. An immediate consequence of this property is that  the operator $M$ commutes with the translation operator $T_t$ (see Theorem \ref{mainmain} and Corollary \ref{main_inv55}): 
\begin{align*}
(MT_tf)(x)&=\rho((T_tf)(x))=\rho(f(x-t))=T_t\rho (f(x))\\&=(T_tMf)(x),\hspace{0.5cm} \forall f\in L^2(\Rd),\forall t\in \Rd.
\end{align*}
\textit{Modulus function:} The modulus function
$$
|\cdot|:L^2(\Rd)\to L^2(\Rd), \hspace{0.5cm}|f|(x):=|f(x)|,
$$
has been applied successfully in the deep learning literature, e.g., in \cite{Jarrett,GaborLowe}, and most prominently in   scattering networks \cite{MallatS}. Lipschitz continuity with $L=1$ follows from 
\begin{align*}
\VV |f|-|h| \VV_2^2&=\int_{\Rd}||f(x)|-|h(x)||^2\mathrm dx \\&\leq \int_{\Rd}|f(x)-h(x)|^2\mathrm dx= \VV f-h\VV^2_2, 
\end{align*}
for $f,h \in L^2(\Rd)$, by the reverse triangle inequality. Furthermore, obviously $|f|=0$ for $f=0$, and finally $|\cdot|$ is pointwise as \eqref{eq:n1} is satisfied with $\rho(x):=|x|$.\\[1ex]
\textit{Rectified linear unit:} The rectified linear unit non-linearity (see, e.g., \cite{Glorot,Nair}) is defined as 
$R:L^2(\Rd)\to L^2(\Rd),$ 
$$(Rf)(x):=\max\{0,\Real(f(x))\}+i\max\{0,\Imag(f(x))\}.
$$
We start by establishing that $R$ is Lipschitz-continuous with $L=2$. 
To this end, fix $f,h \in L^2(\Rd)$. We have 
\begin{align}
&|(Rf)(x)-(Rh)(x)|\nonumber\\
&=\big|\max\{0,\Real(f(x))\}+i\max\{0,\Imag(f(x))\}\nonumber\\
&-\big(\max\{0,\Real(h(x))\}+i\max\{0,\Imag(h(x))\}\big)\big|\nonumber\\
&\leq\big|\max\{0,\Real(f(x))\}-\max\{0,\Real(h(x))\}\big|\label{eq:nn1}\\
&+ \big| \max\{0,\Imag(f(x))\}-\max\{0,\Imag(h(x))\}\big|\nonumber\\
&\leq\big|\Real(f(x))-\Real(h(x))\big| + \big|\Imag(f(x))-\Imag(h(x))\big|\label{eq:nn2}\\
&\leq\big|f(x)-h(x)\big| + \big|f(x)-h(x)\big|= 2|f(x)-h(x)|\label{eq:nnn2},
\end{align}
where we used the triangle inequality in \eqref{eq:nn1}, $$|\max\{0,a\}-\max\{0,b\}|\leq |a-b|, \hspace{0.5cm} \forall a,b\in \RR,$$
in \eqref{eq:nn2}, and the Lipschitz continuity (with $L=1$) of  the mappings $\Real:\mathbb{C} \to\RR$ and $\Imag:\mathbb{C}\to\RR$  in \eqref{eq:nnn2}. We therefore get 
\begin{align*}
\VV Rf-Rh \VV_2=&\Big(\int_{\Rd}|(Rf)(x)-(Rh)(x)|^2\mathrm dx\Big)^{1/2}\\
\leq&\,2\,\Big(\int_{\Rd}|f(x)-h(x)|^2\mathrm dx\Big)^{1/2}\\
=&\, 2\,  \VV f-h\VV_2,
\end{align*}
which establishes Lipschitz continuity of $R$ with Lipschitz constant $L=2$. Furthermore, obviously $Rf=0$ for $f=0$, and finally \eqref{eq:n1} is satisfied with $\rho(x):=\max\{0,\Real(x)\}+i\max\{0,\Imag(x)\}$.\\[1ex]
\textit{Hyperbolic tangent:} The hyperbolic tangent non-linearity (see, e.g., \cite{Huang,Jarrett,hierachies}) is defined as $H:L^2(\Rd)\to L^2(\Rd)$, 
$$(Hf)(x):=\tanh(\Real(f(x)))+i\tanh(\Imag(f(x))),
$$
where $\tanh(x):=\frac{e^{x}-e^{-x}}{e^{x}+e^{-x}}$. We start by proving that $H$ is Lipschitz-continuous with $L=2$. 
To this end, fix $f,h \in L^2(\Rd)$. We have 
\begin{align}
&|(Hf)(x)-(Hh)(x)|\nonumber\\
&=\big|\tanh(\Real(f(x)))+i\tanh(\Imag(f(x)))\nonumber\\
&-\big(\tanh(\Real(h(x)))+i\tanh(\Imag(h(x)))\big)\big|\nonumber\\
&\leq\big|\tanh(\Real(f(x)))-\tanh(\Real(h(x)))\big|\nonumber\\
&+\big|\tanh(\Imag(f(x)))-\tanh(\Imag(h(x)))\big|\label{eq:nn3},
\end{align}
where, again, we used the triangle inequality. In order to further upper-bound \eqref{eq:nn3}, we show that $\tanh$ is Lipschitz-continuous. To this end, we make use of the following result.
\begin{lemma}\label{Lip}
Let $h:\RR \to \RR$ be a continuously differentiable function satisfying $\sup_{x\in \RR}|h'(x)|\leq L$. Then, $h$ is Lipschitz-continuous  with Lipschitz constant $L$.
\end{lemma}
\begin{proof}
See \cite[Theorem 9.5.1]{LipLip}.
\end{proof}
Since $\tanh'(x)=1-\tanh^2(x)$, $x\in \RR$, we have $\sup_{x\in\RR}|\tanh'(x)|\leq 1$. By Lemma \ref{Lip} we can therefore conclude that $\tanh$ is Lipschitz-continuous with $L=1$, which when used in \eqref{eq:nn3}, yields
\begin{align*}
|(Hf)(x)-(Hh)(x)|&\leq\big|\Real(f(x))-\Real(h(x))\big| \\
&+ \big|\Imag(f(x))-\Imag(h(x))\big|\\
&\leq\big|f(x)-h(x)\big| + \big|f(x)-h(x)\big|\\
&= 2|f(x)-h(x)|.
\end{align*}
Here, again, we used the Lipschitz continuity (with $L=1$) of $\Real:\mathbb{C}\to\RR$ and $\Imag:\mathbb{C}\to\RR$. Putting things together, we obtain  
\begin{align*}
\VV Hf-Hh \VV_2=&\Big(\int_{\Rd}|(Hf)(x)-(Hh)(x)|^2\mathrm dx\Big)^{1/2}\\
\leq&\,2\,\Big(\int_{\Rd}|f(x)-h(x)|^2\mathrm dx\Big)^{1/2}\\
=&\,2\,\VV f-h\VV_2,
\end{align*}
which proves that $H$ is Lipschitz-continuous with $L=2$. Since $\tanh(0)=0$, we trivially have $Hf=0$ for $f=0$. Finally, \eqref{eq:n1} is satisfied with $\rho(x):=\tanh(\Real(x))+i\tanh(\Imag(x))$.\\[1ex]
\textit{Shifted logistic sigmoid:} The shifted logistic sigmoid non-linearity\footnote{Strictly speaking, it is actually the sigmoid function $x\mapsto \frac{1}{1+e^{-x}}$ rather than the shifted sigmoid function $x\mapsto \frac{1}{1+e^{-x}}-\frac{1}{2}$ that is used in \cite{Mohamed,Glorot2}. We incorporated the offset $\frac{1}{2}$ in order to satisfy the requirement $Pf=0$ for $f=0$.} (see, e.g., \cite{Mohamed,Glorot2}) is defined as $P:L^2(\Rd)\to L^2(\Rd)$, $$(Pf)(x):=\sig(\Real(f(x)))+i\sig(\Imag(f(x))),
$$
where $\sig(x):=\frac{1}{1+e^{-x}}-\frac{1}{2}$. We first establish that $P$ is Lipschitz-continuous with $L=\frac{1}{2}$. 
To this end, fix $f,h \in L^2(\Rd)$. We have 
\begin{align}
|(Pf)(x)&-(Ph)(x)|=\big|\sig(\Real(f(x)))+i\sig(\Imag(f(x)))\nonumber\\
&-\big(\sig(\Real(h(x)))+i\sig(\Imag(h(x)))\big)\big|\nonumber\\
&\leq\big|\sig(\Real(f(x)))-\sig(\Real(h(x)))\big|\nonumber\\
&+\big|\sig(\Imag(f(x)))-\sig(\Imag(h(x)))\big|\label{eq:nn5},
\end{align}
where, again, we employed the triangle inequality. As before, to further upper-bound  \eqref{eq:nn5}, we show that $\sig$ is Lipschitz-continuous. Specifically, we apply Lemma \ref{Lip} with $\sig'(x)=\frac{e^{-x}}{(1+e^{-x})^2}$, $x\in \RR$, and hence $\sup_{x\in\RR}|\sig'(x)|\leq \frac{1}{4}$, to conclude that $\sig$ is Lipschitz-continuous with $L=\frac{1}{4}$.
When used in \eqref{eq:nn5} this yields (together with the Lipschitz continuity, with $L=1$, of $\Real:\mathbb{C}\to\RR$ and $\Imag:\mathbb{C}\to\RR$)
\begin{align}
&|(Pf)(x)-(Ph)(x)|\leq\frac{1}{4}\,\Big|\Real(f(x))-\Real(h(x))\Big| \nonumber \\
&+ \frac{1}{4}\,\Big|\Imag(f(x))-\Imag(h(x))\Big|\leq\frac{1}{4}\,\Big|f(x)-h(x)\Big| \nonumber\\
&+ \frac{1}{4}\,\Big|f(x)-h(x)\Big|= \frac{1}{2}\,\Big|f(x)-h(x)\Big|\label{eq:merci}.
\end{align}
It now follows from \eqref{eq:merci} that 
\begin{align*}
\VV Pf-Ph \VV_2=&\Big(\int_{\Rd}|(Pf)(x)-(Ph)(x)|^2\mathrm dx\Big)^{1/2}\\
\leq&\frac{1}{2}\,\Big(\int_{\Rd}|f(x)-h(x)|^2\mathrm dx\Big)^{1/2}\\
=&\, \frac{1}{2}\,\VV f-h\VV_2,
\end{align*}
which establishes Lipschitz continuity of $P$ with $L=\frac{1}{2}$. Since $\sig(0)=0$, we trivially have $Pf=0$ for $f=0$. Finally, \eqref{eq:n1} is satisfied with $\rho(x):=\sig(\Real(x))+i\sig(\Imag(x))$.

\section{Proof of Proposition \ref{main_well}}\label{app_well}
We need to show that $\Phi_\Omega(f)\in (L^2(\Rd))^\QQ$, for all $f \in L^2(\Rd)$. This will be accomplished by proving an even stronger result, namely 
\begin{equation}\label{eq:w1}
|||\Phi_\Omega(f)|||\leq\VV f\VV_2,\hspace{0.5cm}\forall f\in L^2(\Rd),
\end{equation}
which, by $\VV f \VV_2<\infty$, establishes the claim. For ease of notation, we let $f_q:=U[q]f$, for $f \in L^2(\Rd)$, in the following. Thanks to \eqref{eq:p13} and \eqref{weak_admiss2}, we have $\VV f_q\VV_2 \leq \VV f \VV _2<\infty$, and thus $f_q \in L^2(\Rd)$. The key idea of the proof is now---similarly to the proof of \cite[Proposition 2.5]{MallatS}---to judiciously employ a telescoping series argument. We start by writing
\begin{align}
||| \Phi_\Omega(f)|||^2&=\sum_{n=0}^\infty\sum_{q\in \Lambda_{1}^n} ||f_q \ast\chi_{n}||^2_2\nonumber\\
&=\lim\limits_{N\to \infty}\sum_{n=0}^N\underbrace{\sum_{q\in \Lambda_{1}^n} ||f_q\ast\chi_{n} ||^2_2}_{:=a_n}.\label{eq:w-7}
\end{align}
The key step is then to establish that $a_n$ can be upper-bounded according to
\begin{equation}\label{eq:w3}
\begin{split}
a_n\leq b_n-b_{n+1},\hspace{0.5cm} \forall n \in \mathbb{N}_0,
\end{split}
\end{equation}
with $b_n:=\sum_{q\in\Lambda_1^n}\VV f_q\VV^2_2,$ $n \in \mathbb{N}_0,$ and to use this result in a telescoping series argument according to 
 \begin{align}
&\sum_{n=0}^Na_n\leq\sum_{n=0}^N(b_n-b_{n+1})= (b_0-b_1) + (b_1-b_2) \nonumber\\&+\cdots 
+(b_N-b_{N+1})=b_0 - \underbrace{b_{N+1}}_{\geq0}\\
&\leq b_0= \sum_{q\in\Lambda_1^0}\VV f_q\VV^2_2=\VV U[e]f\VV^2_2=\VV f\VV^2_2.\label{eq:w-5}
\end{align}
By \eqref{eq:w-7} this then implies \eqref{eq:w1}.
We start by noting that \eqref{eq:w3} reads 
\begin{equation}\label{eq:ww5}
 \sum_{q\in\Lambda_1^n}\VV f_q\ast\chi_n\VV^2_2 \leq \sum_{q\in \Lambda_{1}^n} ||f_q\VV^2_2-\sum_{q\in\Lambda_1^{n+1}}\VV f_q\VV^2_2, 
 \end{equation}
for all $n \in \mathbb{N}_0$, and proceed by examining the second term on the right hand side (RHS) of \eqref{eq:ww5}. Every path 
\begin{equation*}\label{eq:ee1}
\tilde{q} \in \Lambda_{1}^{n+1}=\underbrace{\Lambda_{1}\times\dots\times\Lambda_{n}}_{=\Lambda_{1}^n}\times\Lambda_{n+1}
\end{equation*}
of length $n+1$ can  be decomposed into a path $q \in \Lambda_1^{n}$ of length $n$ and an index $\lambda_{n+1} \in \Lambda_{n+1}$ according to $\tilde{q}=(q,\lambda_{n+1})$. Thanks to \eqref{aaaa} we have  $$U[\tilde{q}]=U[(q,\lambda_{n+1})]=U_{n+1}[\lambda_{n+1}]U[q],$$ which yields 
\begin{align}
\sum_{\tilde{q}\in\Lambda_1^{n+1}}\VV f_{\tilde{q}}\VV^2_2&=\sum_{q\in\Lambda_1^{n}}\sum_{\lambda_{n+1}\in \Lambda_{n+1}}\VV U_{n+1}[\lambda_{n+1}]f_	q\VV^2_2\label{eq:w6}.
\end{align}
Substituting the second term on the RHS of \eqref{eq:ww5} by \eqref{eq:w6} now yields
\begin{align*}
 &\sum_{q\in\Lambda_1^n}\VV f_q\ast\chi_n\VV^2_2 \nonumber\\
 &\leq \sum_{q\in \Lambda_{1}^n}\Big( ||f_q\VV^2_2-\sum_{\lambda_{n+1}\in \Lambda_{n+1}}\VV U_{n+1}[\lambda_{n+1}]f_q\VV^2_2\Big), \hspace{0.5cm} \forall n \in \mathbb{N}_0,
\end{align*}
which can be rewritten as 
\begin{align}
&\sum_{q\in\Lambda_1^n}\Big(\VV f_q\ast\chi_n\VV^2_2+\sum_{\lambda_{n+1}\in\Lambda_{n+1}}\VV U_{n+1}[\lambda_{n+1}]f_q\VV^2_2\Big)\label{eq:w9}\\
&\leq \sum_{q\in \Lambda_{1}^n} ||f_q\VV^2_2, \hspace{0.5cm}\forall n\in \mathbb{N}_0.\nonumber
\end{align}
Next, note that the second term inside the sum on the left hand side (LHS) of \eqref{eq:w9} can be written as 
\begin{align}
&\sum_{\lambda_{n+1}\in\Lambda_{n+1}}\VV U_{n+1}[\lambda_{n+1}]f_q\VV^2_2\nonumber\\
&= \sum_{\lambda_{n+1}\in\Lambda_{n+1}}\int_{\Rd} |(U_{n+1}[\lambda_{n+1}]f_q)(x)|^2\mathrm dx\nonumber\\
 &= \sum_{\lambda_{n+1}\in\Lambda_{n+1}}\hspace{-0.4cm}S_{n+1}^{d}\int_{\Rd} \Big|P_{n+1}\big(M_{n+1}(f_q\ast g_{\lambda_{n+1}})\big)(S_{n+1}x)\Big|^2\mathrm dx\nonumber\\
  &= \sum_{\lambda_{n+1}\in\Lambda_{n+1}}\int_{\Rd} \Big|P_{n+1}\big(M_{n+1}(f_q\ast g_{\lambda_{n+1}})\big)(y)\Big|^2\mathrm dy\nonumber\\
&=\sum_{\lambda_{n+1}\in\Lambda_{n+1}}\VV P_{n+1}\big(M_{n+1}(f_q\ast g_{\lambda_{n+1}})\big)\VV_2^2, \label{eq:w-10}
\end{align}
for all $n \in \mathbb{N}_0$. Noting that $f_q \in L^2(\Rd)$, as established above, and $g_{\lambda_{n+1}} \in L^1(\Rd)$, by assumption, it follows that $(f_q\ast g_{\lambda_{n+1}})\in L^2(\Rd)$ thanks to Young's inequality \cite[Theorem 1.2.12]{Grafakos}. We use the Lipschitz property of $M_{n+1}$ and $P_{n+1}$, i.e., $\VV M_{n+1}(f_q\ast g_{\lambda_{n+1}})-M_{n+1}h\VV_2\leq L_{n+1}\VV f_q\ast g_{\lambda_{n+1}}-h\VV,$ and $\VV P_{n+1}(f_q\ast g_{\lambda_{n+1}})-P_{n+1}h\VV_2\leq R_{n+1}\VV f_q\ast g_{\lambda_{n+1}}-h\VV,$ together with $M_{n+1}h=0$ and  $P_{n+1}h=0$ for $h=0$, to upper-bound the term inside the sum in \eqref{eq:w-10} according to
\begin{align}
&\VV P_{n+1}\big(M_{n+1}(f_q\ast g_{\lambda_{n+1}})\big)\VV_2^2\leq R_{n+1}^2\VV M_{n+1}(f_q\ast g_{\lambda_{n+1}})\VV_2^2\nonumber\\
&\leq L_{n+1}^2R_{n+1}^2\VV f_q\ast g_{\lambda_{n+1}}\VV_2^2,\hspace{0.5cm} \forall n \in \mathbb{N}_0.\label{eq:w70}
\end{align} 
Substituting the second term inside the sum on the LHS of \eqref{eq:w9} by the upper bound resulting from insertion of \eqref{eq:w70} into \eqref{eq:w-10} yields 
\begin{align}
 &\sum_{q\in\Lambda_1^n}\Big(\VV f_q\ast\chi_n\VV^2_2+L_{n+1}^{2}R_{n+1}^{2}\sum_{\lambda_{n+1}\in\Lambda_{n+1}}\VV f_q\ast g_{\lambda_{n+1}}\VV_2^2 \Big)\nonumber\\
\leq&\sum_{q\in\Lambda_1^n}\max\{1,L^2_{n+1}R_{n+1}^2\}\Big(\VV f_q\ast\chi_n\VV^2_2\nonumber\\&+\sum_{\lambda_{n+1}\in\Lambda_{n+1}}\VV f_q\ast g_{\lambda_{n+1}}\VV_2^2 \Big), \hspace{0.5cm} \forall n \in \mathbb{N}_0.\label{eq:w-199}
\end{align}
As the functions $\{ g_{\lambda_{n+1}}\}_{\lambda_{n+1}\in\Lambda_{n+1}}\cup\{ \chi_n\}$ are the atoms of the semi-discrete frame $\Psi_{n+1}$ for $L^2(\Rd)$ and $f_q \in L^2(\Rd)$, as established above, we have
\begin{equation*}\label{eq:w15}
\VV f_q\ast\chi_n\VV^2_2+\sum_{\lambda_{n+1}\in\Lambda_{n+1}}\VV f_q\ast g_{\lambda_{n+1}}\VV_2^2\leq B_{n+1}\VV f_q\VV^2_2,
\end{equation*}
which, when used in \eqref{eq:w-199} yields
\begin{align}
&\sum_{q\in\Lambda_1^n}\Big(\VV f_q\ast\chi_n\VV^2_2+\sum_{\lambda_{n+1}\in\Lambda_{n+1}}\VV U_{n+1}[\lambda_{n+1}]f_q\VV^2_2\Big)\nonumber\\
\leq&\sum_{q\in\Lambda_1^n}\max\{1,L^2_{n+1}R^2_{n+1}\}B_{n+1}\VV f_q\VV^2_2 \nonumber\\
=&\sum_{q\in\Lambda_1^n}\max\{B_{n+1},B_{n+1}L^2_{n+1}R^2_{n+1}\}\VV f_q\VV^2_2,\label{eq:w-200}
\end{align}
for all $n \in \mathbb{N}_0$. Finally, invoking the assumption
$$\max\{B_n,B_nL_n^2R^2_{n+1} \}\leq 1,\hspace{0.5cm} \forall n \in \mathbb{N},$$
in \eqref{eq:w-200} yields \eqref{eq:w9} and thereby completes the proof.

\section{Proof of Theorem \ref{main_inv2}}\label{app_inv}
We start by proving i). The key step in establishing \eqref{eq:i2} is to show that the operator $U_n$, $n\in \mathbb{N}$, defined in \eqref{eq:e1} satisfies the relation
\begin{equation}\label{eq:i100}
U_n[\lambda_n]T_tf=T_{t/ S_n}U_n[\lambda_n]f, 
\end{equation}
 for all $f\in L^2(\Rd),$   $t\in \Rd,$ and  $\lambda_{n}\in \Lambda_n.$ With the definition of $U[q]$ in \eqref{aaaa} this then yields 
\begin{equation}\label{eq:i101}
\begin{split}
U[q]T_tf=T_{t/(S_1\cdots S_{n})}U[q]f, 
\end{split}
\end{equation}
 for all $f\in L^2(\Rd),$   $t\in \Rd,$ and  $q\in \Lambda_1^n.$ The identity \eqref{eq:i2} is then a direct consequence of \eqref{eq:i101} and the translation-covariance of the convolution operator: 
\begin{equation*}\label{eq:i5}
\begin{split}
\Phi_\Omega^n(T_tf)&=\big\{ \big(U[q]T_tf \big) \ast \chi_n \big\}_{q\in \Lambda_1^n}\\
&=\big\{ \big(T_{t/(S_1\cdots S_{n})}U[q]f\big) \ast \chi_n \big\}_{q\in \Lambda_1^n}\\
&=\big\{ T_{t/(S_1\cdots S_{n})}\big((U[q]f) \ast \chi_n\big) \big\}_{q\in \Lambda_1^n}\\
&=T_{t/(S_1\cdots S_{n})}\big\{ (U[q]f) \ast \chi_n \big\}_{q\in \Lambda_1^n}\\
&=T_{t/(S_1\cdots S_{n})}\Phi_\Omega^n(f), \hspace{0.5cm} \forall f \in L^2(\Rd), \, \forall t \in \Rd.
\end{split}
\end{equation*}
To establish \eqref{eq:i100}, we first define the unitary operator $$D_{n}:L^2(\Rd) \to L^2(\Rd), \quad D_{n}f:=S_n^{d/2}f(S_n\cdot),$$ and note that
\begin{align}
U_n[\lambda_n]T_tf&=S_n^{d/2}P_n\Big(M_n\big((T_tf)\ast g_{\lambda_n}\big)\Big)(S_n\cdot)\nonumber\\
&=D_{n}P_n\Big(M_n\big((T_tf)\ast g_{\lambda_n}\big)\Big)\nonumber\\&=D_{n}P_n\Big(M_n\big(T_t(f\ast g_{\lambda_n})\big)\Big)\nonumber\\
&=D_{n}P_n\Big(T_t\big(M_n(f\ast g_{\lambda_n})\big)\Big)\label{isi110}\\
&=D_{n}T_t\bigg(P_n\Big(\big(M_n(f\ast g_{\lambda_n})\big)\Big)\bigg),\label{eq:i10}
\end{align}
for all $f \in L^2(\Rd)$ and  $t \in \Rd$, where in \eqref{isi110} and \eqref{eq:i10} we employed $$M_nT_t=T_tM_n, \quad \text{ and } \quad P_nT_t=T_tP_n,$$ for all $n\in\mathbb{N}$ and $t\in \mathbb{R}^d$, respectively, both of which are by assumption. Next, using  
\begin{align*}
D_nT_tf&=S_n^{d/2}f(S_n\cdot-t)=S_n^{d/2}f(S_n(\cdot-t/S_n))\\
&=T_{t/S_n}D_nf,\hspace{0.5cm} \forall f \in L^2(\Rd), \, \forall t \in \Rd,
\end{align*}
in \eqref{eq:i10} yields
\begin{align*}
U_n[\lambda_n]T_tf&=D_{n}T_t\bigg(P_n\Big(\big(M_n(f\ast g_{\lambda_n})\big)\Big)\bigg)\nonumber\\
&=T_{t/S_n}\bigg(D_nP_n\Big(\big(M_n(f\ast g_{\lambda_n})\big)\Big)\bigg)\nonumber\\
&=T_{t/S_n}U_n[\lambda_n]f,
\end{align*}
for all $f\in L^2(\Rd)$ and $t\in \Rd$. This completes the proof of i). 

Next, we prove ii). For ease of notation, again, we let $f_q:=U[q]f$, for $f \in L^2(\Rd)$. Thanks to \eqref{eq:p13} and the admissibility condition \eqref{weak_admiss2}, we have $\VV f_q \VV_2 \leq \VV f \VV_2<\infty$, and thus $f_q \in L^2(\Rd)$. We first write 
\begin{align}
&||| \Phi^n_\Omega(T_tf) - \Phi^n_\Omega(f)|||^2\nonumber\\
=& \ |||T_{t/(S_1\cdots S_{n})} \Phi^n_\Omega(f) - \Phi^n_\Omega(f)|||^2\label{eq:a1}\\
=&\sum_{q\in \Lambda_1^n}\VV T_{t/(S_1\cdots S_{n})} (f_q\ast\chi_{n})-f_q\ast\chi_{n} \VV_2^2\nonumber\\
=&\sum_{q\in \Lambda_1^n}\VV M_{-t/(S_1\cdots S_{n})} (\widehat{f_q\ast\chi_{n}})-\widehat{f_q\ast\chi_{n}} \VV_2^2\label{eq:ab2},
\end{align}
for all $n \in \mathbb{N}$, where in \eqref{eq:a1} we used \eqref{eq:i2}, and in \eqref{eq:ab2} we employed Parseval's formula \cite[p. 189]{Rudin} (noting that $(f_q\ast \chi_n)\in L^2(\Rd)$ thanks to Young's inequality \cite[Theorem 1.2.12]{Grafakos}) together with the  relation $$\widehat{T_tf}=M_{-t}\widehat{f},\quad  \forall f\in L^2(\Rd), \forall \, t\in \Rd.$$ The key step is then to establish the upper bound 
\begin{align}
&\VV M_{-t/(S_1\cdots S_{n})} (\widehat{f_q\ast\chi_{n}})-\widehat{f_q\ast\chi_{n}} \VV_2^2\nonumber\\
&\leq \frac{4\pi^2| t|^2K^2}{(S_1\cdots S_{n})^2}\VV f_q \VV_2^2, \hspace{0.5cm} \forall n \in \mathbb{N},\label{eq:aaaa2}
\end{align}
where $K>0$ corresponds to the constant in the decay condition \eqref{decaycondition}, and to note that 
\begin{align}
\sum_{q\in \Lambda_1^n} \VV f_q \VV_2^2\leq \sum_{q\in \Lambda_{1}^{n-1}}\VV f_q\VV^2_2\label{eq:aa7},\hspace{0.5cm} \forall n \in \mathbb{N},
\end{align}
which follows from \eqref{eq:w3} thanks to 
\begin{align} 0&\leq\sum_{q\in \Lambda_{1}^{n-1}} ||f_q\ast\chi_{n-1} ||^2_2=a_{n-1}\leq b_{n-1}-b_{n}\\
&=  \sum_{q\in \Lambda_{1}^{n-1}}\VV f_q\VV^2_2-\sum_{q\in \Lambda_1^n} \VV f_q \VV_2^2,\hspace{0.5cm} \forall n \in \mathbb{N}.
\end{align}
 Iterating on \eqref{eq:aa7} yields 
\begin{align}
\sum_{q\in \Lambda_1^n} \VV f_q \VV_2^2&\leq \sum_{q\in \Lambda_1^{n-1}} \VV f_q \VV_2^2 \leq\dots\leq \sum_{q\in \Lambda_1^0} \VV f_q \VV_2^2\nonumber\\
&=\VV U[e]f\VV_2^2=\VV f\VV_2^2,\label{eq:aa10} \hspace{0.5cm} \forall n \in \mathbb{N}.
\end{align}
The identity \eqref{eq:ab2} together with the inequalities \eqref{eq:aaaa2} and \eqref{eq:aa10} then directly imply 
\begin{equation}\label{theend}
||| \Phi^n_\Omega(T_tf) - \Phi^n_\Omega(f)|||^2\leq \frac{4\pi^2| t|^2K^2}{(S_1\cdots S_{n})^2}\VV f \VV_2^2,
\end{equation}
for all $n \in \mathbb{N}.$ It remains to prove \eqref{eq:aaaa2}. To this end, we first note that 
\begin{align}
&\VV M_{-t/(S_1\cdots S_{n})} (\widehat{f_q\ast\chi_{n}})-\widehat{f_q\ast\chi_{n}} \VV_2^2\nonumber\\
=&\int_{\Rd}\big|e^{-2\pi i \langle t , \omega \rangle /(S_1\cdots S_{n})}-1\big|^2|\widehat{\chi_{n}}(\omega)|^2|\widehat{f_q}(\omega)|^2\mathrm d\omega\label{eq:aa21}.
\end{align}

Since $|e^{-2\pi i x} -1|\leq 2\pi |x|$, for all $ x \in \RR,$ it follows that
\begin{align}
|e^{-2\pi i \langle t , \omega \rangle /(S_1\cdots S_{n})}-1|^2&\leq \frac{4\pi^2|\langle t , \omega \rangle|^2}{(S_1\cdots S_{n})^2}\nonumber\\
&\leq \frac{4\pi^2| t|^2 |\omega|^2}{(S_1\cdots S_{n})^2},\label{eq:2}
\end{align}
where in the last step we employed the Cauchy-Schwartz inequality. Substituting \eqref{eq:2} into \eqref{eq:aa21} yields
\begin{align}
&\VV M_{-t/(S_1\cdots S_{n})} (\widehat{f_q\ast\chi_{n}})-\widehat{f_q\ast\chi_{n}} \VV_2^2\nonumber\\
\leq\, &\frac{4\pi^2| t|^2}{(S_1\cdots S_{n})^2} \int_{\Rd} |\omega|^2|\widehat{\chi_{n}}(\omega)|^2|\widehat{f_q}(\omega)|^2\mathrm d\omega\nonumber\\
\leq\, &\frac{4\pi^2| t|^2K^2}{(S_1\cdots S_{n})^2} \int_{\Rd} |\widehat{f_q}(\omega)|^2\mathrm d\omega\label{eq:aa29}\\
=\, &\frac{4\pi^2| t|^2K^2}{(S_1\cdots S_{n})^2} \,\VV \widehat{f_q}\VV_2^2= \frac{4\pi^2| t|^2K^2}{(S_1\cdots S_{n})^2} \,\VV f_q\VV_2^2,\label{miau}
\end{align}
for all $n \in \mathbb{N}$, where in \eqref{eq:aa29} we employed the decay condition \eqref{decaycondition}, and in the last step, again, we used Parseval's formula \cite[p. 189]{Rudin}. This establishes \eqref{eq:aaaa2} and thereby completes the proof of ii).
\section{Proof of Corollary \ref{main_inv55}}\label{app:covcov}
The key idea of the proof is---similarly to the proof of ii) in Theorem \ref{main_inv2}---to upper-bound the deviation from perfect covariance in the frequency domain. For ease of notation, again, we let $f_q:=U[q]f$, for $f \in L^2(\Rd)$. Thanks to \eqref{eq:p13} and the admissibility condition \eqref{weak_admiss2}, we have $\VV f_q \VV_2 \leq \VV f \VV_2<\infty$, and thus $f_q \in L^2(\Rd)$. We first write 
\begin{align}
&||| \Phi^n_\Omega(T_tf) - T_t\Phi^n_\Omega(f)|||^2\nonumber\\
&= \ |||T_{t/(S_1\cdots S_{n})} \Phi^n_\Omega(f) - T_t\Phi^n_\Omega(f)|||^2\label{eq:aa1}\\
&=\sum_{q\in \Lambda_1^n}\VV (T_{t/(S_1\cdots S_{n})}-T_t) (f_q\ast\chi_{n}) \VV_2^2\nonumber\\
&=\sum_{q\in \Lambda_1^n}\VV (M_{-t/(S_1\cdots S_{n})}-M_{-t}) (\widehat{f_q\ast\chi_{n}}) \VV_2^2\label{eq:aa2},
\end{align}
for all $n \in \mathbb{N}$, 
where in \eqref{eq:aa1} we used \eqref{eq:i2}, and in \eqref{eq:aa2} we employed Parseval's formula \cite[p. 189]{Rudin} (noting that $(f_q\ast \chi_n)\in L^2(\Rd)$ thanks to Young's inequality \cite[Theorem 1.2.12]{Grafakos}) together with the  relation $$\widehat{T_tf}=M_{-t}\widehat{f}, \quad \forall f\in L^2(\Rd), \forall\, t\in \Rd.$$ The key step is then to establish the upper bound 
\begin{align}
&\VV (M_{-t/(S_1\cdots S_{n})}-M_{-t}) (\widehat{f_q\ast\chi_{n}})\VV_2^2\nonumber\\
&\leq 4\pi^2|t|^2K^2\big|1/(S_1\cdots S_{n})-1\big|^2\VV f_q \VV_2^2,\label{eq:aaa2}
\end{align}
where $K>0$ corresponds to the constant in the decay condition \eqref{decaycondition}. Arguments similar to those leading to \eqref{theend} then complete the proof. It remains to prove \eqref{eq:aaa2}:
\begin{align}
&\VV (M_{-t/(S_1\cdots S_{n})}-M_{-t}) (\widehat{f_q\ast\chi_{n}}) \VV_2^2\nonumber\\
=&\int_{\Rd}\big|e^{-2\pi i \langle t , \omega \rangle /(S_1\cdots S_{n})}\nonumber\\
&-e^{-2\pi i \langle t , \omega \rangle}\big|^2|\widehat{\chi_{n}}(\omega)|^2|\widehat{f_q}(\omega)|^2\mathrm d\omega \label{eq:aaa21}.
\end{align}
Since $|e^{-2\pi i x} -e^{-2\pi i y}|\leq 2\pi |x-y|$, for all $ x,y \in \RR,$ it follows that
\begin{align}
&\big|e^{-2\pi i \langle t , \omega \rangle /(S_1\cdots S_{n})}-e^{-2\pi i \langle t , \omega \rangle}\big|^2\nonumber\\
&\leq 4\pi^2|t|^2| \omega|^2\big|1/(S_1\cdots S_{n})-1\big|^2\label{eq:a2},
\end{align}
where, again, we employed the Cauchy-Schwartz inequality. Substituting \eqref{eq:a2} into \eqref{eq:aaa21}, and employing arguments similar to those leading to \eqref{miau}, establishes \eqref{eq:aaa2} and thereby completes the proof. 

\section{Proof of Theorem \ref{mainmain}}\label{Appendix}
As already mentioned at the beginning of Section \ref{sec:defmstab}, the proof of the deformation sensitivity bound \eqref{mainmainmain} is based on two key ingredients. The first one, stated in Proposition \ref{summary} in Appendix \ref{proof:nonexpan}, establishes that the feature extractor $\Phi_\Omega$ is Lipschitz-continuous with Lipschitz constant $L_\Omega=1$, i.e.,
\begin{equation}\label{eq:p22}
||| \Phi_\Omega(f) -\Phi_\Omega(h)||| \leq \VV f-h \VV_2, \hspace{0.5cm} \forall f,h \in L^2(\Rd),
\end{equation}
and needs the admissibility condition \eqref{weak_admiss2} only. The second ingredient, stated in Proposition \ref{uppersuper} in Appendix \ref{proof:uppersuper}, is an upper bound on the deformation error $\VV f-F_{\tau,\omega} f\VV_2$ given by  
\begin{equation}\label{eq:p23}
\VV f-F_{\tau,\omega} f\VV_2 \leq C\big(R\VV \tau \VV_{\infty} + \VV \omega \VV_\infty \big)\VV f\VV_2,
\end{equation}
for all  $f \in L^2_R(\Rd)$, and is valid under the assumptions $\omega \in C(\Rd,\RR)$ and $\tau \in C^1(\Rd,\Rd)$ with $\VV D\tau \VV_\infty<\frac{1}{2d}$. We now show how \eqref{eq:p22} and \eqref{eq:p23} can be combined to establish the deformation sensitivity bound \eqref{mainmainmain}. To this end, we first apply \eqref{eq:p22} with $h:=F_{\tau,\omega}f=e^{2\pi i \omega(\cdot)}f(\cdot-\tau(\cdot))$ to get
\begin{equation}\label{eq:p20}
||| \Phi_\Omega(f) -\Phi_\Omega(F_{\tau,\omega} f)||| \leq \VV f-F_{\tau,\omega} f \VV_2,
\end{equation}
for all $f \in L^2(\Rd)$. Here, we used $F_{\tau,\omega}f \in L^2(\Rd)$, which is thanks to
\begin{equation*}
\begin{split}
\VV F_{\tau,\omega}f\VV_2^2 =\int_{\Rd}|f(x-\tau(x))|^2\mathrm dx
\leq 2\VV f\VV^2_2,
\end{split}
\end{equation*}
obtained through the change of variables $u=x-\tau(x)$, together with 
\begin{equation}\label{cokabss}
\frac{\mathrm du}{\mathrm dx}=|\hspace{-0.05cm}\det(E-(D \tau)(x))| \geq 1-d \VV D \tau \VV_\infty\geq 1/2, 
\end{equation}
for $x \in \Rd.$ The first inequality in \eqref{cokabss} follows from:
\begin{lemma}\cite[Corollary 1]{Brent}: \label{detabsch}
Let $M\in \RR^{d\times d}$ be such that  $|M_{i,j}|\leq \alpha$, for all $i,j$ with $1\leq i,j\leq d$. If $d \alpha \leq 1$, then 
$$|\hspace{-0.05cm}\det(E-M)|\geq 1-d \alpha.$$
\end{lemma}
The second inequality in \eqref{cokabss} is a consequence of the assumption $\VV D\tau \VV_\infty \leq\frac{1}{2d}$. The proof is finalized by replacing the RHS of  \eqref{eq:p20} by the RHS of \eqref{eq:p23}.

\section{Proposition \ref{summary}}\label{proof:nonexpan}
\begin{proposition}\label{summary}
 Let $\Omega=\big( (\Psi_n,M_n,P_n)\big)_{n \in \mathbb{N}}$ be an admissible module-sequence. The correspon\-ding feature extractor $\Phi_\Omega:L^2(\Rd) \to (L^2(\Rd))^\QQ$  is Lipschitz-continuous with Lipschitz constant $L_\Omega=1$, i.e.,  
\begin{equation}\label{zzz}
||| \Phi_\Omega(f) -\Phi_\Omega(h)||| \leq \VV f-h \VV_2, \hspace{0.5cm} \forall f,h \in L^2(\Rd).
\end{equation}
\end{proposition}
\begin{remark}
Proposition \ref{summary} generalizes \cite[Proposition 2.5]{MallatS}, which shows that the wavelet-modulus feature extractor $\Phi_W$ generated by scattering networks is Lipschitz-continuous with Lipschitz constant $L_W=1$. Specifically, our generalization allows for general semi-discrete frames (i.e., general convolution kernels), general Lipschitz-continuous non-linearities $M_n$, and general Lipschitz-continuous operators $P_n$, all of which can be different in different layers. Moreover, thanks to the admissibility condition \eqref{weak_admiss2}, the Lipschitz constant $L_\Omega=1$ in \eqref{zzz} is completely independent of the frame upper bounds $B_n$ and the Lipschitz-constants $L_n$ and $R_n$ of $M_n$ and $P_n$, respectively.
\end{remark}
\begin{proof}
The key idea of the proof is again---similarly to the proof of Proposition \ref{main_well} in Appendix \ref{app_well}---to judiciously employ a telescoping series argument. For ease of notation, we let $f_q:=U[q]f$ and $h_q:=U[q]h$, for $f,h \in L^2(\Rd)$. Thanks to \eqref{eq:p13} and the admissibility condition \eqref{weak_admiss2}, we have $\VV f_q\VV_2 \leq \VV f \VV _2<\infty$ and $\VV h_q\VV_2 \leq \VV h \VV _2<\infty$ and thus $f_q, h_q \in L^2(\Rd)$.  We start by writing 
\begin{equation*}
\begin{split}
&||| \Phi_\Omega(f) -\Phi_\Omega(h)|||^2=\sum_{n=0}^\infty\sum_{q\in \Lambda_{1}^n} ||f_q\ast\chi_n - h_q\ast\chi_n ||^2_2\\ 
&=\lim\limits_{N\to \infty}\sum_{n=0}^N\underbrace{\sum_{q\in \Lambda_{1}^n} ||f_q\ast\chi_n - h_q\ast\chi_n ||^2_2}_{=:a_n}.\\ 
\end{split}
\end{equation*}
As in the proof of Proposition \ref{main_well} in Appendix \ref{app_well}, the key step is to show that $a_n$ can be upper-bounded according to
\begin{equation}\label{eq:p1}
\begin{split}
a_n\leq b_n-b_{n+1},\hspace{0.5cm} \forall \,n \in \mathbb{N}_0,
\end{split}
\end{equation}
where here $$b_n:=\sum_{q\in\Lambda_1^n}\VV f_q-h_q\VV^2_2,\quad \forall\, n\in \mathbb{N}_0,$$
and to note that, similarly to \eqref{eq:w-5}, 
\begin{align*}
\sum_{n=0}^Na_n&\leq\sum_{n=0}^N(b_n-b_{n+1})= (b_0-b_1) + (b_1-b_2) \nonumber\\
&+\cdots +(b_N-b_{N+1})=b_0 - \underbrace{b_{N+1}}_{\geq0}\\
&\leq b_0= \sum_{q\in\Lambda_1^0}\VV f_q-h_q\VV^2_2=\VV U[e]f-U[e]h\VV^2_2\nonumber\\
&=\VV f-h\VV^2_2,
\end{align*}
which then yields  \eqref{zzz} according to
\begin{align*}
||| \Phi_\Omega(f) -\Phi_\Omega(h)|||^2&=\lim\limits_{N\to \infty}\sum_{n=0}^Na_n\leq\lim\limits_{N\to \infty}  \VV f-h\VV^2_2\nonumber\\
&=\VV f-h\VV^2_2.
\end{align*}
Writing out \eqref{eq:p1}, it follows that we need to establish 
\begin{align}
 &\sum_{q\in\Lambda_1^n}\VV f_q\ast\chi_n-h_q\ast\chi_n\VV^2_2 \nonumber\\
 &\leq \sum_{q\in \Lambda_{1}^n} ||f_q-h_q\VV^2_2-\sum_{q\in\Lambda_1^{n+1}}\VV f_q - h_q\VV^2_2, \label{eq:p5}
\end{align}
for all  $n \in \mathbb{N}_0$. We start by examining the second term on the RHS of \eqref{eq:p5} and note that, thanks to the decomposition $$\tilde{q} \in \Lambda_{1}^{n+1}=\underbrace{\Lambda_{1}\times\dots\times\Lambda_{n}}_{=\Lambda_{1}^n}\times\Lambda_{n+1}$$
and $U[\tilde{q}]=U[(q,\lambda_{n+1})]=U_{n+1}[\lambda_{n+1}]U[q]$, by \eqref{aaaa}, we have 
\begin{align}
\sum_{\tilde{q}\in\Lambda_1^{n+1}}\VV f_{\tilde{q}}-h_{\tilde{q}}\VV^2_2&=\sum_{q\in\Lambda_1^{n}}\sum_{\lambda_{n+1}\in \Lambda_{n+1}}\VV U_{n+1}[\lambda_{n+1}]f_q\nonumber\\&-U_{n+1}[\lambda_{n+1}]h_q\VV^2_2.\label{eq:6}
\end{align}
Substituting \eqref{eq:6} into \eqref{eq:p5} and rearranging terms, we obtain 
\begin{align}
&\sum_{q\in\Lambda_1^n}\Big(\VV f_q\ast\chi_n-h_q\ast\chi_n\VV^2_2 \nonumber\\
&+\sum_{\lambda_{n+1}\in\Lambda_{n+1}}\VV U_{n+1}[\lambda_{n+1}]f_q-U_{n+1}[\lambda_{n+1}]h_q\VV^2_2\Big)\nonumber\\&\leq \sum_{q\in \Lambda_{1}^n} ||f_q - h_q\VV^2_2, \hspace{0.5cm} \forall n \in \mathbb{N}_0\label{eq:p9}.
\end{align}
We next note that the second term inside the sum on the LHS of \eqref{eq:p9} satisfies 
\begin{align}&
\sum_{\lambda_{n+1}\in\Lambda_{n+1}}\VV U_{n+1}[\lambda_{n+1}]f_q-U_{n+1}[\lambda_{n+1}]h_q\VV^2_2\nonumber\\
&\leq \sum_{\lambda_{n+1}\in\Lambda_{n+1}}\VV P_{n+1}\big(M_{n+1}(f_q\ast g_{\lambda_{n+1}})\big)\nonumber\\
&-P_{n+1}\big(M_{n+1}(h_q\ast g_{\lambda_{n+1}})\big)\VV_2^2\label{eq:ee2}, 
\end{align}
where we employed arguments similar to those leading to \eqref{eq:w-10}. Substituting the second term inside the sum on the LHS of \eqref{eq:p9} by the upper bound \eqref{eq:ee2}, and using the Lipschitz property of $M_{n+1}$ and $P_{n+1}$ yields 
\begin{align}
&\sum_{q\in\Lambda_1^n}\Big(\VV f_q\ast\chi_n-h_q\ast\chi_n\VV^2_2 \nonumber\\
&+\sum_{\lambda_{n+1}\in\Lambda_{n+1}}\VV U_{n+1}[\lambda_{n+1}]f_q-U_{n+1}[\lambda_{n+1}]h_q\VV^2_2\Big)\nonumber\\
\leq&\sum_{q\in\Lambda_1^n}\max\{1,L^2_{n+1}R_{n+1}^2\}\Big(\VV (f_q-h_q)\ast\chi_n\VV^2_2\nonumber\\
&+\sum_{\lambda_{n+1}\in\Lambda_{n+1}}\VV (f_q-h_q)\ast g_{\lambda_{n+1}}\VV_2^2 \Big),\label{eq:e-199}
\end{align}
for all $n\in \mathbb{N}_0$. As the functions $\{ g_{\lambda_{n+1}}\}_{\lambda_{n+1}\in\Lambda_{n+1}}\cup\{ \chi_n\}$ are the atoms of the semi-discrete frame $\Psi_{n+1}$ for $L^2(\Rd)$ and $f_q, h_q \in L^2(\Rd)$, as established above, we have
\begin{align*}
&\VV (f_q-h_q)\ast\chi_n\VV^2_2+\sum_{\lambda_{n+1}\in\Lambda_{n+1}}\VV (f_q-h_q)\ast g_{\lambda_{n+1}}\VV_2^2\\
&\leq B_{n+1}\VV f_q-h_q\VV^2_2,
\end{align*}
which, when used in \eqref{eq:e-199} yields
\begin{align}
&\sum_{q\in\Lambda_1^n}\Big(\VV f_q\ast\chi_n-h_q\ast\chi_n\VV^2_2 \nonumber\\
&+\sum_{\lambda_{n+1}\in\Lambda_{n+1}}\VV U_{n+1}[\lambda_{n+1}]f_q-U_{n+1}[\lambda_{n+1}]h_q\VV^2_2\Big)\nonumber\\
\leq&\sum_{q\in\Lambda_1^n}\max\{B_{n+1},B_{n+1}L^2_{n+1}R_{n+1}^2\}\VV f_q-h_q\VV^2_2,\label{eq:e-200}
\end{align}
for all $n \in \mathbb{N}_0$.
Finally, invoking the admissibility condition
$$\max\{B_n,B_nL_n^2R_{n}^2 \}\leq 1,\hspace{0.5cm} \forall n \in \mathbb{N},$$
in \eqref{eq:e-200} we get \eqref{eq:p9} and hence \eqref{eq:p1}. This completes the proof.
\end{proof}
\section{Proposition \ref{uppersuper}}\label{proof:uppersuper}
\begin{proposition}\label{uppersuper}
There exists a constant $C>0$ such that for all $f\in L^2_R(\Rd)$,  $\omega \in C(\Rd,\RR)$, and $\tau \in C^1(\Rd,\Rd)$ with $\VV D\tau \VV_\infty<\frac{1}{2d}$, it holds that 
\begin{equation}\label{eq:p21}
\VV f-F_{\tau,\omega} f\VV_2 \leq C\big(R\VV \tau \VV_{\infty} + \VV \omega \VV_\infty \big)\VV f\VV_2.
\end{equation}
\end{proposition}
\begin{remark}
A similar bound  was derived in \cite[App. B]{MallatS} for scattering networks, namely
\begin{equation}\label{eq:w92}
\VV f\ast \psi_{(-J,0)} - F_{\tau} (f\ast\psi_{(-J,0)})\VV_2\leq C2^{-J+d}\VV \tau \VV_\infty \VV f \VV_2,
\end{equation}
for all $f \in L^2(\Rd)$, where $\psi_{(-J,0)}$ is the low-pass filter of a semi-discrete directional wavelet frame for $L^2(\Rd)$, and $(F_\tau f)(x)=f(x-\tau(x))$. The techniques for  proving \eqref{eq:p21} and \eqref{eq:w92} are related in the sense of both employing Schur's Lemma \cite[App. I.1]{Grafakos} and a Taylor series expansion argument \cite[p. 411]{RudinComplex}. The signal-class specificity of our bound \eqref{eq:p21} comes with new technical elements  detailed at the beginning of the proof. 
\end{remark}
\begin{proof}
We first determine an integral operator  \begin{equation}\label{eq:p25}(Kf)(x)=\int_{\Rd} k(x,u)f(u)\mathrm du\end{equation} satisfying the signal-class specific identity $$Kf=F_{\tau,\omega}f-f, \quad \forall f\in L^2(\Rd),$$ and then upper-bound the deformation error $\VV f-F_{\tau,\omega} f \VV_2$ according to 
\begin{equation*}
\VV f-F_{\tau,\omega} f \VV_2=\VV F_{\tau,\omega} f -f\VV_2=\VV Kf\VV_2\leq \VV K \VV_{2,2} \VV f \VV_2, 
\end{equation*}
for all $f \in L^2_R(\Rd).$ Application of Schur's Lemma, stated below, then yields 
\begin{equation*}\label{eq:p26}
\VV K \VV_{2,2}\leq C \big(R\VV \tau \VV_\infty +\VV \omega \VV_\infty\big),\hspace{0.5cm} \text{ with } C>0,
\end{equation*}
which completes the proof. 
\begin{SL}\cite[App. I.1]{Grafakos}: 
Let $k:\Rd \times \Rd \to \mathbb{C}$ be a locally integrable function satisfying 
\begin{align}
&(i)\ \sup_{x\in \Rd} \int_{\Rd}|k(x,u)|\mathrm du \leq \alpha,\nonumber\\
&(ii)\ \sup_{u \in \Rd}\int_{\Rd}|k(x,u)|\mathrm dx \leq \alpha,\label{condi:Schur}
\end{align} where $\alpha>0$. Then, $(Kf)(x)=\int_{\Rd} k(x,u)f(u)\mathrm du$ is a bounded operator from $L^2(\Rd)$ to $L^2(\Rd)$ with operator norm $\VV K \VV_{2,2}\leq \alpha$.
\end{SL}
We start by determining the integral operator $K$ in \eqref{eq:p25}. To this end, consider $ \eta \in S(\Rd,\mathbb{C})$ such that $\widehat{\eta}(\omega)=1$, for all $ \omega\in B_{1}(0)$. Setting $\gamma(x):=R^d  \eta(Rx)$ yields $\gamma \in S(\Rd,\mathbb{C})$ and $\widehat{\gamma}(\omega)=\widehat{\eta}(\omega/R)$. Thus, $\widehat{\gamma}(\omega)=1$, for all $\omega \in B_{R}(0)$, and hence $\widehat{f}=\widehat{f}\cdot \widehat{\gamma}$, so that $f=f\ast \gamma$, for all $f \in L^2_R(\Rd)$. Next, we define the operator $A_{\gamma}:L^2(\Rd) \to L^2(\Rd)$, $A_{\gamma}f:=f\ast \gamma$, and note that $A_\gamma$ is well-defined, i.e., $A_\gamma f\in L^2(\Rd)$, for all $f \in L^2(\Rd)$, thanks to Young's inequality \cite[Theorem 1.2.12]{Grafakos} (since $f\in L^2(\Rd)$ and $\gamma \in S(\Rd,\mathbb{C})\subseteq L^1(\Rd)$). Moreover, $A_\gamma f=f$, for all $f\in L^2_R(\Rd)$. Setting $K:=F_{\tau,\omega}A_\gamma - A_\gamma$, we get $Kf=F_{\tau,\omega}A_\gamma f- A_\gamma f=F_{\tau,\omega}f-f$, for all $f\in L^2_R(\Rd)$, as desired. Furthermore, it follows from 
\begin{equation*}
\begin{split}
(F_{\tau,\omega} A_\gamma f)(x)&=e^{2\pi i \omega(x)}\int_{\Rd}  \gamma(x-\tau(x)-u)f(u)\mathrm du,\\
\end{split}
\end{equation*}
that the integral operator $K=F_{\tau,\omega} A_{\gamma}-A_{\gamma}$, i.e.,
$$(Kf)(x)=\int_{\Rd} k(x,u)f(u)\mathrm du,$$ has the kernel
\begin{equation}\label{int_kernel}
k(x,u):= e^{2\pi i \omega(x)}\gamma(x-\tau(x)-u)-\gamma(x-u).
\end{equation} 
Before we can apply Schur's Lemma to establish an upper bound on $\VV K \VV_{2,2}$, we  need to verify that $k$ in \eqref{int_kernel} is locally integrable, i.e., we need to show that for every compact set $S\subseteq \Rd\times \Rd$ we have $$\int_{S}|k(x,u)|\mathrm d(x,u)<\infty.$$ To this end, let $S\subseteq \Rd\times\Rd$ be a compact set. Next, choose compact sets $S_1,S_2 \subseteq \Rd$ such that $S\subseteq S_1\times S_2$. Thanks to $\gamma \in S(\Rd,\mathbb{C})$, $\tau \in C^1(\Rd,\Rd)$, and $\omega \in C(\Rd,\RR)$, all by assumption, the function $|k|:S_1\times S_2\to \mathbb{C}$ is continuous as a composition of continuous functions, and therefore also Lebesgue-measurable. We further have 
\begin{align}
&\int_{S_1} \int_{S_2} | k(x,u)|\mathrm dx   \mathrm du \leq \int_{S_1} \int_{\Rd} | k(x,u)|\mathrm dx  \mathrm du\nonumber\\
\leq&\int_{S_1} \int_{\Rd} | \gamma(x-\tau(x)-u)|\mathrm dx   \mathrm du + \int_{S_1} \int_{\Rd} | \gamma(x-u)|\mathrm dx   \mathrm du\nonumber \\ \leq&\ 2 \int_{S_1} \int_{\Rd} | \gamma(y)|\mathrm dy  \mathrm du + \int_{S_1} \int_{\Rd} | \gamma(y)|\mathrm dy \  \mathrm du\label{eq:p40}\\
=&\,3\mu_L(S_1)\VV \gamma \VV_1<\infty,\nonumber
\end{align}
where the first term in \eqref{eq:p40} follows by the change of variables $y=x-\tau(x) -u$, together with 
\begin{equation}\label{helpabcd}
\frac{\mathrm dy}{\mathrm dx}=|\hspace{-0.05cm}\det(E- (D \tau)(x))| \geq 1- d\VV D \tau \VV_\infty\geq 1/2,
\end{equation}
for all $x\in \Rd$. The arguments underlying \eqref{helpabcd} were already detailed at the end of Appendix \ref{Appendix}. It follows that $k$ is locally integrable owing to  
\begin{align}
\int_{S}|k(x,u)|\mathrm d(x,u)&\leq \int_{S_1\times S_2}|k(x,u)|\mathrm d(x,u)\nonumber\\
&=\int_{S_1} \int_{S_2} | k(x,u)|\mathrm dx  \mathrm du<\infty,\label{eq:www111}
\end{align}
where the first step in \eqref{eq:www111} follows from $S\subseteq S_1\times S_2$, the second step is thanks to the Fubini-Tonelli Theorem \cite[Theorem 14.2]{Fubini} noting that $|k|:S_1\times S_2\to \mathbb{C}$ is Lebesgue-measurable (as established above) and non-negative, and the last step is due to \eqref{eq:p40}. Next, we need to verify conditions (i) and (ii) in \eqref{condi:Schur} and determine the corresponding $\alpha>0$. In fact, we seek a specific constant $\alpha$  of the form  
\begin{equation}\label{eq:p43}
\alpha=C \big(R\VV \tau \VV_\infty +\VV \omega \VV_\infty\big), \hspace{0.5cm}\text{ with } C>0.
\end{equation}
This will be accomplished as follows: For $x,u \in \Rd$, we parametrize the integral kernel in \eqref{int_kernel} according to  
$h_{x,u}(t):=e^{2\pi i t \omega(x)}\gamma(x-t\tau(x)-u)-\gamma(x-u).$ A Taylor series expansion \cite[p. 411]{RudinComplex} of $h_{x,u}(t)$ w.r.t. the variable $t$  now yields
\begin{equation}\label{ea:ee2}
h_{x,u}(t)=\underbrace{h_{x,u}(0)}_{=0}+\int_{0}^t h'_{x,u}(\lambda)\mathrm d\lambda=\int_{0}^t h'_{x,u}(\lambda)\mathrm d\lambda,
\end{equation} 
for  $t \in \RR$, where $h'_{x,u}(t)=(\frac{\mathrm d}{\mathrm dt}h_{x,u})(t)$. Note that $h_{x,u}\in C^1(\RR,\mathbb{C})$ thanks to $\gamma \in S(\Rd,\mathbb{C})$. Setting $t=1$ in \eqref{ea:ee2} we get 
\begin{equation}\label{eq:p46}
|k(x,u)|=|h_{x,u}(1)|\leq \int_{0}^1|h'_{x,u}(\lambda)|\mathrm d\lambda,
\end{equation}
where
\begin{align}
h'_{x,u}(\lambda)&=-e^{2\pi i \lambda\omega(x)}\langle \nabla \gamma(x-\lambda\tau(x)-u),\tau(x)\rangle\nonumber\\
&+2\pi i \omega(x)e^{2\pi i \lambda \omega(x)}\gamma(x-\lambda \tau(x)-u),
\end{align}
for $\lambda \in [0,1]$. We further have
\begin{align}
|h'_{x,u}(\lambda)|&\leq \big| \big\langle \nabla \gamma(x-\lambda \tau(x)-u), \tau(x)\big\rangle \big|\nonumber\\
&+ | 2\pi \omega(x)\gamma(x-\lambda\tau(x)-u)|\nonumber\\
&\leq |\tau(x)| |\nabla \gamma(x-\lambda \tau(x)-u)|\nonumber\\
&+ 2\pi | \omega(x)| |\gamma(x-\lambda\tau(x)-u)|\label{eq:p42}.
\end{align}
Now, using $|\tau(x)|\leq \sup_{y\in \Rd}|\tau(y)|=\VV \tau \VV_{\infty}$ and  $|\omega(x)|\leq \sup_{y\in \Rd}|\omega(y)|=\VV \omega \VV_{\infty}$
in \eqref{eq:p42}, together with \eqref{eq:p46}, we get the upper bound 
\begin{align}
|k(x,u)&|\leq\VV \tau\VV_\infty\int_{0}^1 |\nabla \gamma(x-\lambda \tau(x)-u)|\mathrm d\lambda\nonumber\\
&+ 2\pi \VV \omega\VV_\infty\int_{0}^1 |\gamma(x-\lambda\tau(x)-u)|\mathrm d\lambda\label{eq:p44}.
\end{align}
Next, we integrate \eqref{eq:p44} w.r.t. $u$ to establish (i) in \eqref{condi:Schur}:
\begin{align}
\int_{\Rd}& |k(x,u)|\mathrm du\nonumber\\ 
\leq\,& \VV \tau\VV_\infty\int_{\Rd}\int_{0}^1 |\nabla \gamma(x-\lambda \tau(x)-u)|\mathrm d\lambda  \mathrm du\nonumber\\ +\,& 2\pi \VV \omega\VV_\infty \int_{\Rd}\int_{0}^1 |\gamma(x-\lambda\tau(x)-u)|\mathrm d\lambda  \mathrm du \nonumber\\
=\,& \VV \tau\VV_\infty\int_{0}^1 \int_{\Rd}|\nabla \gamma(x-\lambda \tau(x)-u)| \mathrm du  \mathrm d\lambda\nonumber\\ +\,& 2\pi \VV \omega\VV_\infty\int_{0}^1\int_{\Rd} |\gamma(x-\lambda\tau(x)-u)| \mathrm du  \mathrm  d\lambda\label{eq:p60} \\
=\,& \VV \tau\VV_\infty\int_{0}^1 \int_{\Rd}|\nabla \gamma(y)|\mathrm dy  \mathrm d\lambda \nonumber\\
+\,&2\pi \VV \omega\VV_\infty\int_{0}^1\int_{\Rd} |\gamma(y)| \mathrm dy  \mathrm d\lambda\nonumber\\
=\,&\VV \tau\VV_\infty\VV \nabla \gamma\VV_1 +2\pi \VV \omega\VV_\infty\VV \gamma\VV_1\label{eq:eee1},
\end{align}
where \eqref{eq:p60} follows by application of the Fubini-Tonelli Theorem \cite[Theorem 14.2]{Fubini} noting that the functions $(u,\lambda) \mapsto  |\nabla\gamma(x-\lambda\tau(x)-u)|$, $(u,\lambda)\in\Rd\times[0,1]$,  and $(u,\lambda) \mapsto |\gamma(x-\lambda\tau(x)-u)|$, $ (u,\lambda)\in\Rd\times[0,1]$, are both non-negative and continuous (and thus Lebesgue-measurable) as  compositions of continuous functions. Finally, using $\gamma=R^d\eta(R\cdot)$, and thus $\nabla \gamma=R^{d+1}\nabla \eta(R\cdot)$, $\VV \gamma \VV_1=\VV \eta\VV_1$, and $\VV \nabla \gamma \VV_1=R\VV \nabla \eta\VV_1$ in \eqref{eq:eee1} yields
\begin{align}
&\sup_{x\in \Rd} \int_{\Rd}|k(x,u)|\mathrm du \leq R\VV \tau\VV_\infty\VV \nabla \eta\VV_1 +2\pi \VV \omega\VV_\infty\VV \eta\VV_1\nonumber\\
&\leq \max\big\{\VV \nabla \eta\VV_1, 2\pi\VV \eta\VV_1\big\}\big(R\VV \tau\VV_\infty+\VV \omega\VV_\infty \big)\label{eq:p50},
\end{align}
which establishes an upper bound of the form (i) in \eqref{condi:Schur} that exhibits the desired structure for $\alpha$. Condition (ii) in \eqref{condi:Schur} is established similarly by integrating \eqref{eq:p44} w.r.t. $x$ according to 
 \begin{align}
\int_{\Rd}& |k(x,u)|\mathrm dx\nonumber\\
\leq\,& \VV \tau\VV_\infty\int_{\Rd}\int_{0}^1 |\nabla \gamma(x-\lambda \tau(x)-u)|\mathrm d\lambda  \mathrm dx\nonumber\\+\,& 2\pi \VV \omega\VV_\infty \int_{\Rd}\int_{0}^1 |\gamma(x-\lambda\tau(x)-u)|\mathrm d\lambda  \mathrm dx\nonumber\\
=\,& \VV \tau\VV_\infty\int_{0}^1 \int_{\Rd}|\nabla \gamma(x-\lambda \tau(x)-u)| \mathrm dx  \mathrm d\lambda\nonumber\\+\,& 2\pi \VV \omega\VV_\infty\int_{0}^1\int_{\Rd} |\gamma(x-\lambda\tau(x)-u)| \mathrm dx \mathrm d\lambda\label{eq:p47}\\
\leq\,& 2\,\VV \tau\VV_\infty\int_{0}^1 \int_{\Rd}|\nabla \gamma(y)| \mathrm dy  \mathrm d\lambda \nonumber\\+\,&4\pi \VV \omega\VV_\infty\int_{0}^1\int_{\Rd} |\gamma(y)| \mathrm dy  \mathrm d\lambda\label{eq:eM}\\
=\,&2\,\VV \tau\VV_\infty\VV \nabla \gamma\VV_1 4\pi \VV \omega\VV_\infty\VV \gamma\VV_1\nonumber\\\leq\,& \max\big\{2\VV \nabla \eta\VV_1, 4\pi\VV \eta\VV_1\big\}\big(R\VV \tau\VV_\infty+\VV \omega\VV_\infty \big)\label{eq:p55}.
\end{align}
Here, again, \eqref{eq:p47} follows by application of the Fubini-Tonelli Theorem \cite[Theorem 14.2]{Fubini} noting that the functions $(x,\lambda) \mapsto  |\nabla\gamma(x-\lambda\tau(x)-u)|$, $ (x,\lambda)\in\Rd\times[0,1]$,  and $(x,\lambda) \mapsto |\gamma(x-\lambda\tau(x)-u)|$, $ (x,\lambda)\in\Rd\times[0,1]$, are both non-negative and continuous (and thus Lebesgue-measurable) as compositions of continuous functions. The inequality \eqref{eq:eM} follows from a change of variables argument similar to the one in \eqref{eq:p40} and \eqref{helpabcd}. Combining \eqref{eq:p50} and \eqref{eq:p55}, we finally get \eqref{eq:p43} with 
\begin{equation}\label{stab_const}
C:= \max\big\{2\VV \nabla \eta\VV_1, 4\pi\VV \eta\VV_1\big\}.
\end{equation}
This completes the proof.
\end{proof}

\section*{Acknowledgments}
The authors would like to thank P. Grohs, S. Mallat, R. Alaifari, M. Tschannen, and G. Kutyniok for helpful discussions and comments on the paper.

\newpage
\vspace{-1cm}
\begin{IEEEbiographynophoto}{Thomas Wiatowski} was born in Strzelce Opolskie, Poland, on December 20, 1987, and received the BSc and MSc degrees, both in Mathematics, from the Technical University of Munich, Germany, in 2010 and 2012, respectively. In 2012 he was a  researcher with the Institute of Computational Biology at the Helmholtz Zentrum in Munich, Germany. He joined ETH Zurich in 2013, where he graduated  with the Dr.\ sc.\ degree in 2017. His research interests are in deep machine learning, mathematical signal processing, and applied harmonic analysis.
\end{IEEEbiographynophoto}

\vspace{-9cm}
\begin{IEEEbiographynophoto}{Helmut B\"olcskei} was born in M\"odling, Austria, on May 29, 1970, and received the Dipl.-Ing.\ and Dr.\ techn.\ degrees in electrical engineering from Vienna University of Technology, Vienna, Austria, in 1994 and 1997, respectively. In 1998 he was with Vienna University of Technology. From 1999 to 2001 he was a postdoctoral researcher in the Information Systems Laboratory, Department of Electrical Engineering, and in the Department of Statistics, Stanford University, Stanford, CA. He was in the founding team of Iospan Wireless Inc., a Silicon Valley-based startup company (acquired by Intel Corporation in 2002) specialized in multiple-input multiple-output (MIMO) wireless systems for high-speed Internet access, and was a co-founder of Celestrius AG, Zurich, Switzerland. From 2001 to 2002 he was an Assistant Professor of Electrical Engineering at the University of Illinois at Urbana-Champaign. He has been with ETH Zurich since 2002, where he is a Professor of Electrical Engineering. He was a visiting researcher at Philips Research Laboratories Eindhoven, The Netherlands, ENST Paris, France, and the Heinrich Hertz Institute Berlin, Germany. His research interests are in information theory, mathematical signal processing, machine learning, and statistics.

He received the 2001 IEEE Signal Processing Society Young Author Best Paper Award, the 2006 IEEE Communications Society Leonard G. Abraham Best Paper Award, the 2010 Vodafone Innovations Award, the ETH ``Golden Owl'' Teaching Award, is a Fellow of the IEEE, a 2011 EURASIP Fellow, was a Distinguished Lecturer (2013-2014) of the IEEE Information Theory Society, an Erwin Schr\"odinger Fellow (1999-2001) of the Austrian National Science Foundation (FWF), was included in the 2014 Thomson Reuters List of Highly Cited Researchers in Computer Science, and is the 2016 Padovani Lecturer of the IEEE Information Theory Society. He served as an associate editor of the IEEE Transactions on Information Theory, the IEEE Transactions on Signal Processing, the IEEE Transactions on Wireless Communications, and the EURASIP Journal on Applied Signal Processing. He was editor-in-chief of the IEEE Transactions on Information Theory during the period 2010-2013. He served on the editorial board of the IEEE Signal Processing Magazine  and is currently on the editorial boards of ``Foundations and Trends in Networking'' and ``Foundations and Trends in Communications and Information Theory''.  He was TPC co-chair of the 2008 IEEE International Symposium on Information Theory and the 2016 IEEE Information Theory Workshop and serves on the Board of Governors of the IEEE Information Theory Society. He has been a delegate of the president of ETH Zurich for faculty appointments since 2008.
\end{IEEEbiographynophoto}

\end{document}